\documentclass[10pt]{article}
\usepackage{authblk}
\usepackage{blindtext}

\usepackage[utf8]{inputenc} 
\usepackage[T1]{fontenc}    
\usepackage{hyperref}       
\usepackage{url}            
\usepackage{booktabs}       
\usepackage{amsfonts}       
\usepackage{nicefrac}       

\usepackage{verbatim}

\usepackage{multirow}
\usepackage{amstext}
\usepackage{amsmath}
\usepackage{amsthm}
\usepackage{amssymb}
\usepackage{bm}
\usepackage{wrapfig}
\usepackage{amsfonts, amsmath}
\usepackage[utf8]{inputenc}
\usepackage{graphicx}
\usepackage{bbm, comment}
\usepackage{color}
\usepackage{float}
\usepackage{wrapfig}
\usepackage{url}
\usepackage{MnSymbol,wasysym,marvosym,ifsym}
\usepackage[T1]{fontenc}
\usepackage[inline]{enumitem}

\usepackage{autobreak}
\usepackage{microtype}      
\usepackage{booktabs} 
\usepackage[ruled]{algorithm2e} 
\usepackage[tight]{subfigure}
\usepackage{threeparttable}
\usepackage[utf8]{inputenc}
\usepackage{graphicx} 
\usepackage{float} 
\usepackage{pgfplots}
\usepackage{caption}
\usepackage{graphicx}

\usepackage{geometry}
\usepackage[numbers]{natbib}

\newtheorem*{theorem*}{Theorem}
\newtheorem{theorem}{Theorem}[section]
\newtheorem{lemma}[theorem]{Lemma}

\newtheorem{claim}[theorem]{Claim}

\newtheorem{corollary}[theorem]{Corollary}

\newtheorem{definition}[theorem]{Definition}
\newtheorem{observation}[theorem]{Observation}

\newcommand{\E}{\mathrm{E}}
\newcommand{\last}{x}
\newcommand{\round}{\textsc{rd}}
\newcommand{\optmset}{\mathcal{X}(n)}
\newcommand{\argmax}{\mathop{\arg\max}}

\newcommand{\his}{H}
\newcommand{\hisset}{\mathcal{\his}}
\newcommand{\hisins}{h}
\newcommand{\bel}{B}
\newcommand{\belset}{\mathcal{\bel}}

\newcommand{\harmo}{\mathbb{H}}
\newcommand{\pluscount}{\mathrm{COUNT}}
\newcommand{\state}{S}

\newcommand{\asyf}{Z}
\newcommand{\betadis}{\mathcal{B}e}
\newcommand{\betadens}[1]{Beta(#1;\alpha,\beta)}

\newcommand{\mytextn}{\text{final round}}
\newcommand{\mytexta}{\text{penultimate round (at point $L_{n-2}$)}}
\newcommand{\mytextb}{\text{penultimate round (at point $U_{n-2}$)}}
\newcommand{\mytextm}{\text{penultimate round (between $L_{n-2}$ and $U_{n-2}$)}}

\newcommand{\mytextnbig}{\textbf{Final Round}}
\newcommand{\mytextabig}{\textbf{Penultimate Round (at point $L_{n-2}$)}}
\newcommand{\mytextbbig}{\textbf{Penultimate Round (at point $U_{n-2}$)}}
\newcommand{\mytextmbig}{\textbf{Penultimate Round (between $L_{n-2}$ and $U_{n-2}$)}}

\usepackage{xcolor}

\begin{document}

\title{BONUS! Maximizing Surprise}
\author[1]{Zhihuan Huang}
\author[1]{Yuqing Kong\thanks{Supported by National Natural Science Foundation of China award number~62002001}$^,$} 
\author[2]{Tracy Xiao Liu\thanks{Supported by the National Key Research and Development Program of China award number~2018YFB1004503}$^,$}
\author[3]{Grant Schoenebeck\thanks{Supported by (United States) National Science Foundation award number~2007256}$^,$}
\author[1]{Shengwei Xu}

\affil[1]{Department of Computer Science, Peking University}
\affil[1]{Center on Frontiers of Computing Studies, Peking University}
\affil[2]{School of Economics and Management, Tsinghua University}
\affil[3]{School of Information, University of Michigan}
\affil[1]{\texttt{\{zhihuan.huang, shengwei.xu, yuqing.kong\}@pku.edu.cn}}
\affil[2]{\texttt{liuxiao@sem.tsinghua.edu.cn}}
\affil[3]{\texttt{schoeneb@umich.edu}}
\date{}
\maketitle

\begin{abstract}
Multi-round competitions often double or triple the points awarded in the final round, calling it a bonus, to maximize spectators' excitement. In a two-player competition with $n$ rounds, we aim to derive the optimal bonus size to maximize the audience's overall expected surprise (as defined in \cite{ely2015suspense}). We model the audience's prior belief over the two players' ability levels as a beta distribution.  Using a novel analysis that clarifies and simplifies the computation, we find that the optimal bonus depends greatly upon the prior belief and obtain solutions of various forms for both the case of a finite number of rounds and the asymptotic case. In an interesting special case, we show that the optimal bonus approximately and asymptotically equals to the ``expected lead'', the number of points the weaker player will need to come back in expectation. Moreover, we observe that priors with a higher skewness lead to a higher optimal bonus size, and in the symmetric case, priors with a higher uncertainty also lead to a higher optimal bonus size. This matches our intuition since a highly asymmetric prior leads to a high ``expected lead'', and a highly uncertain symmetric prior often leads to a lopsided game, which again benefits from a larger bonus.  

\end{abstract}

\thispagestyle{empty}
\newpage
\setcounter{page}{1}
\section{Introduction}

People love watching competitions.  Who will snatch victory?  Who will be vanquished? Thus, the business of sports\footnote{FIFA generated more than 4.6 billion USD revenue in 2018, mainly on 2018 FIFA World Cup, according to https://www.investopedia.com/articles/investing/070915/how-does-fifa-make-money.asp}, e-sports\footnote{ Douyu, one of the major e-sports live streaming platform in China generated about 6.6 billion CNY of revenue (around 1 billion USD) in 2019, according to https://www.statista.com/statistics/1222790/china-douyu-live-streaming-revenue/}, and live streaming platforms\footnote{The worldwide video streaming market size reached 50.11 billion USD in 2020, according to  https://www.grandviewresearch.com/industry-analysis/video-streaming-market} create a huge amount of revenue. 

When people watch a game, their belief for who would win in the end changes over the duration of the game. Intuitively, the \emph{surprise}, measures how much audience's beliefs change over time \cite{ely2015suspense}. An important question for both theorists and practitioners is how to design winner selection schemes that maximize the amount of surprise in a competition and, consequently, improve the entertainment utility of a competition and increase its potential revenue.

In practice, many games use point systems to determine the winner. For some games, the point value remains static throughout the game, such as soccer, cricket, and tennis. However, in other games, the point value in the final round is higher than in others. 
Mind King, a very popular two-player quiz game on WeChat,\footnote{over 1 million daily users according to an author interview with Tencent staff in July 2021.} has 5 rounds. For each round, each player receives points depending on the correctness and speed of their answer.  The final round doubles the points.   Since 2015, IndyCar racing has doubled the points for the final race of the season.\footnote{According to http://www.champcarstats.com/points.htm} The Diamond League,\footnote{According to Wikipedia: https://en.wikipedia.org/wiki/Diamond\_League\#Scoring\_system} a track and field league, from 2010-2016, determined the Diamond Race winner by a point system over a season of 7 meets, and the points for the final tournament are doubled.\footnote{The scoring system was substantially changed in 2017 to, in particular, include a final restricted to the top-scoring athletes.}
Another very popular example is the quidditch matches in Harry Potter\footnote{Harry Potter is a very popular fantasy novels series that have sold out more than 500 million copies according to https://www.wizardingworld.com/news/500-million-harry-potter-books-have-now-been-sold-worldwide}.  The game concludes when the golden snitch, which is worth 15 times a normal goal, is captured by one team.  While this point structure makes sense for the plot of the books (the title character often catches the snitch), most of the action is ancillary to the games' outcome.  Tellingly, ``muggle'' quidditch is now a real-life game, however, the golden snitch is only worth 3 times as much as a normal goal. 

\paragraph{Key Question} 
In a simple model, we study the effect of the final round's point value on the overall surprise and what point value maximizes the surprise. In our setting, there are two teams and a fixed number of rounds. The winner in each of the first $n-1$ rounds is awarded 1 point and the winner in the final round can possibly receive more points, e.g., double or triple the points earned in previous rounds.  The team which accumulates the most points wins.  

\paragraph{Types of Prior Beliefs} To measure surprise, we need to model audience belief. The audience's belief of who will win (and thus the surprise) depends on the audience's prior belief about the two contestants' chances of winning in each round. Note that this belief may be updated as the match progresses.   Intuitively, in a game where audiences believe that the two players' ability is highly asymmetric, e.g., a strong vs. a weak player, we should set up a high bonus score, otherwise, the weaker player will likely be eliminated before the game's conclusion. In contrast, if the two players are perfectly evenly matched and the total number of rounds is odd, we will show the optimal bonus value is 1, the same point value as in the previous rounds.  Below, we consider three special cases for prior beliefs and then introduce the general case.  

The first case is that the audience has a fixed and unchanging belief about the chance that each competitor wins for each round.  We call this the \emph{certain} case.  The size of the optimal bonus depends on the belief in the difference between two contestants' ability levels. 

The second case is where the audience has no prior knowledge about the two players' abilities.\footnote{It is worthwhile to note that this is different from the case where the prior belief is that both players have an equal chance to win. The reason is that no belief updating exists in the prior case while belief updating exists in the later one. } We call this the \emph{uniform} case. Specifically, we model the audience's prior over the two players as a uniform distribution and derive the optimal bonus size accordingly. A slight generalization of the uniform case is the \emph{symmetric} case when the audience has a certain amount of prior knowledge but no prior knowledge favoring one contestant over another. In this case, we model the audience's prior as a symmetric Beta distribution. In the general case, we model the prior as a general Beta distribution, $\betadis(\alpha,\beta)$\footnote{When Alice and Bob have played before and Alice wins $n_A$ rounds and Bob wins $n_B$, we can use $\betadis(n_A+1,n_B+1)$ to model the prior belief. In general, we allow $\alpha,\beta$ to be non-integers.  }, a family of continuous probability distributions on the interval [0, 1] parameterized by two positive shape parameters, $\alpha,\beta\geq 1$. The figures above Table~\ref{table:basic} illustrate the above cases. 

\paragraph{Techniques}  Our analysis is built on two insights.  First, for Beta prior beliefs, we show that for $i \leq n - 1$, the ratio of the expected surprise in round $i$ and $i -1$ is a constant that only depends on $i$ and the prior.  This allows us to reduce the entire analysis to the trade-off between the final round's surprise amount and the penultimate round's surprise amount. Second, we find that the final round's surprise amount increases when the bonus increases while the penultimate round's surprise consists of two parts where one part increases with the bonus size and the other part decreases with the bonus size. Optimally trading off these terms yields our results.  Third, we show how to optimize this trade-off.

\begin{table}[!ht]
	\center
	\begin{minipage}{\linewidth}
	\begin{tabular}{cccc}\centering
		&\begin{minipage}{.25\linewidth}\includegraphics[width=1\linewidth]{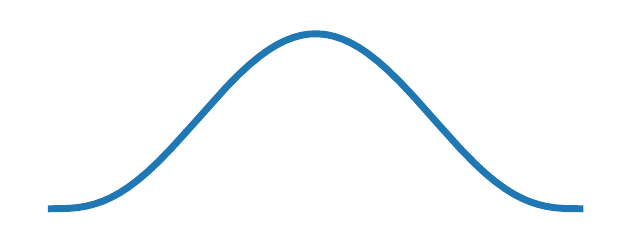}\end{minipage} & \begin{minipage}{.25\linewidth}\includegraphics[width=1\linewidth]{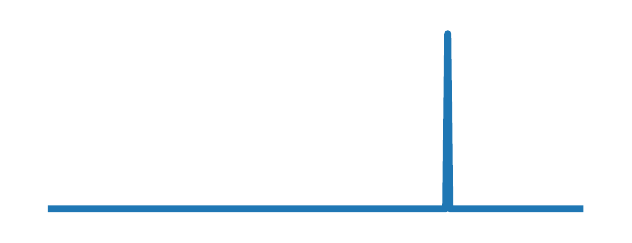}\end{minipage} & \begin{minipage}{.25\linewidth}\includegraphics[width=1\linewidth]{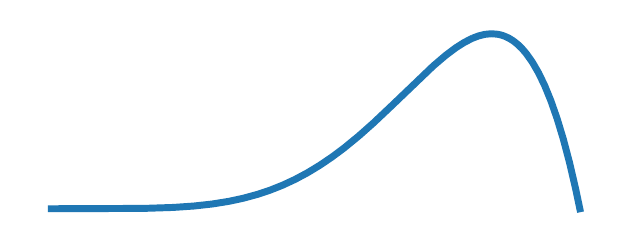}\end{minipage}\\
		\hline
		 & Symmetric  & Certain  & General\\
		 & $\alpha=\beta$ & $\alpha=\lambda p,\beta=\lambda (1-p),\lambda\rightarrow\infty$ & \\
		\hline
		Finite  & \begin{minipage}{.25\linewidth}\[\round(\frac{n-1}{2\alpha\harmo-\frac{n-1}{n+2\alpha-1}})\]\end{minipage} & \begin{minipage}{.25\linewidth}\[\round(\text{Solution of $F(x)=0$})\]\end{minipage} & \begin{minipage}{.25\linewidth}\[O(n)\text{ algorithm}\]\end{minipage}\\
		  & \begin{minipage}{.25\linewidth}\end{minipage} & \begin{minipage}{.25\linewidth}\[\approx n \frac{\alpha-\beta}{\alpha+\beta}\]\end{minipage} & \begin{minipage}{.25\linewidth}\end{minipage}\\
		\specialrule{0em}{2pt}{2pt}
		Asymptotic & \begin{minipage}{.15\linewidth}\[n\frac{1}{2\alpha \harmo-1} \approx \frac{\frac{n}{2\alpha}}{\ln (\frac{n}{2\alpha})}\]\end{minipage} & \begin{minipage}{.25\linewidth}\[n\frac{(\alpha-\beta)\harmo+1}{(\alpha+\beta)\harmo-1}\approx n \frac{\alpha-\beta}{\alpha+\beta}\]\end{minipage} & \begin{minipage}{.25\linewidth}\[n*\text{Solution of $G(\mu)=0$}\]\end{minipage} \\
		\specialrule{0em}{2pt}{2pt}
		\hline

	\end{tabular}
	\caption{\textbf{Optimal bonus}}
	\begin{minipage}{\linewidth}
	\footnotesize
	\[\text{Without loss of generality, we assume $\alpha\geq\beta$ which implies that $p\geq \frac12$.}\]
	\[\round(x):=\text{the nearest integer to $x$ that has the same parity as the number of rounds $n$.}\footnote{When there is a tie, we pick the smaller one.}\]
	\[\harmo=\harmo_{\alpha+\beta}(n-1):=\sum_{i=1}^{n-1}\frac{1}{i+\alpha+\beta-1},F(x):=(2np-n-(x-1))p^{x-1}+(n-2np-(x-1))(1-p)^{x-1}, x\in [1,n+1)\footnote{When $p=\frac12$ or $n\leq \frac{1}{(\frac{1}{2}-p)\ln(\frac{1-p}{p})}$, $F(x)$ only has a trivial solution $x=1$ and the optimal bonus is $\round(1)$, otherwise, the optimal bonus is $\round(\Tilde{x})$ where $\Tilde{x}$ is the unique non-trivial solution.}\]
	\[G(\mu):=(1+\mu)^{\alpha-\beta}\left(\frac{(\alpha-\beta)\harmo+1}{(\alpha+\beta)\harmo-1}-\mu\right)+(1-\mu)^{\alpha-\beta}\left(\frac{(-\alpha+\beta)\harmo+1}{(\alpha+\beta)\harmo-1}-\mu\right),\mu\in(0,1)\] 
	\end{minipage}
	\label{table:basic}
	\end{minipage}
\end{table}

\begin{figure}[!ht]\centering 
  \includegraphics[width=.95\linewidth]{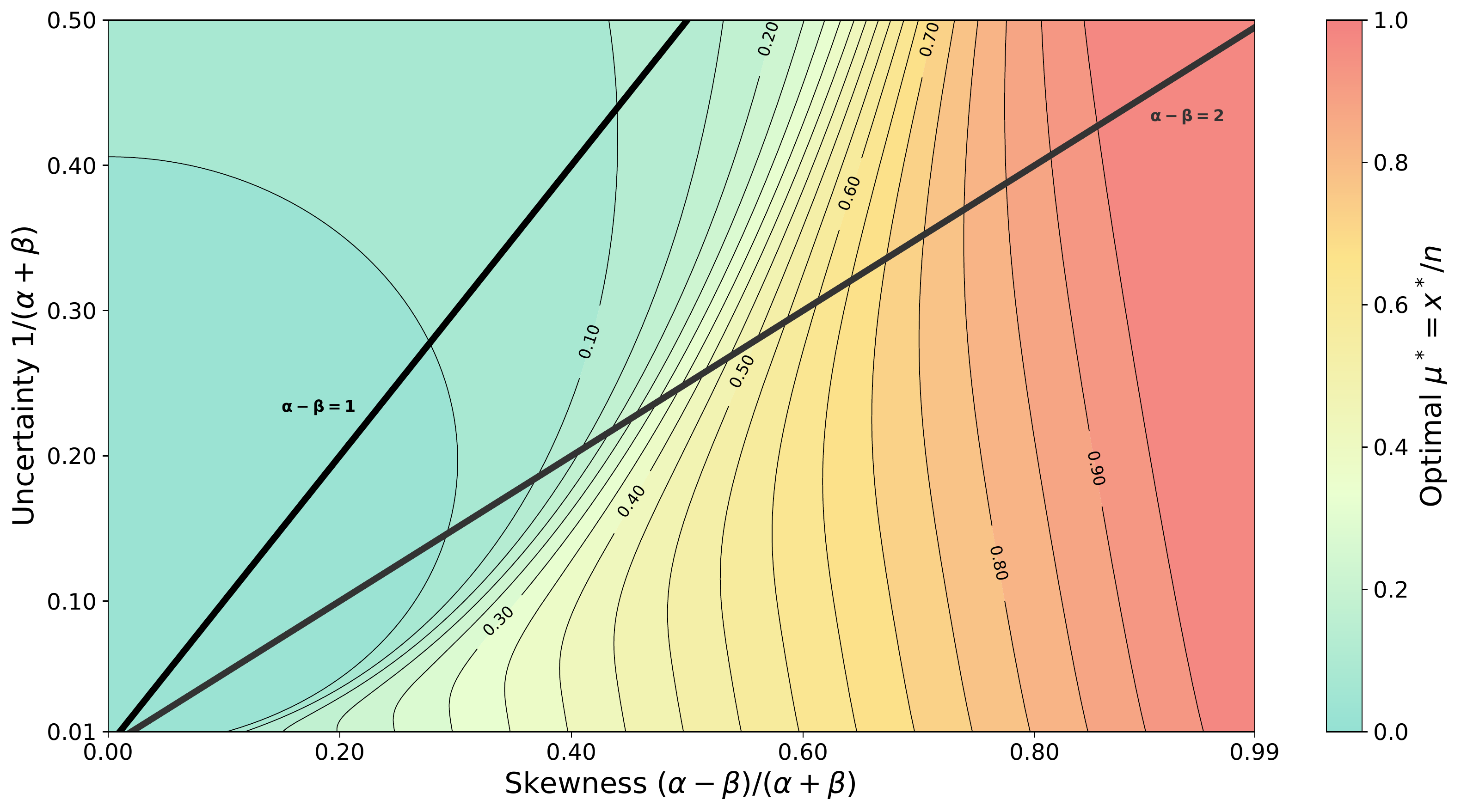}
  \caption{\textbf{Optimal bonus in asymptotic case}: for all $\alpha,\beta\geq 1$, we use $\frac{\alpha-\beta}{\alpha+\beta}$ to measure the skewness of the prior and $\frac{1}{\alpha+\beta}$ to measure the uncertainty of the prior. The figure shows that when $n$ is sufficiently large (we use $n=10000$ here), the relationship between the optimal bonus size and skewness/uncertainty.}
  \label{fig:infinite}
\end{figure}

\subsection{Results} 
Table~\ref{table:basic} shows the optimal bonus size in each case of the prior---symmetric, certain, and general---for both any finite number of rounds and the asymptotic value as the number of rounds increases.  Note that we have closed-form formulae for the symmetric case and the asymptotic certain case.

We obtain the following interesting insights: 

\paragraph{Insight 1. In the certain case, more uneven match-ups lead to a larger optimal bonus} When the match is more uneven, i.e., when the prior has a higher skewness $\frac{\alpha-\beta}{\alpha+\beta}$, the optimal bonus is larger. Interestingly, in the certain case, we find that the optimal bonus is around $\frac{\alpha-\beta}{\alpha+\beta}n$ which is the expected amount of points the weaker player needs to come back, which we call the ``expected lead''. For example, we let Marvel superheroes compete in multi-rounds. When Thor fights Hawkeye, there should be a large bonus than when Thor fights Iron Man.

\paragraph{Insight 2. In the symmetric case, more uncertainty leads to a higher optimal bonus.} In the symmetric case when the prior has more uncertainty, i.e.,  when $\frac{1}{\alpha+\beta}=\frac{1}{2\alpha}$ is larger, the optimal bonus is larger. In particular, the uniform case has a higher bonus than the certain case. For example, if Thor fights Superman, no one really knows the relative abilities of the heroes, since they come from different worlds. There should be a larger bonus than if Thor fights Iron Man. Since Thor has fought Iron Man before and it is known that they will be a good match-up. Note that in general, the optimal bonus size is not monotone in the uncertainty while holding the skewness fixed. 

Figure~\ref{fig:infinite} shows the optimal bonus's for the asymptotic case.   For small $\alpha-\beta$, it is similar to the symmetric case (the $y$-axis); and, analogously, for large $\alpha-\beta$, it is similar to the certain case  (the $x$-axis). Thus, we can use the special cases' results to approximate the general case. 

\subsection{Related Work}

\citet{ely2015suspense} provide a clear definition of surprise amount. Starting from \citet{ely2015suspense}, a growing literature examines the relationship between the surprise and the perceived quality in different games and presents empirical support that audiences' perceived quality of a game is partly determined by the amount of surprise generated in this process, such as tennis~\cite{bizzozero2016importance}, soccer~\cite{buraimo2020unscripted,LUCAS201758}, and rugby~\cite{scarf2019outcome}.  \citet{ely2015suspense} additionally study what number of rounds maximizes surprise in our certain case.  In contrast,  we study how to derive the optimal number of points for the final round and consider a more general prior model.

In addition to the surprise-related work, there has been a growing number of theoretical and empirical results studying how game rules affect different properties of the game, primarily fairness. For example, \citet{brams2018making} study how to make penalty shootouts rule fairer. \citet{braverman2008mafia} study how to make a popular party game, mafia, fair by tuning the number of different characters in the game. \citet{percy2015mathematical} shows that the change of badminton scoring system in 2006, in an attempt to make the game faster paced, did not adversely affect the fairness. 

One conception of fairness corresponds to the ``better'' team winning.  However, in some cases, like the certain case, this makes for a very dull game as there will be no surprise possible.  Thus, the goal of optimizing surprise and making the ``better'' team consistently win are sometimes in tension.

In addition to fairness, \citet{percy2015mathematical} also studies the influence of changes in the badminton scoring system implemented in 2006 on the entertainment value which is related to the number of rallies. However, in contrast to our work,  \citet{percy2015mathematical} focuses on comparing two scoring rules rather than optimization and does not formally model the entertainment value. \citet{kovacs2009effect} studies the effect of changes in the volleyball scoring system and empirically shows that the change may make the length of the matches more predictable. 

Results from \citet{mossel2014majority} based on Fourier analysis find that if one team wins each round with probability $p > 1/2$ then they win the match with probability at least $p$.  Moreover, the only way for the winning probability to be exactly $p$ is for the match to be decided by the outcome of one round. (They were studying opinion aggregating on social networks and this result applies to all deterministic symmetric and monotone Boolean functions with i.i.d. inputs.) This directly relates to the present work because, ideally, each team would start with a prior probability of winning close to 1/2.  This result says that the dictator function  (which can be enacted in our setting with a bonus of size $n$) yields a prior as close to 1/2 as possible.  Our results in the certain case can be seen as trading off the two goals of having a prior probability close to 1/2 but also revealing information more slowly over time than a dictator function.  

\section{Problem Statement}

We consider an $n$ round competition between two players, Alice and Bob. In round $i$, a task is assigned to the players and the winner receives points. After the $n$th round, the player with the higher accumulated score wins the competition. In our setting, we assume that each of the first $n-1$ rounds is worth one point. However, the point value $x$ for the final round might be different, and we call $x$ the \emph{bonus}.  Setting the final round bonus is a special, but interesting case to study.   Without loss of generality, we consider $x$ to be an integer where $0\leq x\leq n$.\footnote{Mathematically, $x>n$ is equivalent to $x=n$ where only the final round matters. Any non-integer $x$ is equivalent to an integer. For example, when $n=4$, $x=1.3$ is equivalent to $x=2$.}  To ensure there is no tie, we additionally require  $x$ has the same parity as $n$.  Let $\mathcal{X}(n)$ represent our considered range for $x$, that is

\[
\mathcal{X}(n) = \begin{cases}
\{1,3,5, \ldots, n\}\text{ , when n is odd}\\
\{0,2,4, \ldots, n\}\text{ , when n is even}
\end{cases} .
\]

\subsection{Optimization Goal}

To introduce our optimization goal formally, we first introduce the concept of belief curve. 

\paragraph{Belief Curve}The audience's belief curve is a sequence of random variables \[\belset:=(\bel_0,\bel_1,\ldots,\bel_n)\] where $\bel_i,i\in[n]$ is the audience's belief for the probability that Alice wins the whole competition after round $i$. $\bel_0$ is the initial belief. $\bel_n$ is either zero or one since the final outcome must be revealed in the end. 

We define $O\in\{0,1\}$ as the outcome of the whole competition, that is \[O=\begin{cases}0 & \text{Alice loses the whole competition} \\ 1 & \text{Alice wins the whole competition}\end{cases},\] then we define random variables
\begin{align*}
    \his_i&=\begin{cases}- & \text{Alice loses i-th round} \\ + & \text{Alice wins i-th round}\end{cases}\\
    \hisset^{(i)}&=(\his_1,\his_2,\ldots,\his_{i})
\end{align*}
i.e., $\hisset^{(i)}$ consists of the history of the first $i$ rounds, we call $\hisset^{(i)}$ an \emph{i-history}.

Next, we define random variables
\begin{align*}
    \bel_i&=\Pr[O=1|\hisset^{(i)}]\\
    \belset^{(i)}&=(\bel_0,\bel_1,\ldots,\bel_i)
\end{align*}
i.e., $\bel_i$ is the conditional probability that Alice wins the whole competition.  Additionally, we use $\belset$ as a shorthand for $\belset^{(n)}$. Note that $\belset$ is a doob martingale~\cite{1940Regularity} thus $\forall i, \E[\bel_{i+1}|\hisset^{(i)}]=\bel_i$. 

\begin{definition}[Surprise] \cite{ely2015suspense}
Given the belief curve $\belset$, we define $\Delta_\belset^i:=|\bel_i-\bel_{i-1}|$ as the \emph{amount of surprise generated by round $i$}. 
We define the \emph{overall surprise} for a given belief curve to be \[\Delta_\belset:= \sum_i \Delta_\belset^i. \]
\end{definition}

\paragraph{Maximizing the Overall Expected Surprise}

We aim to compute the bonus size $x$ which, in expectation, maximizes the audience's overall surprise.  That is, we aim to find $x$ to maximize the overall surprise \[ \argmax_{x\in \optmset} \E[\Delta_\belset(x)] \] where $\Delta_\belset(x)$ is the overall surprise when the bonus round's point value is equal to $x$.  

\begin{figure}[!ht]\centering
  \includegraphics[width=.8\linewidth]{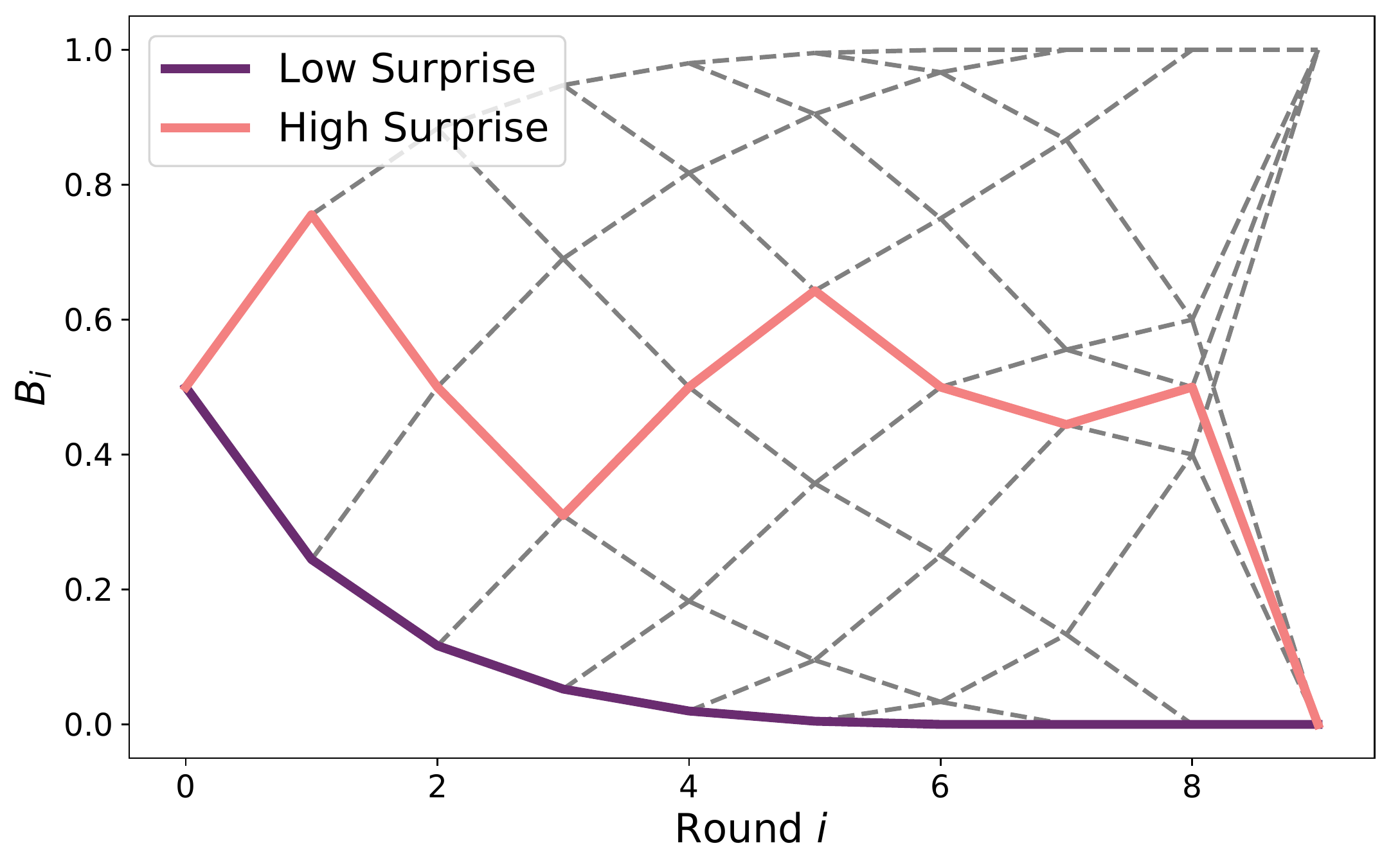}
  \caption{\textbf{Belief curves with low/high overall surprise}}
  \label{fig:two_example_curves}
\end{figure}

\subsection{Model of Prior Belief}

We introduce a natural model for the audience's prior. We will maximize the overall expected surprise in this model. 

We assume that each player's winning probability across rounds is constant, i.e., Alice wins with the \emph{same} probability $p$ in each round. Moreover, we assume the outcomes of these tasks are independent. 

\paragraph{Prior over $p$} We use Beta distribution $\betadis(\alpha,\beta)$ to model the audience's prior about Alice's winning probability $p$ in each round. The family of Beta distributions is sufficiently rich to cover a variety of important scenarios including the uniform case ($\alpha=\beta=1$), the symmetric case ($\alpha=\beta$) and the certain cases ($\alpha=\lambda p,\beta = \lambda(1-p), \lambda\rightarrow \infty$).  A key property of the Beta distribution is that, if $p$ is drawn from a Beta distribution and then we see the outcome of a coin which lands heads with probability $p$,  the posterior of $p$ after observing a coin flip is still a Beta distribution.
\begin{claim}[Beta's posterior is Beta]
If the prior $p$ follows $\betadis(\alpha ,\beta)$, then conditioning on Alice winning the first round, the posterior distribution over p is $\betadis(\alpha+1,\beta)$, and conditioning on Alice losing the first round, the posterior distribution over p is $\betadis(\alpha,\beta+1)$.
\label{cla:betaproperty}
\end{claim} Let $p|\hisset^{(i)}$ be a random variable which follows the posterior distribution of $p$ conditioning on i-history $\hisset^{(i)}$. For all $i\leq n-1$, we define an induced random variable $\state_i:=\pluscount(\hisset^{(i)})$ as the number of rounds Alice wins in the first $i$ rounds, called the \emph{state} after $i$ rounds. $p|\state_i$ is a random variable which follows the posterior distribution of $p$ conditioning on state $\state_i$. The property of Beta distribution (Claim~\ref{cla:betaproperty}) directly induces the following claim.

\begin{claim}[Order Independence]\label{cla:beta}
For all $i\leq n-1$, for all $\hisins\in\{+,-\}^i$, $p|(\hisset^{(i)}=\hisins)$ follows distribution $\betadis(\alpha+\pluscount(h), \beta+i-\pluscount(h))$.
\end{claim}

This immediately implies the following corollary.  

\begin{corollary}
[State Dependence]\label{cor:beta}
For all $i\leq n-1$, for all $\hisins\in\{+,-\}^i$, $p|(\hisset^{(i)}=\hisins)$ follows the same distribution as $p|(\state_i = \pluscount(\hisins))$ and $\Pr[O=1|\hisset^{(i)}=\hisins]=\Pr[O=1|\state_i=\pluscount(\hisins)]$.
\end{corollary}

That is, the posterior distribution of $p$ or $O$ only depends on the state, i.e., the number of rounds Alice wins and the order does not matter. For example, history $++-$ and history $-++$ induce the same posterior. 

The above properties make our prior model tractable. In this model, the optimal bonus size $x^*$ depends on $\alpha,\beta,n$ thus is denoted by $x^*(\alpha,\beta,n)$. 

\subsection{Method Overview}

\paragraph{Technical Challenge} The key technical challenge is that we do not have a clean format for the belief value across all rounds (especially early rounds). For example, for the asymmetric case, it is even difficult to represent the initial belief $B_0$ for different $x$.  A naive way to compute all belief values is using backward induction.  We can then enumerate all possible $x$ to find the optimal one. However, this method has $O(n^3)$ time complexity.  

To overcome this challenge, we utilize the properties of the Beta distributions. First, we show that we only need to analyze the belief values in the final two rounds and choosing $x$ becomes a trade-off between the final and penultimate rounds\footnote{When $x$ increases, the final round generates more surprise while the penultimate round generates less surprise.}; Second, we study a few important special cases (asymptotic, symmetric, certain) which can further simplify the final two rounds' analysis significantly. Third, instead of actually computing $\E[\Delta_\belset(x)]$, we only analyze how $\E[\Delta_\belset(x)]$ changes with $x$.

\paragraph{Method Overview} Our method has three steps. 

\begin{figure}[!ht]\centering
  \includegraphics[width=.9\linewidth]{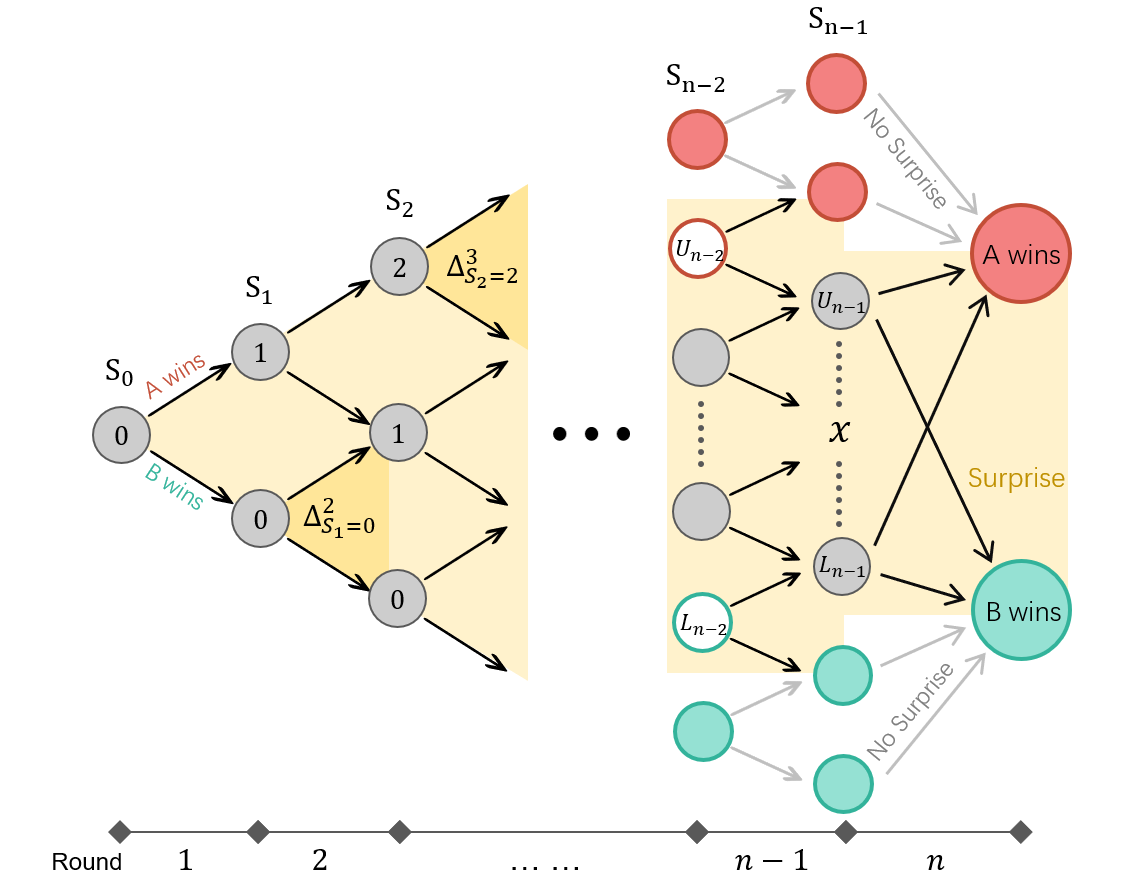}
  \caption{\textbf{Overview of our method} }
  \label{fig:last2example}
\end{figure}

\begin{description}
\item [Step 1 (Main Technical Lemma)] We show that fixing $n,\alpha, \beta$, $\forall$ $x$, there exists a constant $C$ such that 
\[\E[\Delta_\belset(x)]=\E[\Delta_\belset^{n-1}(x)]*C+\E[\Delta_\belset^{n}(x)]\]
Thus, the choice of $x$ only depends on the trade-off of the penultimate and the final round's surprise. This significantly simplifies our analysis since we have a close form expression for the belief in the final two rounds. 
\item [Step 2 (Final \& Penultimate Rounds)] 
We define $\Delta_{\hisset^{(i-1)}}^i$ as the expected amount of surprise generated by round $i$ given the history $\hisset^{(i-1)}$.  $\Delta_{\state_{i-1}}^i$ is the amount of surprise generated by round $i$ given that Alice wins $\state_{i-1}$ rounds in the first $i-1$ rounds. Corollary~\ref{cor:beta} shows that $\Delta_{\hisset^{(i-1)}}^i$ only depends on the state of the history, that is, for any history $\hisins\in \{+,-\}^{i-1}$, $\Delta_{\state_{i-1}=\pluscount(\hisins)}^i=\Delta_{\hisset^{(i-1)}=\hisins}^i$.
\begin{description}
\item [The final round:] In the final round, we notice when the difference between Alice and Bob is strictly greater than $x$, the outcome of the whole competition does not change regardless of who wins the final round, and no surprise is generated in the final round. Formally, we define $L_{n-1}:=\frac{n-x}2$ and $U_{n-1}:=\frac{n+x-2}2$. Only states $S_{n-1}\in[L_{n-1},U_{n-1}]$ generate surprise. Thus, the total surprise generated in round $n$ is
\begin{align} \label{eq:lastround}
\E[\Delta_\belset^{n}(x)]=\overbrace{\sum_{j=L_{n-1}}^{U_{n-1}}\Pr[S_{n-1}=j]*\Delta_{S_{n-1}=j}^n}^\mytextn
\end{align}
In any state $S_{n-1}=j\in[L_{n-1},U_{n-1}]$, whoever wins the final round wins the whole competition and the analysis for all $\Delta_{S_{n-1}=j}^n$ is identical.

\item [The penultimate round:] In the penultimate round, similarly, we define $L_{n-2}:=\frac{n-x-2}2$ and $U_{n-2}:=\frac{n+x-2}2$. Similarly, only states $S_{n-2}\in[L_{n-2},U_{n-2}]$ generate surprise. The states in $(L_{n-2},U_{n-2})$ are similar. However, unlike the final round, here the states at the endpoints require different analysis. Therefore, we divide the analysis into three parts \begin{align}
\E[\Delta_\belset^{n-1}(x)]=&\overbrace{\Pr[S_{n-2}=L_{n-2}]*\Delta^{n-1}_{S_{n-2}=L_{n-2}}}^\mytexta+\overbrace{\Pr[S_{n-2}=U_{n-2}]*\Delta^{n-1}_{S_{n-2}=U_{n-2}}}^\mytextb+ \notag\\ 
&\underbrace{\sum_{j=L_{n-2}+1}^{U_{n-2}-1}\Pr[S_{n-2}=j]*\Delta^{n-1}_{S_{n-2}=j}}_\mytextm \label{eq:2tolast}
\end{align}

\end{description}
\item [Step 3 (Local Maximum): ] To calculate the optimal bonus $x$, we need to analyze the how $\E[\Delta_\belset(x)]$ changes with $x$. 

\begin{description}
\item [Finite Case] Since we require the bonus $x$ to be an integer with the same parity as $n$, we calculate the change of the function $\E[\Delta_\belset(x)]$ when $x$ is increased/decreased by a step size $2$. We find that there is only one local (and thus global) maximum.  
\item [Asymptotic Case] We calculate the derivative of  $\E[\Delta_\belset(x)]$ with respect to $x$, and find that it only has one zero solution which is a local (and also global) maximum.
    
\end{description}
\end{description}

\section{Main Technical Lemma}

In this section, we introduce our main technical lemma: regardless of the final round's point value $x$, the expected surprise in the first $n-1$ rounds has a fixed relative ratio. Thus, the overall surprise can be rewritten as a linear combination of the final and the penultimate rounds' expected surprise, where the coefficients are independent of $x$. This simplifies our analysis significantly since the choice of $x$ only depends on the trade-off between the final and the penultimate rounds' expected surprise.

\begin{lemma}[Main Technical Lemma]\label{lem:ratio}
When the prior is $\betadis(\alpha,\beta)$, the ratio of the surprise of round $i$ and round $i+1(i+1<n)$ is independent of the final round's point value $x$:
\[
\frac{\E[\Delta_\belset^{i}]}{\E[\Delta_\belset^{i+1}]}=\frac{i+\alpha+\beta}{i+\alpha+\beta-1},
\]
thus the overall surprise is a linear combination of the final and the penultimate round's surprise,
\[
\E[\Delta_\belset]=\sum_{i=1}^{n}\E[\Delta_\belset^i]=\E[\Delta_\belset^{n-1}]*(n+\alpha+\beta-2)*\harmo_{\alpha+\beta}(n-1)+\E[\Delta_\belset^{n}]
\]
where $\harmo_{\alpha+\beta}(n-1):=\sum_{i=1}^{n-1}\frac{1}{i+\alpha+\beta-1}$.  We  use $\harmo$ as shorthand for $\harmo_{\alpha+\beta}(n-1)$.
\end{lemma}

\paragraph{Proof Sketch} To prove this lemma, we first introduce Claim~\ref{cla:roundsurp} which gives a simple format of the expected surprise generated by a single round initialized from any history $h$. We then apply the claim to analyze two consecutive rounds initialized from any history and show that the expected surprise produced by these two rounds has a fixed relative ratio which only depends on the round number, $\alpha$, and $\beta$. We then extend the results to the expectation over all possible histories.

To simplify the notation in the proof, we introduce shorthand notations for expectation of $p$, the belief values, and the difference between belief values here:

\[
\begin{cases}
   q:=\E[p|\hisset^{(i-1)}=\hisins]\\
q^+:=\E[p|\hisset^{(i)}=\hisins+]\\
q^-:=\E[p|\hisset^{(i)}=\hisins-]\\
\end{cases} \begin{cases}
   b:=\Pr[O=1|\hisset^{(i-1)}=\hisins]\\
b^+:=\Pr[O=1|\hisset^{(i)}=\hisins+]\\
b^-:=\Pr[O=1|\hisset^{(i)}=\hisins-]\\
\end{cases} \begin{cases}
   d:=b^+\minus b^-\\
d^+:=b^{++}-b^{+-}\\
d^-:=b^{-+}-b^{--}\\

\end{cases}
\]
where $ b^{++},b^{+-},b^{-+},b^{--}$ are defined analogously. These notations are also illustrated in Figure~\ref{fig:surpcalc} and Figure~\ref{fig:surprise_example_2}. 

\begin{figure}[!ht]\centering
  \includegraphics[width=.65\linewidth]{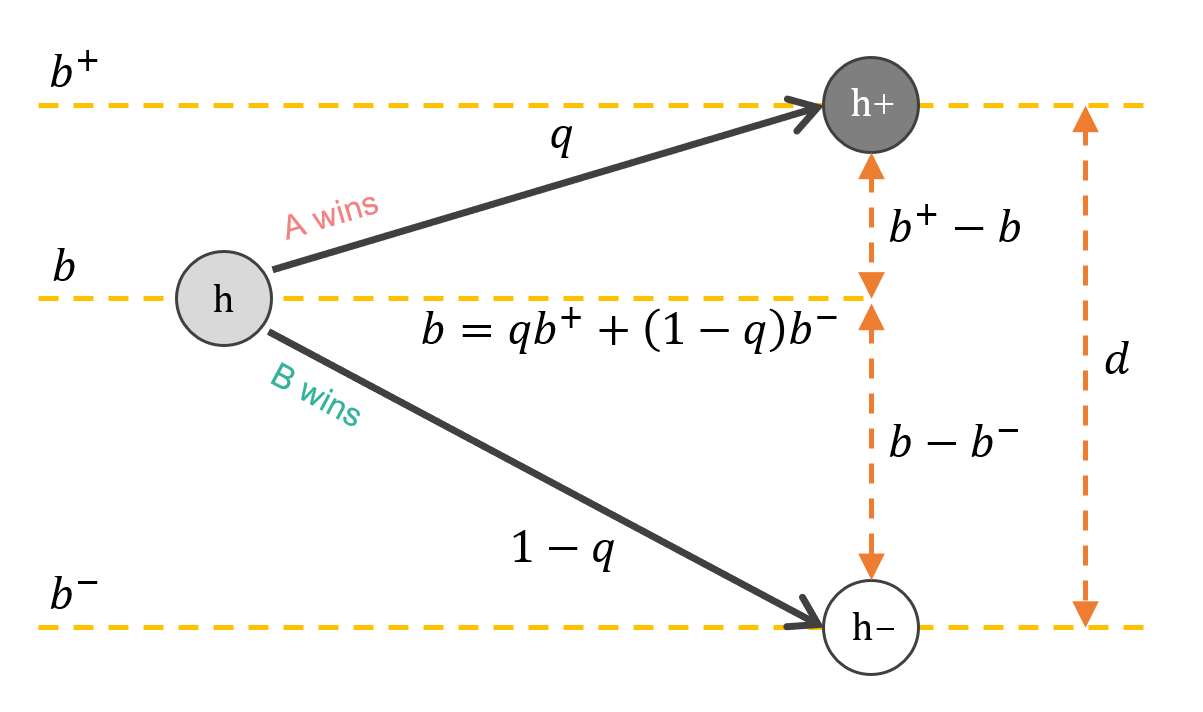}
  \caption{\textbf{A single round} }
  \label{fig:surpcalc}
\end{figure}

\begin{claim}\label{cla:roundsurp}
Given any history $\hisset^{(i-1)}=h$, we have \[\Delta_{\hisset^{(i-1)}=h}^i=d*2q(1-q).\] 
\end{claim}

\begin{proof}[Proof of Claim~\ref{cla:roundsurp}]
The expected amount of surprise generated in round $i$ is 
\begin{align} \label{eq:delta-i}
    \Delta_{\hisset^{(i-1)}=h}^i=&q * (b^+ - b)+(1-q)*(b -b^-)
\end{align}

Since $\belset$ is a martingale, we have 
\begin{align} \label{eq:b}
B_i=&\E[B_{i+1}|\hisset^{(i-1)}=\hisins] \nonumber\\
b =& q*b^++(1-q)*b^-
\end{align}

By substituting $b$ from \eqref{eq:b} into \eqref{eq:delta-i} and simplifying the equation, we have

 \[\Delta_{\hisset^{(i-1)}=h}^i=d*2q(1-q)\]
\end{proof}

\begin{figure}[!ht]\centering
  \includegraphics[width=.9\linewidth]{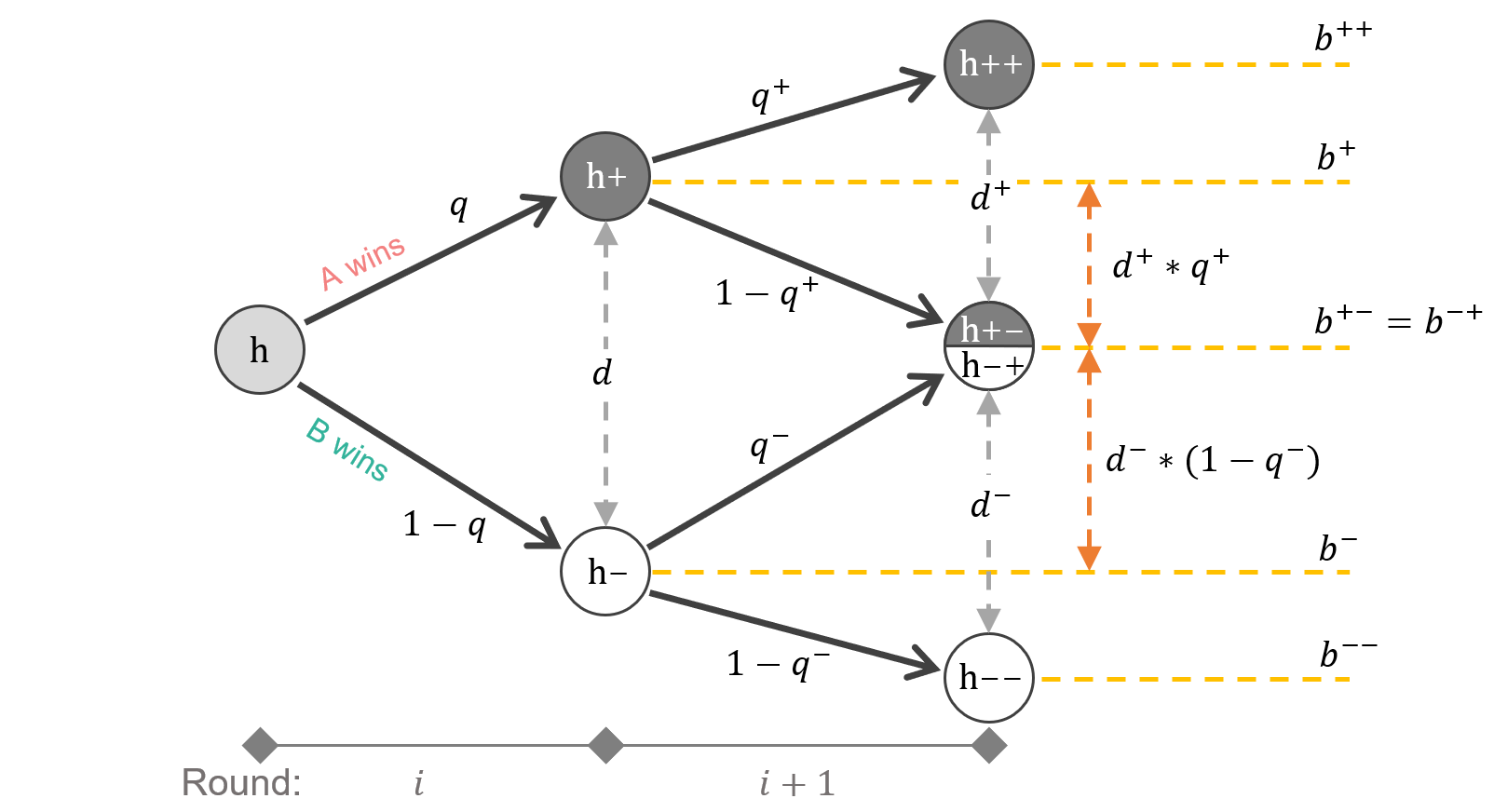}
  \caption{\textbf{Two rounds} }
  \label{fig:surprise_example_2}
\end{figure}

\begin{proof}[Proof of Lemma~\ref{lem:ratio}]

Figure~\ref{fig:surprise_example_2} shows two consecutive rounds starting from history $\hisset^{(i-1)}=\hisins$. We show that the ratio between these two rounds' surprise is fixed.

By Claim~\ref{cla:roundsurp}, we have \[\Delta_{\hisset^{(i-1)}=h}^i=d*2q(1-q)\]

and the expected amount of surprise generated in round $i+1$ is 
\begin{align*}
\E[\Delta_{\hisset^{(i)}|\hisset^{(i-1)}=h}^{i+1}] &= q*\Delta^{i+1}_{\hisset^{(i)}=h+}+(1-q)*\Delta^{i+1}_{\hisset^{(i)}=h-}\\
&=q*d^+*2q^+(1-q^+)+(1-q)*d^-*2q^-(1-q^-)    
\end{align*}
 
Properties of Beta distribution imply that the belief given any history only depends on the state of the history  (Claim~\ref{cla:beta}), thus $b^{+-}=b^{-+}$. Then we have $d=d^+*q^++d^-*(1-q^-)$ (see Figure~\ref{fig:surprise_example_2}) and 
 
\begin{align*}
\frac{\Delta_{\hisset^{(i-1)}=h}^i}{\E[\Delta_{\hisset^{(i)}|\hisset^{(i-1)}=h}^{i+1}]}  &= \frac{d*2q(1-q)}{q*d^+*2q^+(1-q^+)+(1-q)*d^-*2q^-(1-q^-)}\\ &= \frac{(d^+*q^++d^-*(1-q^-))*2q(1-q)}{q*d^+*2q^+(1-q^+)+(1-q)*d^-*2q^-(1-q^-)}\\
                                        &= \frac{d^+q^+q*(1-q)+d^-(1-q^-)(1-q)*q}{d^+q^+q*(1-q^+)+d^-(1-q^-)(1-q)*q^-}
\end{align*}

We observe that if $q(1-q^+)=\Pr[h+-|h]=\Pr[h-+|h]=(1-q)q^-$,  $\frac{1-q}{1-q^+}=\frac{q}{q^-}$, and the ratio becomes $\frac{q}{q^-}$. 

The posterior of Beta distribution is still a Beta distribution. We denote the distribution over $p|(\hisset^{(i-1)}=\hisins)$ by $\betadis(\alpha',\beta')$. Then we have
\[
\begin{cases}
q=\frac{\alpha'}{\alpha'+\beta'}\\
q^-=\frac{\alpha'}{\alpha'+\beta'+1}\\
q^+=\frac{\alpha'+1}{\alpha'+\beta'+1}
\end{cases}
\]
Thus, we have $\Pr[h+-|h]=\Pr[h-+|h]$, and the ratio becomes $\frac{q}{q^-}=\frac{\alpha'+\beta'+1}{\alpha'+\beta'}$, i.e., 
\[
\frac{\Delta_{\hisset^{(i-1)}=\hisins}^{i}}{\E[\Delta_{\hisset^{(i)}|\hisset^{(i-1)}=\hisins}^{i+1}]}=\frac{\alpha'+\beta'+1}{\alpha'+\beta'}
\]

Since the prior $p$ follows $\betadis(\alpha,\beta)$, for any history $\hisins$ of first $i-1$ rounds, the distribution $\betadis(\alpha',\beta')$ over $p|(\hisset^{(i-1)}=\hisins)$ satisfies that $\alpha'+\beta'=i+\alpha+\beta-1$. So starting from any history of first $i-1$ rounds, 
\begin{align}
\frac{\Delta_{\hisset^{(i-1)}=\hisins}^{i}}{\E[\Delta_{\hisset^{(i)}|\hisset^{(i-1)}=\hisins}^{i+1}]} = \frac{i+\alpha+\beta}{i+\alpha+\beta-1} \label{eq:singleratio}
\end{align}

which only depends on $i$, given fixed $\alpha,\beta$.

Therefore 
\begin{align*}
\frac{\E[\Delta_\belset^{i}]}{\E[\Delta_\belset^{i+1}]}=&\frac{\E[\Delta_{\hisset^{(i-1)}}^{i}]}{\E[\Delta_{\hisset^{(i)}}^{i+1}]}\\
=&\frac{\E[\Delta_{\hisset^{(i-1)}}^{i}]}{\E_{\hisset^{(i-1)}}[\E_{\hisset^{(i)}}[\Delta_{\hisset^{(i)}|\hisset^{(i-1)}}^{i+1}|\hisset^{(i-1)}]]}\tag{chain rule}\\
=&\frac{\E_{\hisset^{(i-1)}}[\frac{i+\alpha+\beta}{i+\alpha+\beta-1}*\E_{\hisset^{(i)}}[\Delta_{\hisset^{(i)}|\hisset^{(i-1)}}^{i+1}]]}{\E_{\hisset^{(i-1)}}[\E_{\hisset^{(i)}}[\Delta_{\hisset^{(i)}|\hisset^{(i-1)}}^{i+1}|\hisset^{(i-1)}]]}\tag{due to formula~\eqref{eq:singleratio} }\\
=&\frac{i+\alpha+\beta}{i+\alpha+\beta-1}\\
\end{align*}

Then we can use the expected surprise in round $n-1$ to represent the amount of surprise in any round $i\leq n-1$:

\begin{align*}
\frac{\E[\Delta_\belset^{i}]}{\E[\Delta_\belset^{n-1}]}   &=\prod_{j=i}^{n-2}\frac{\E[\Delta_\belset^{j}]}{\E[\Delta_\belset^{j+1}]}\\
\E[\Delta_\belset^{i}]                             &=\frac{n+\alpha+\beta-2}{i+\alpha+\beta-1}*\E[\Delta_\belset^{n-1}]\\
\end{align*}
Therefore, the sum of the expected surprise in the first $n-1$ rounds is
\begin{align*}
\sum_{i=1}^{n-1}\E[\Delta_\belset^i]   &=\E[\Delta_\belset^{n-1}]*(\sum_{i=1}^{n-1}\frac{n+\alpha+\beta-2}{i+\alpha+\beta-1})\\
&=\E[\Delta_\belset^{n-1}]*(n+\alpha+\beta-2)*\harmo_{\alpha+\beta}(n-1)\\
\end{align*}
\end{proof}

\section{Finite Case}

In this section, we follow our method overview to study the finite case. We first present our results in Section~\ref{sec:finiteresults}. We then show a general analysis in Section~\ref{sec:analyzegeneral} and apply the results of the general analysis to study two special cases: 1) symmetric case: $\alpha=\beta$ including the uniform case: $\alpha=\beta=1$; 2) certain case $\alpha=\lambda p,\beta = \lambda(1-p), \lambda\rightarrow \infty$. Finally, we apply the results in Section~\ref{sec:analyzegeneral} to show a linear algorithm for general beta prior setting. 

\subsection{Results in Finite Case}\label{sec:finiteresults}

Recall that $\round(x):=$the nearest\footnote{When there is a tie, we pick the smaller one.} integer to $x$ that has the same parity as $n$, and $\harmo:=\harmo_{\alpha+\beta}(n-1)=\sum_{i=1}^{n-1}\frac{1}{i+\alpha+\beta-1}$.

\begin{theorem}
\label{thm:spc}
For all $\alpha\geq\beta\geq 1$ \footnote{Note that assuming $\alpha\geq\beta$ does not lose generality since we can exchange Alice and Bob.}, $n>1$, 
\begin{itemize}
    \item \textbf{Symmetric $\alpha=\beta$} 
    \[
    x^*(\alpha,\alpha,n)=\round(\frac{n-1}{2\alpha\harmo-\frac{n-1}{n+2\alpha-1}})
    \]
    \begin{itemize} 
        \item \textbf{Uniform $\alpha=\beta=1$}
        \[
        x^*(1,1,n)=\round(\frac{n-1}{2\harmo-\frac{n-1}{n+1}})
        \]
    \end{itemize}
    \item \textbf{Certain $\alpha=\lambda p,\beta = \lambda(1-p), \lambda\rightarrow \infty$}
        Let $F(x):=(2np-n-(x-1))p^{x-1}+(n-2np-(x-1))(1-p)^{x-1}$, $x\in [1,n-1]$, $F(x)=0$ has a trivial solution at $x=1$ and a unique non-trivial solution $\Tilde{x}$ when $p>\frac12$ and $n>\frac{1}{(\frac{1}{2}-p)\ln(\frac{1-p}{p})}$,
        \[
        x^*(\alpha,\beta,n) = \begin{cases}
        \round(\Tilde{x}) & \text{if $p>\frac12$ and $n>\frac{1}{(\frac{1}{2}-p)\ln(\frac{1-p}{p})}$}\\
        \round(1) & \text{otherwise}\\
        \end{cases}
        \]
        
        Moreover, if
        $
        p>\frac{1}{1+(a+1)^{-\frac{1}{a}}} \text{ where } a=2np-n-2>0 \footnote{ $\frac{1}{1+(a+1)^{-\frac{1}{a}}}<\frac{1}{1+e^{-1}}$ and when $a\rightarrow+\infty$, $\frac{1}{1+(a+1)^{-\frac{1}{a}}}\rightarrow\frac{1}{2}$}
        $, 
        then  \[x^*(\alpha,\beta,n)\in[\round(2np-n)-2,\round(2np-n)+2],\] that is, the optimal bonus is around the ``expected lead''. 
    \item \textbf{General} There exists an $O(n)$ algorithm to compute the optimal bonus $x^*(\alpha,\beta,n)$.
\end{itemize}
\end{theorem}

We then present corresponding numerical results here. Based on Theorem~\ref{thm:spc}, we draw the contours of $\Tilde{x}$ for varies of cases. The optimal $x^*$ is $\round(\Tilde{x})$. Though $n$ can only be positive integers, we also smooth the contours for other $n$. In the symmetric case, the optimal bonus size increases as the number of rounds $n$ increases and the amount of uncertainty $\frac{1}{2\alpha}$ decreases. In the certain case, as we predicted in theory, the optimal bonus size is close to ``expected lead''. In the general case, as $n$ gets larger, the result becomes closer to the asymptotic case (see Figure~\ref{fig:infinite}). 

Finally, we provide additional numerical results illustrating how the overall surprise depends on the bonus size in Figure~\ref{fig:numericalv}. For different settings, for all $x$, we directly compute $\E[\Delta_\belset(x)]$ by using backward induction to compute all belief curves. We also annotate our theoretical optimal bonus $x^*=\round(\Tilde{x})$ based on Theorem~\ref{thm:spc}. The overall surprise varies with the bonus size and in some cases (e.g. certain, $n=20$, $p=0.7$), the optimal bonus creates surprise that doubles the amount of surprise created by the trivial settings ($x=\round(0)$ or $x=\round(n)$). Moreover, the optimal bonus depends on the properties of the setting: number of rounds and uniform, symmetric, or skewed. Additionally, in the figures, we see that the curves all have a single peak, so the local and global optima coincide.

\begin{figure}[!ht]\centering
  \includegraphics[width=.48\linewidth]{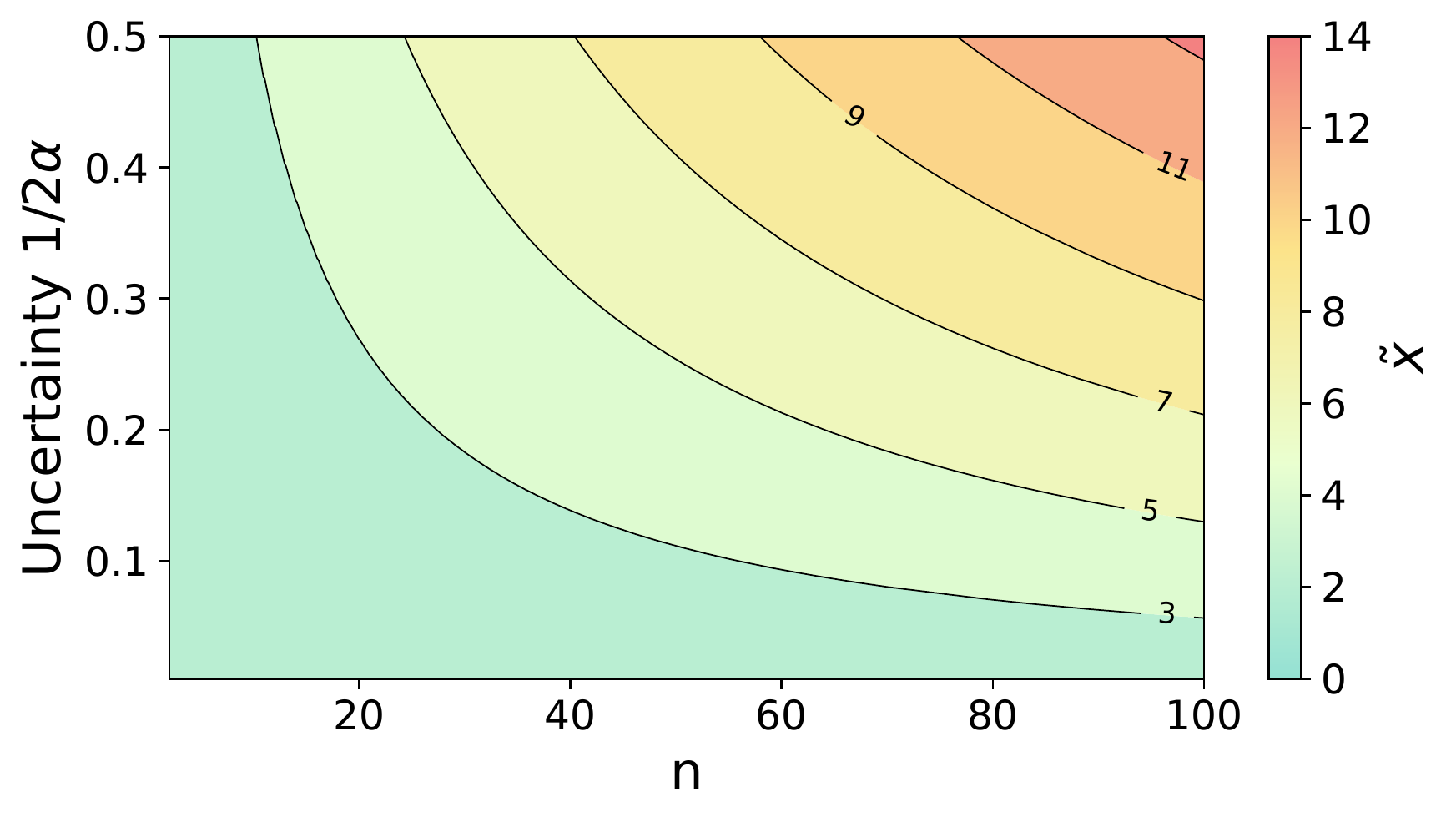}
  \caption{\textbf{Symmetric case}: Optimal bonus $x^*=\round(\Tilde{x})$}
  \label{fig:symmetric}
\end{figure}

\begin{figure}[!ht]\centering
  \subfigure[Optimal bonus $x^*=\round(\Tilde{x})$] {\includegraphics[width=.48\linewidth]{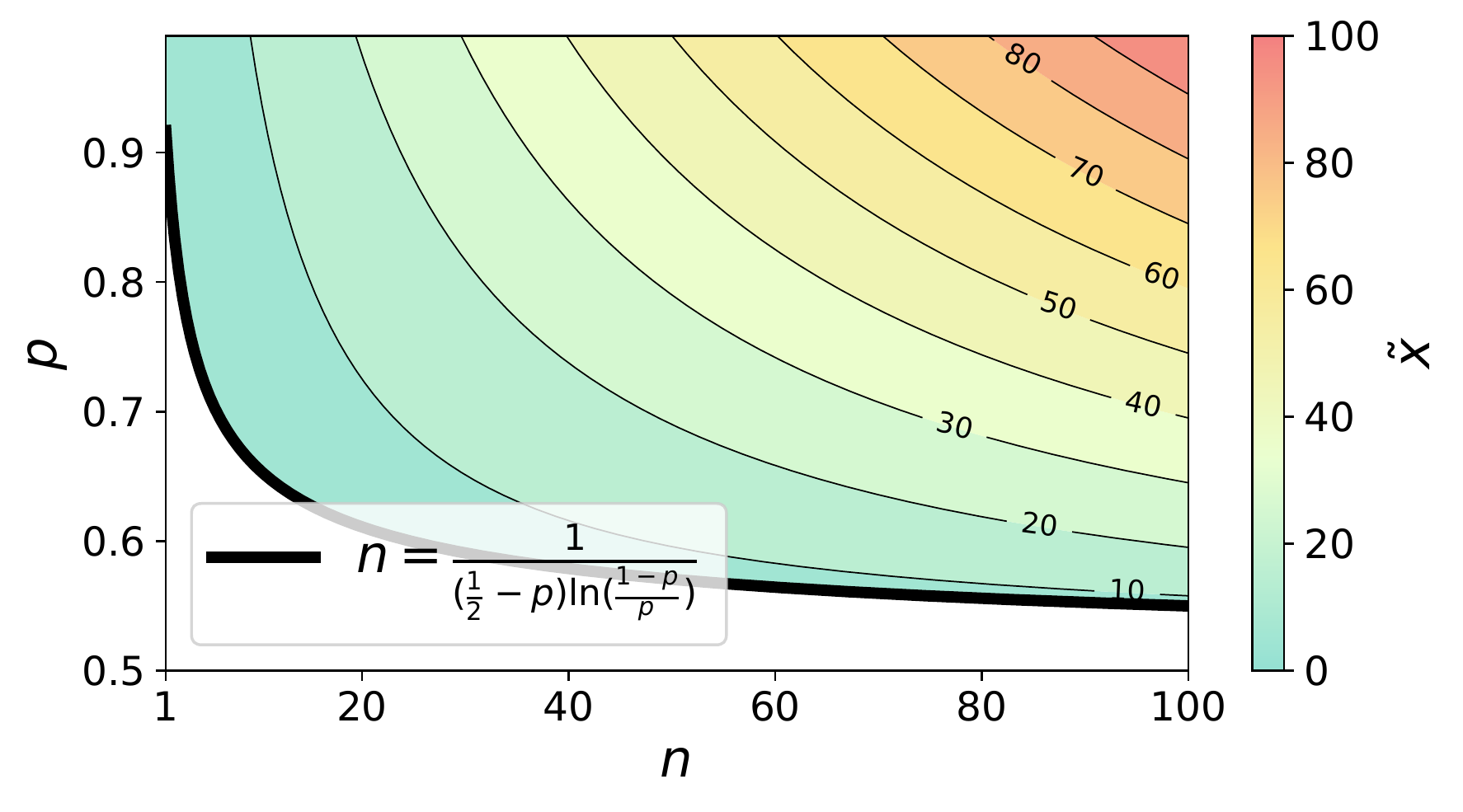}\label{fig:certain_opt}}
  \subfigure[$2np-n$]{\includegraphics[width=.48\linewidth]{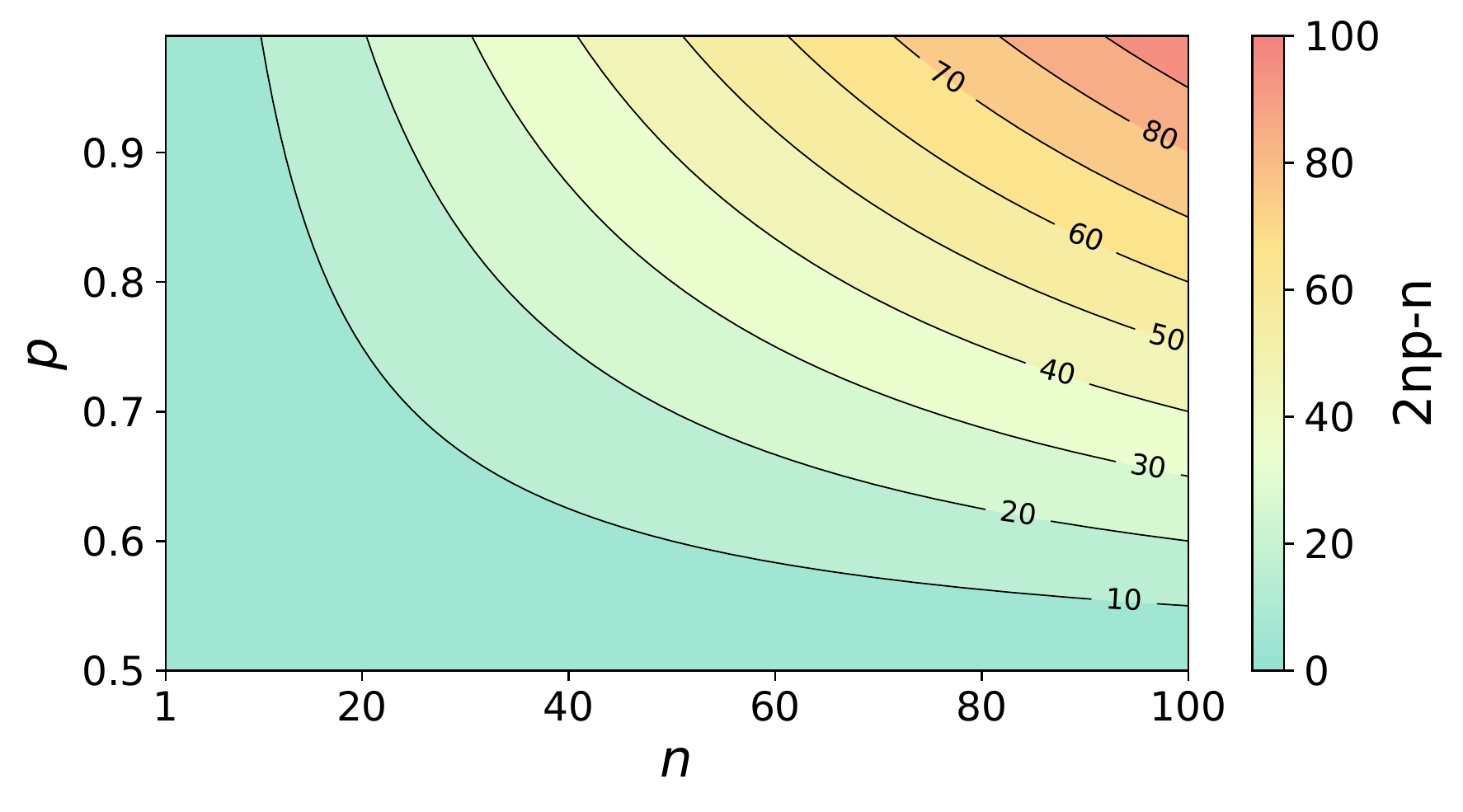}\label{fig:certain_2np}}
  \caption{\textbf{Certain case}}
  \label{fig:certain}
\end{figure}

\begin{figure}[!ht]\centering
  \subfigure[$n=5$]{\includegraphics[width=.48\linewidth]{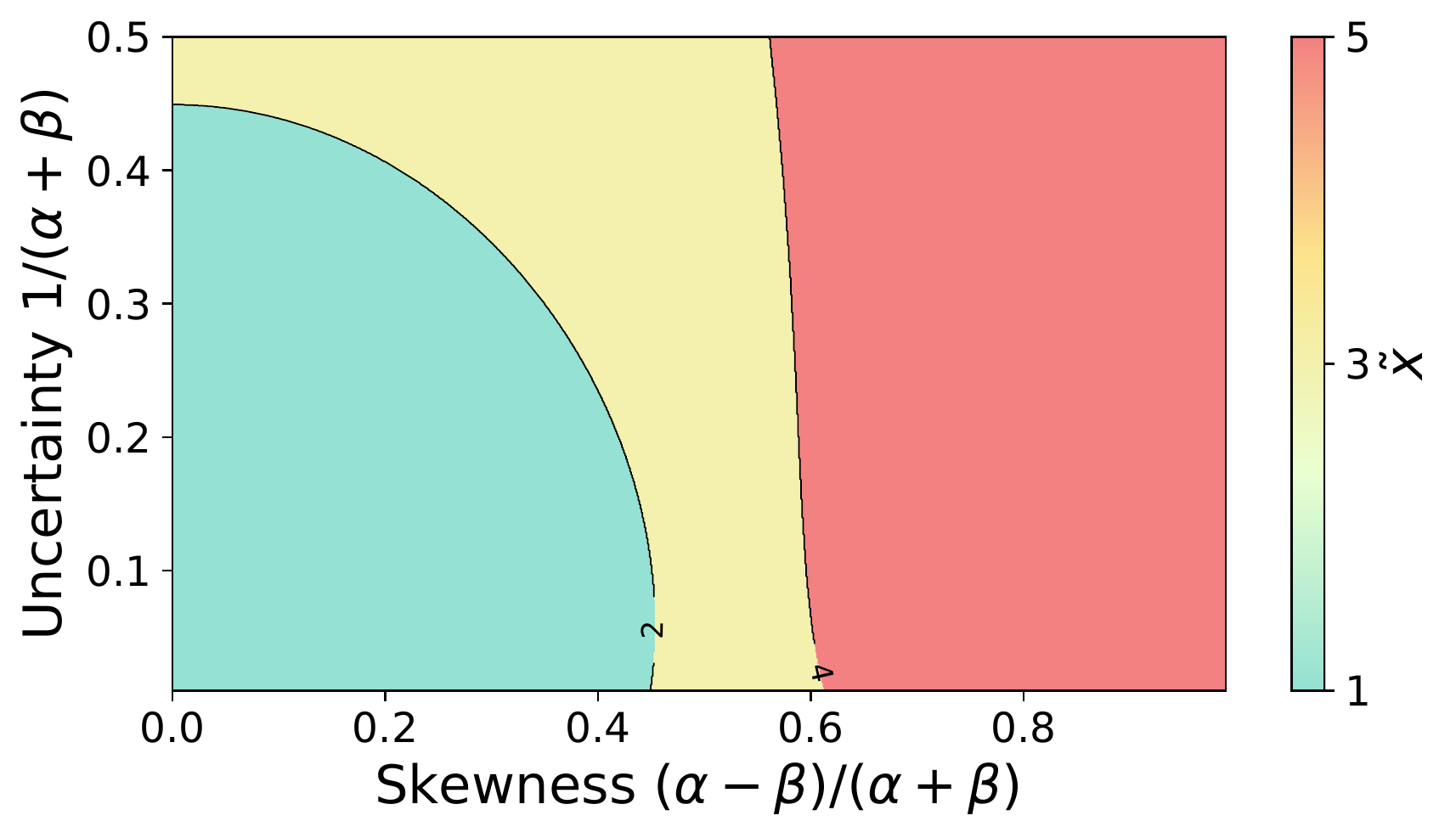}\label{fig:finite_5}}
  \subfigure[$n=10$]{\includegraphics[width=.48\linewidth]{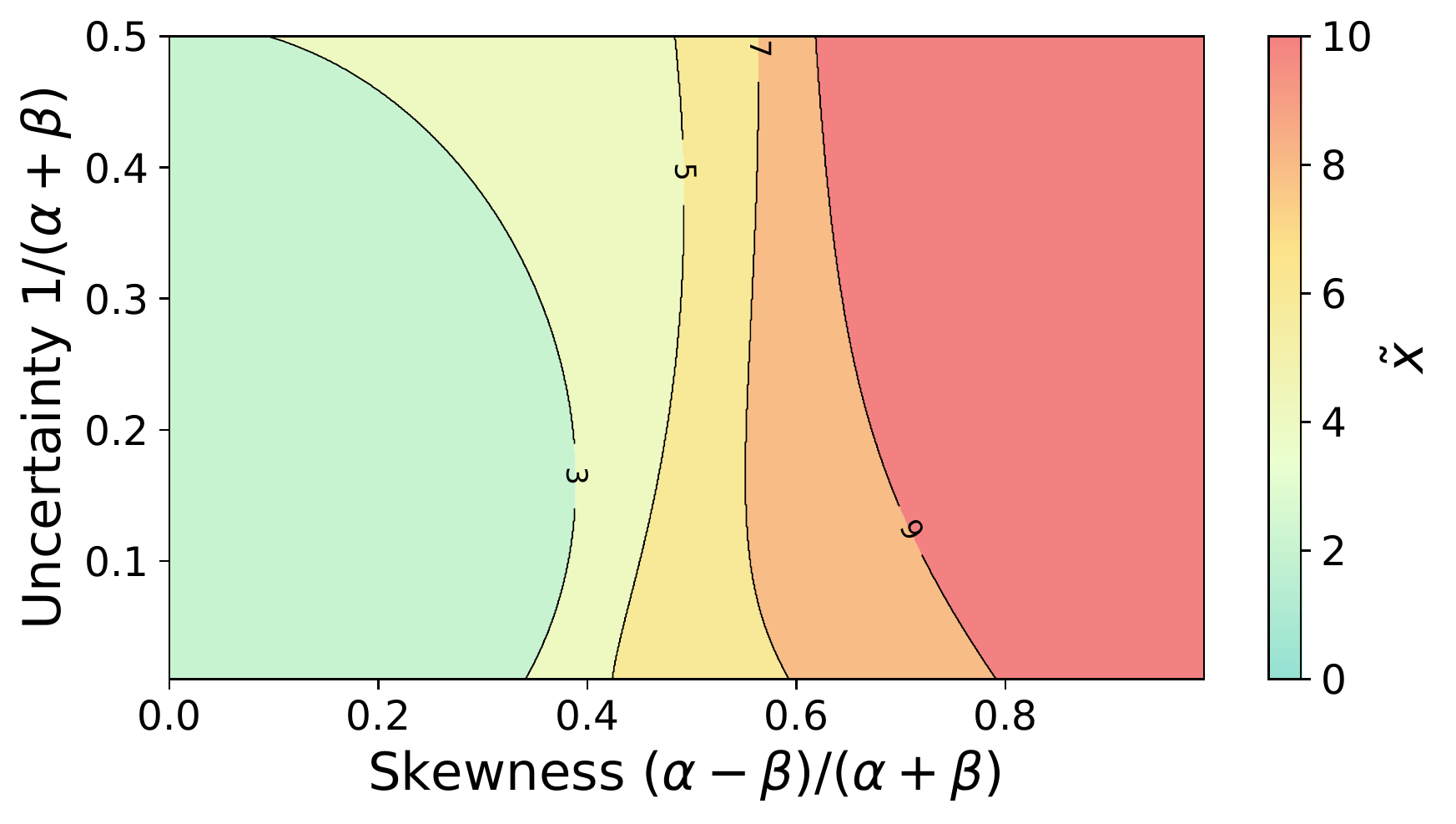}\label{fig:finite_10}}
  \subfigure[$n=20$]{\includegraphics[width=.48\linewidth]{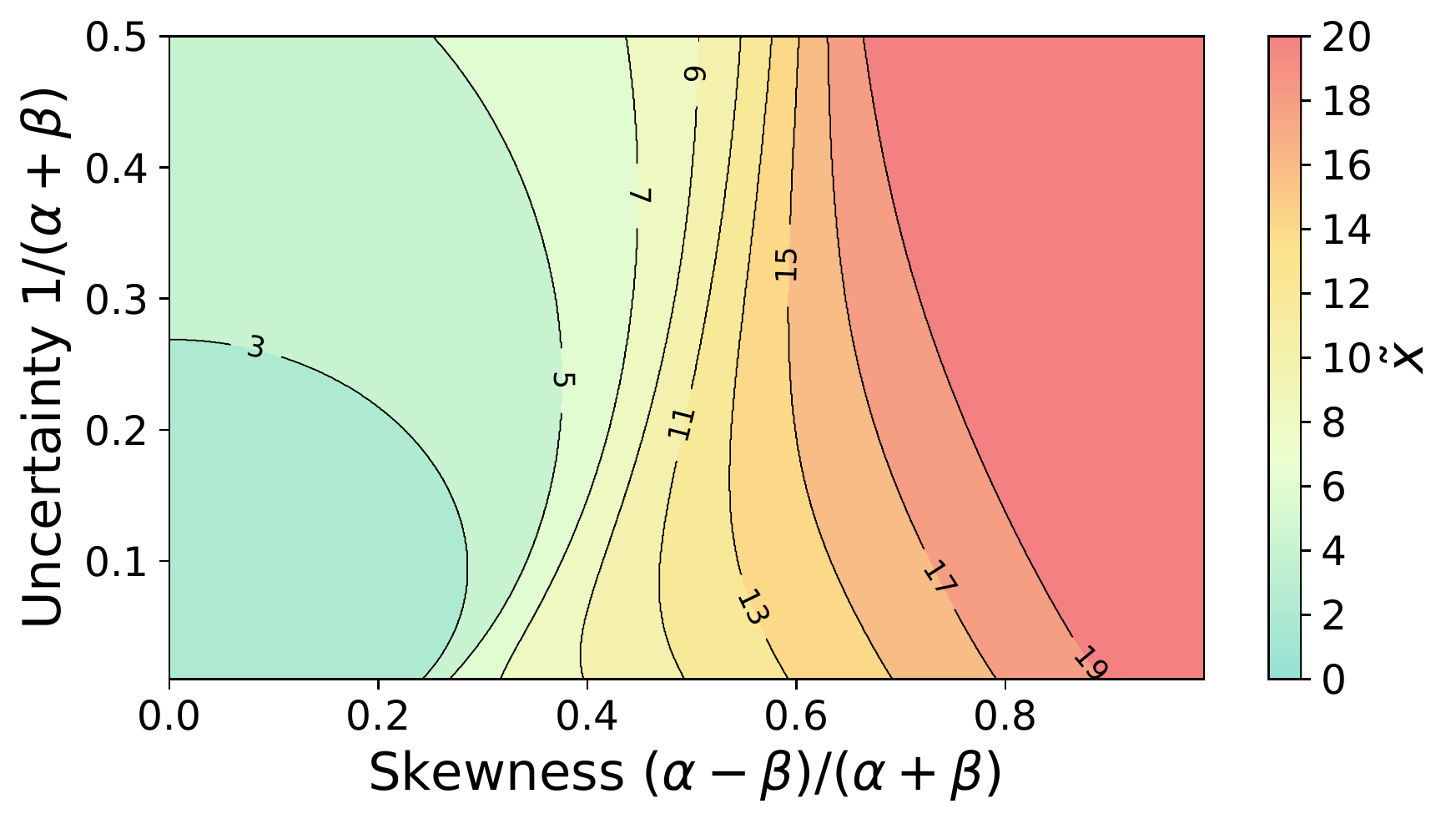}\label{fig:finite_20}}
  \subfigure[$n=40$]{\includegraphics[width=.48\linewidth]{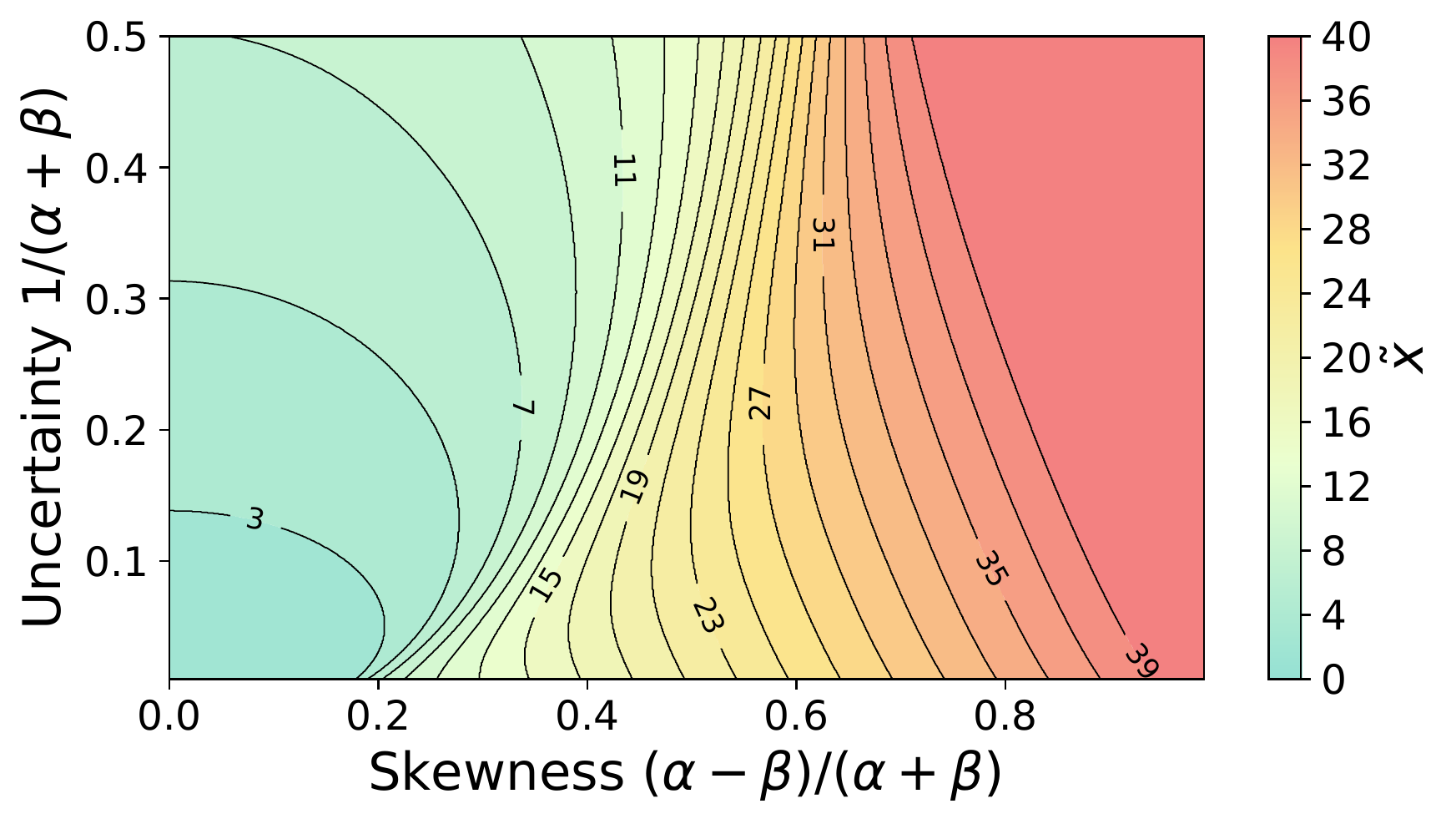}\label{fig:finite_40}}
  \caption{\textbf{General case}: Optimal bonus $x^*=\round(\Tilde{x})$. Here each area between two contour lines has the same optimal bonus $x^*$. For example, in $n=5$, the red area's optimal bonus size is $5$, yellow is $3$, cyan is $1$.}
  \label{fig:contourf_finite}
\end{figure}

\begin{figure}[!ht]\centering
  \subfigure[Uniform,n=10]{\includegraphics[width=.32\linewidth]{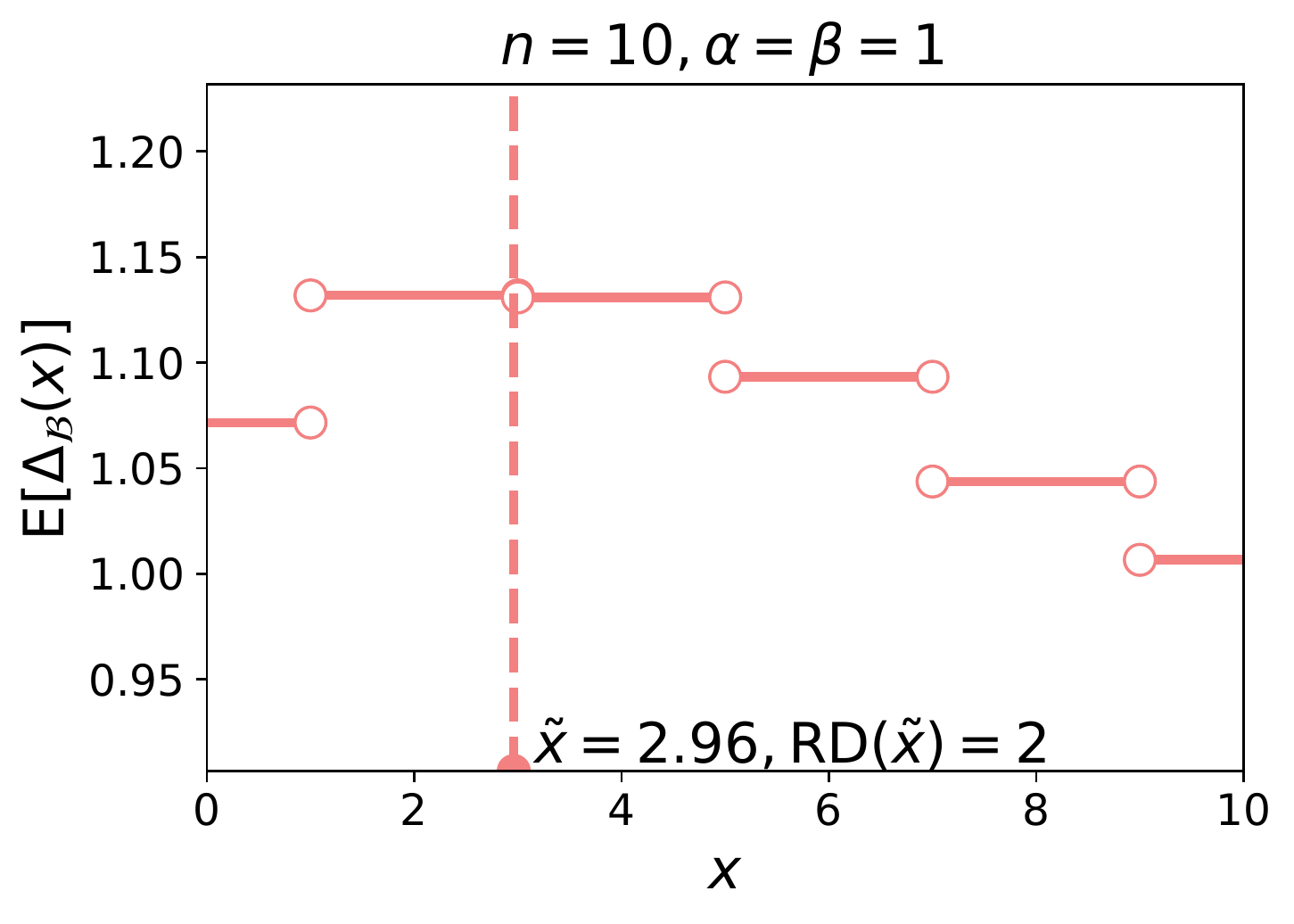}\label{fig:curve_uniform_10}}
  \subfigure[Symmetric,n=10]{\includegraphics[width=.32\linewidth]{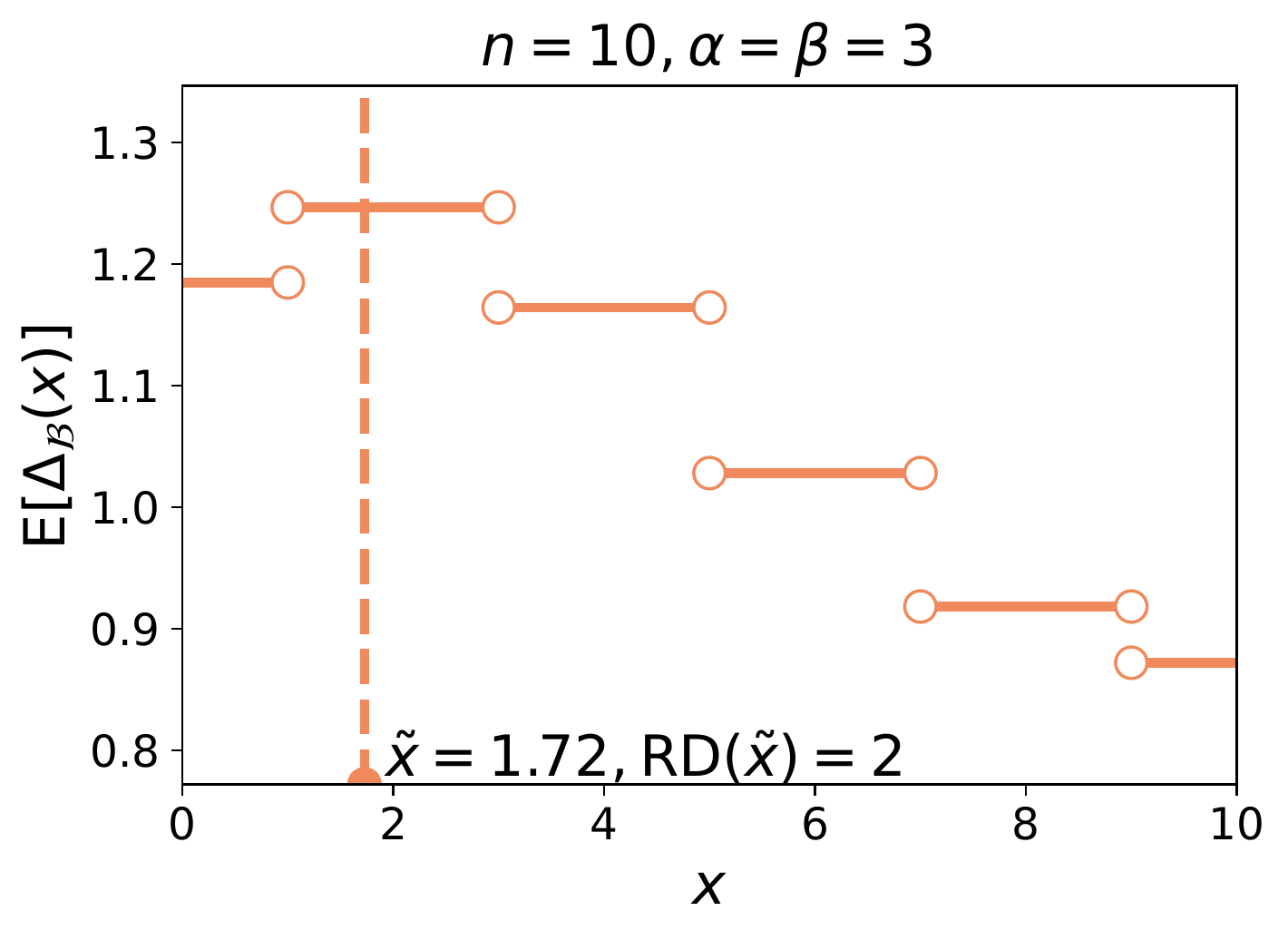}\label{fig:curve_symmetric_10}}
  \subfigure[Certain,n=10]{\includegraphics[width=.32\linewidth]{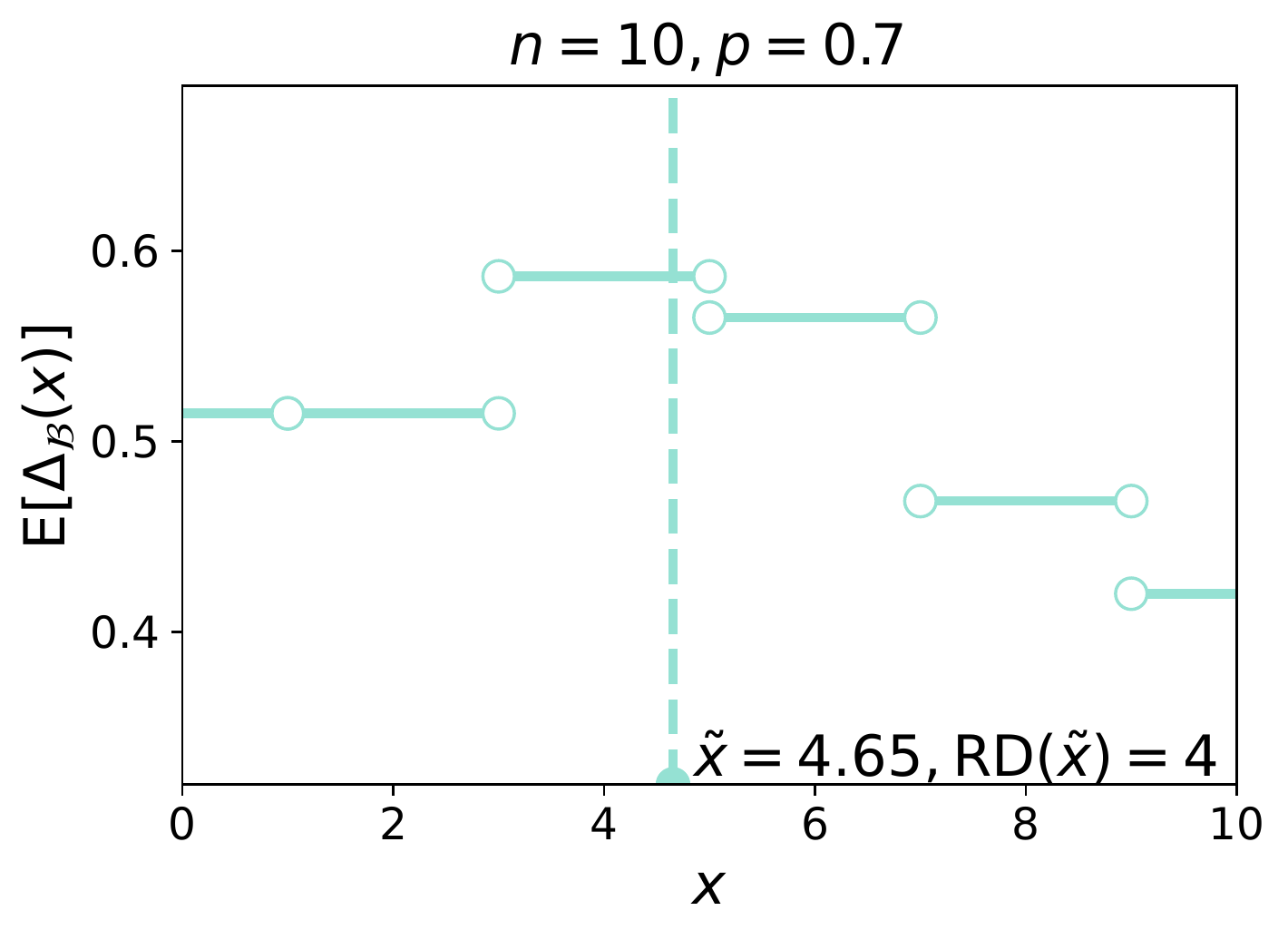}\label{fig:curve_certain_10}}
  \subfigure[Uniform,n=15]{\includegraphics[width=.32\linewidth]{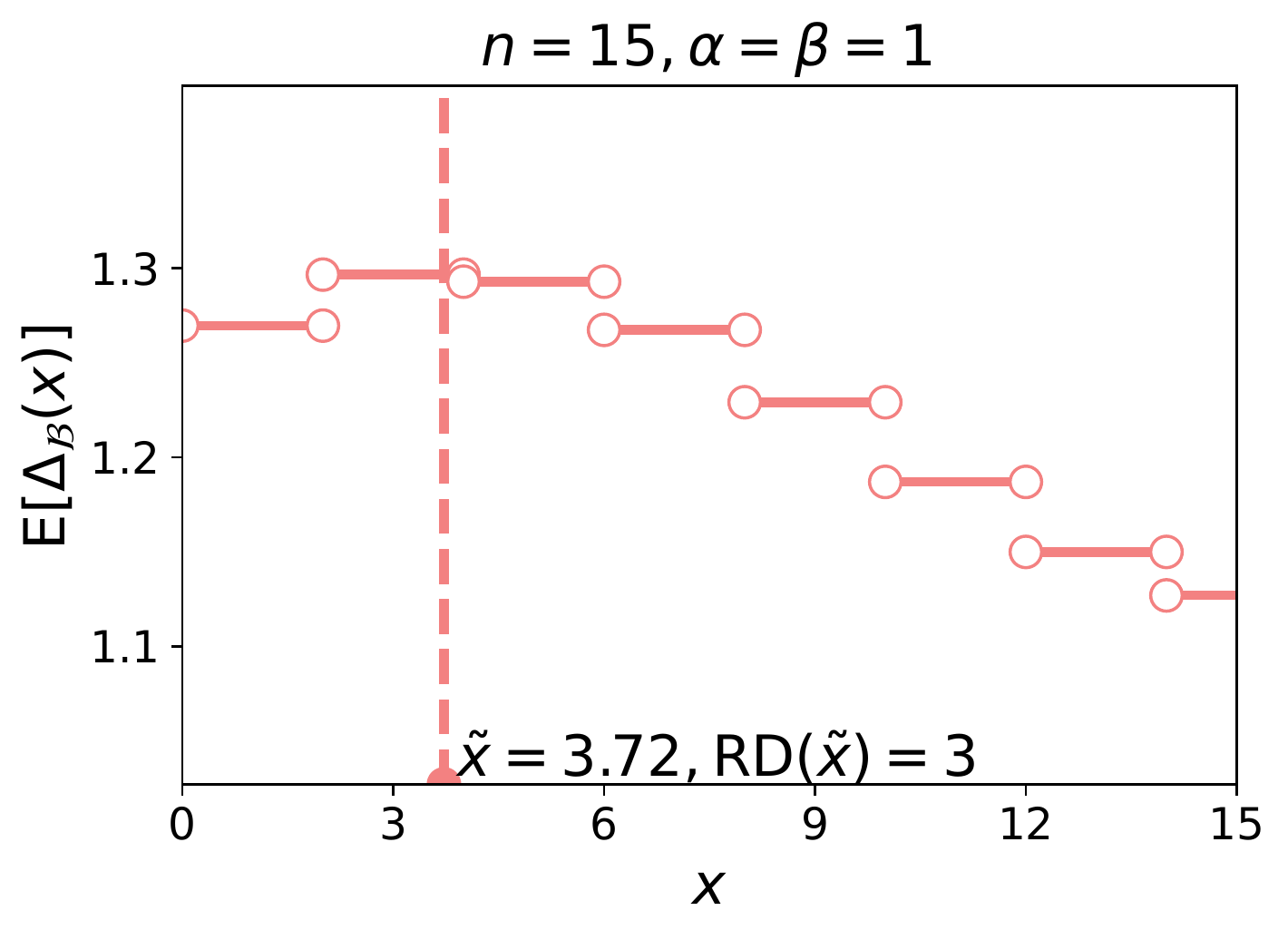}\label{fig:curve_uniform_15}}
  \subfigure[Symmetric,n=15]{\includegraphics[width=.32\linewidth]{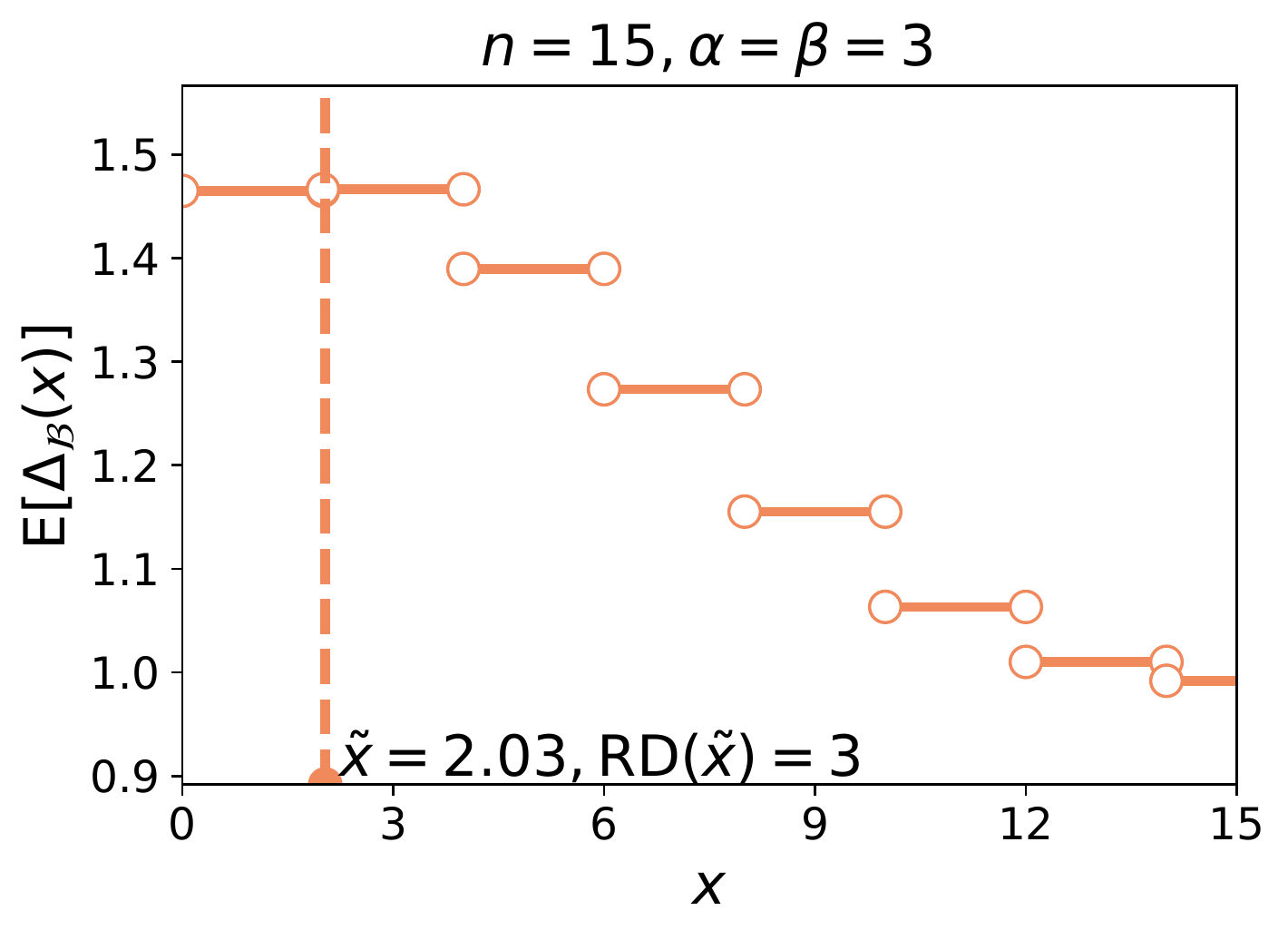}\label{fig:curve_symmetric_15}}
  \subfigure[Certain,n=15]{\includegraphics[width=.32\linewidth]{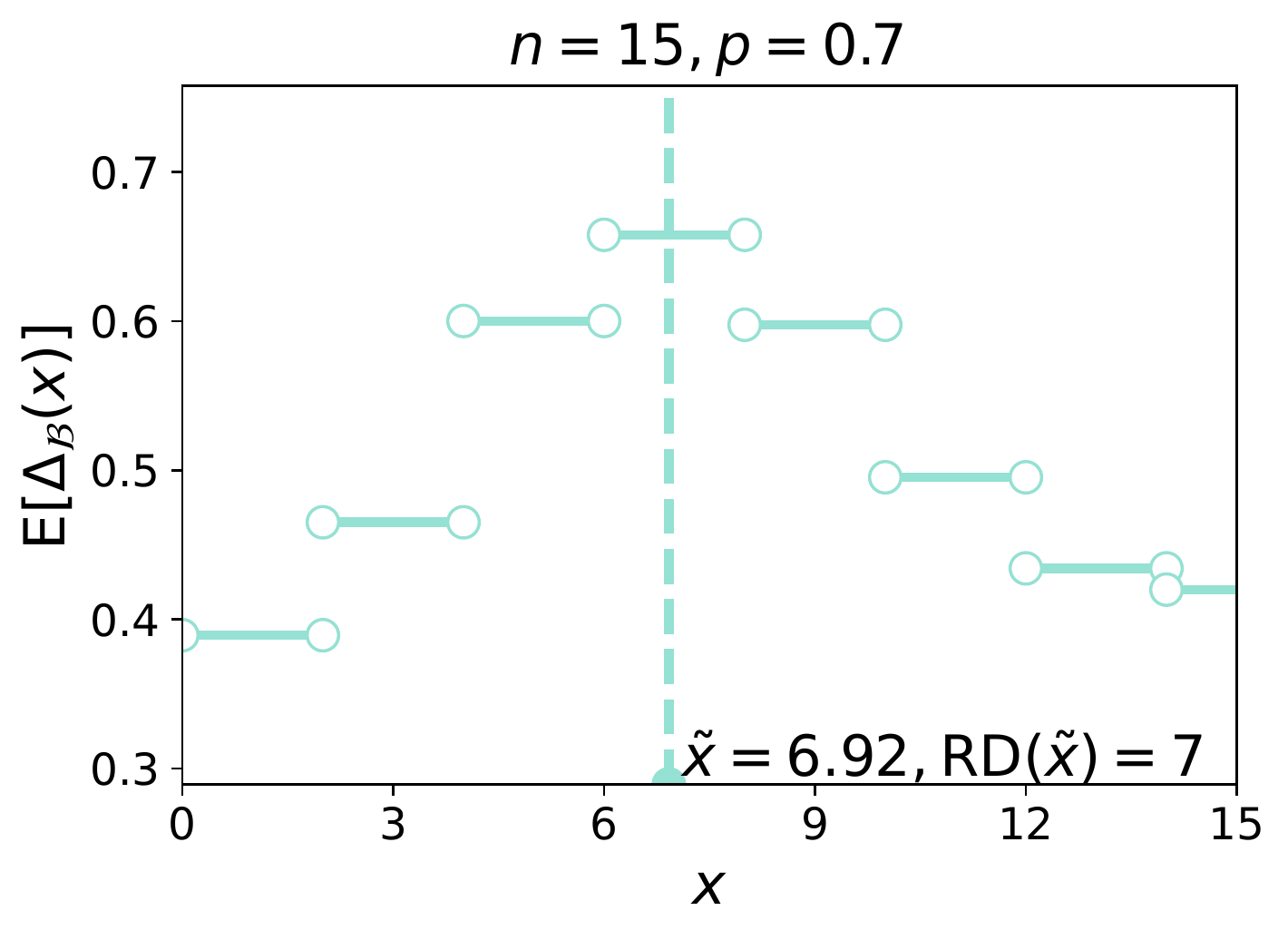}\label{fig:curve_certain_15}}
  \subfigure[Uniform,n=20]{\includegraphics[width=.32\linewidth]{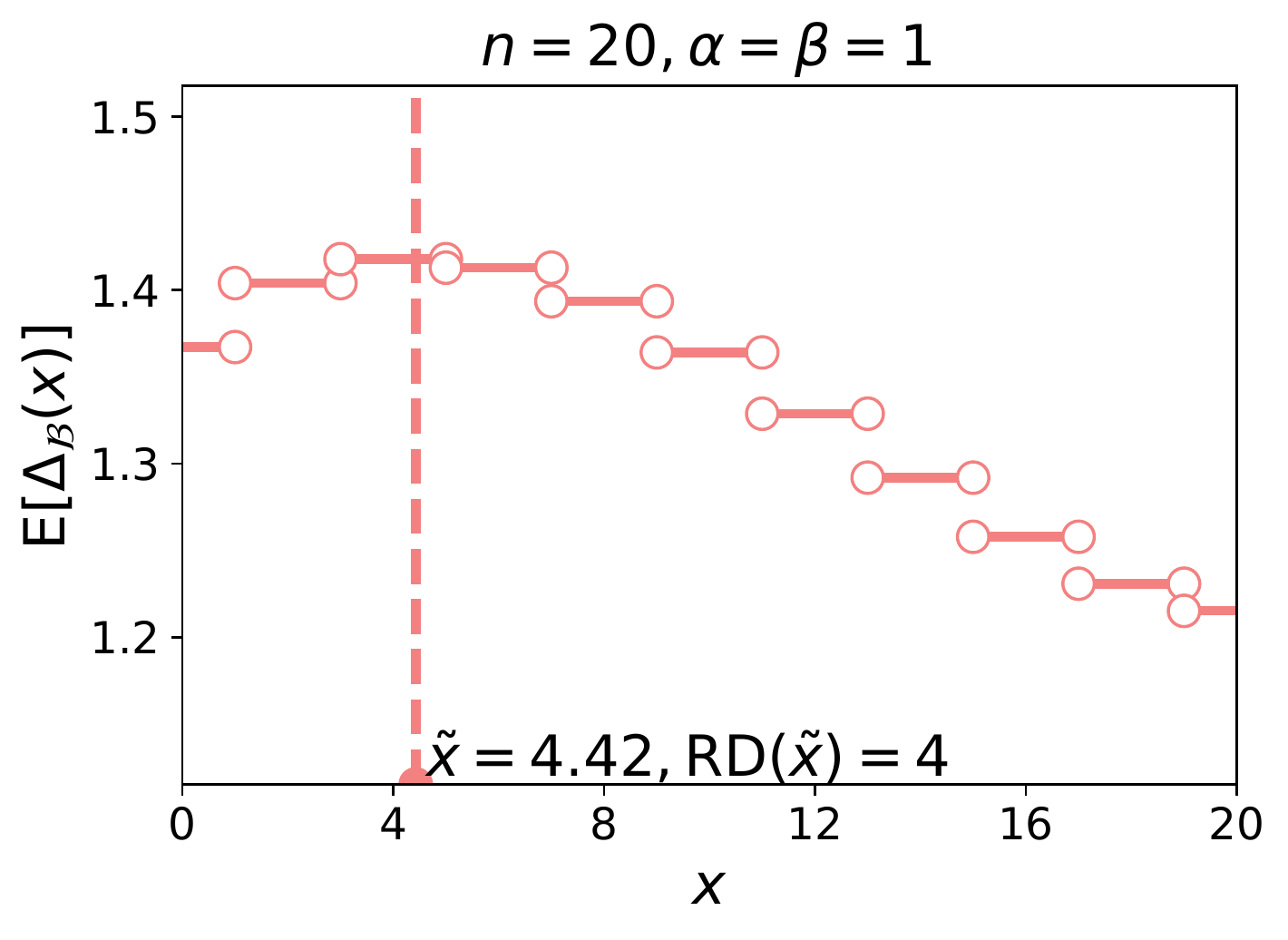}\label{fig:curve_uniform_20}}
  \subfigure[Symmetric,n=20]{\includegraphics[width=.32\linewidth]{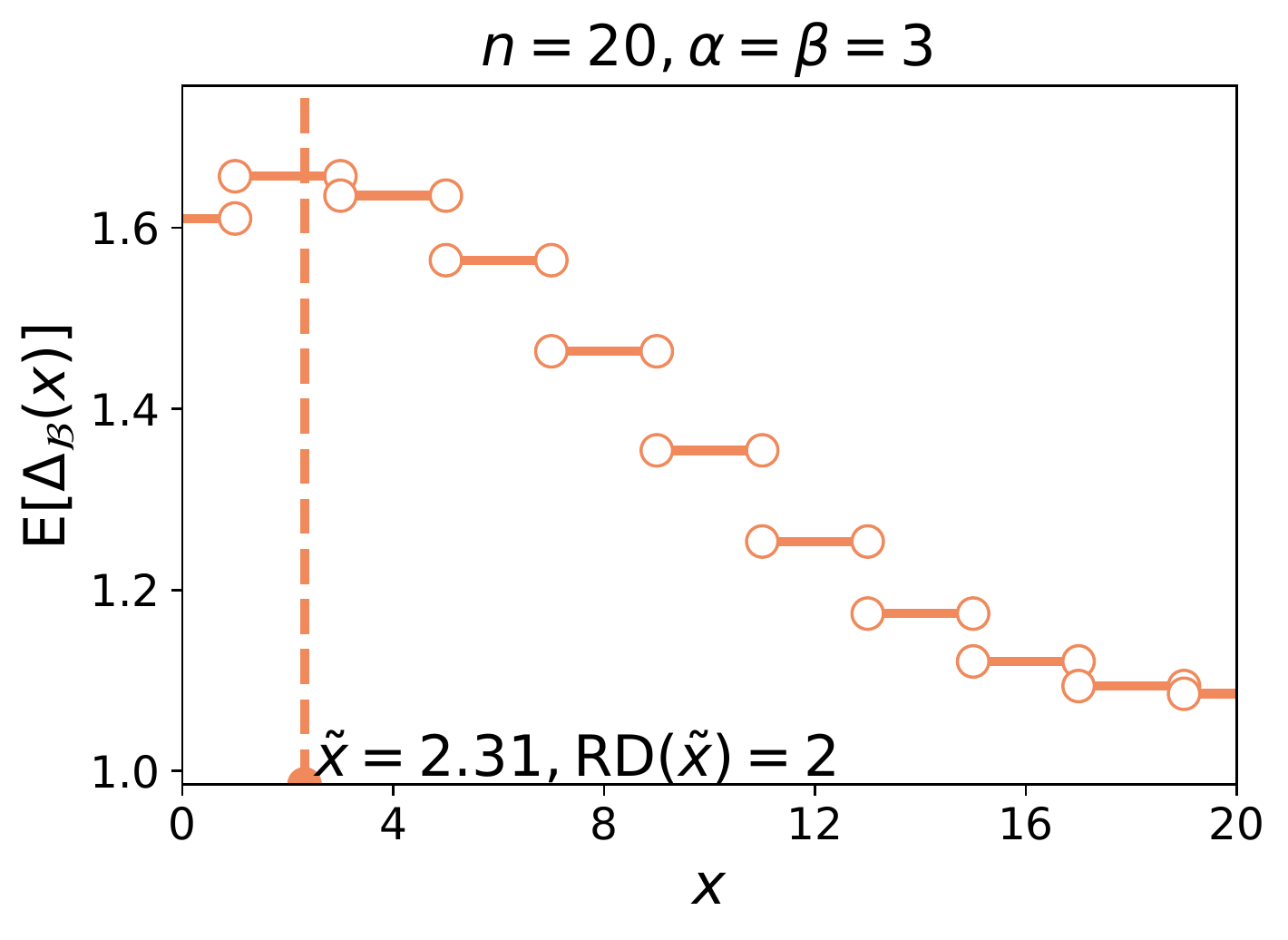}\label{fig:curve_symmetric_20}}
  \subfigure[Certain,n=20]{\includegraphics[width=.32\linewidth]{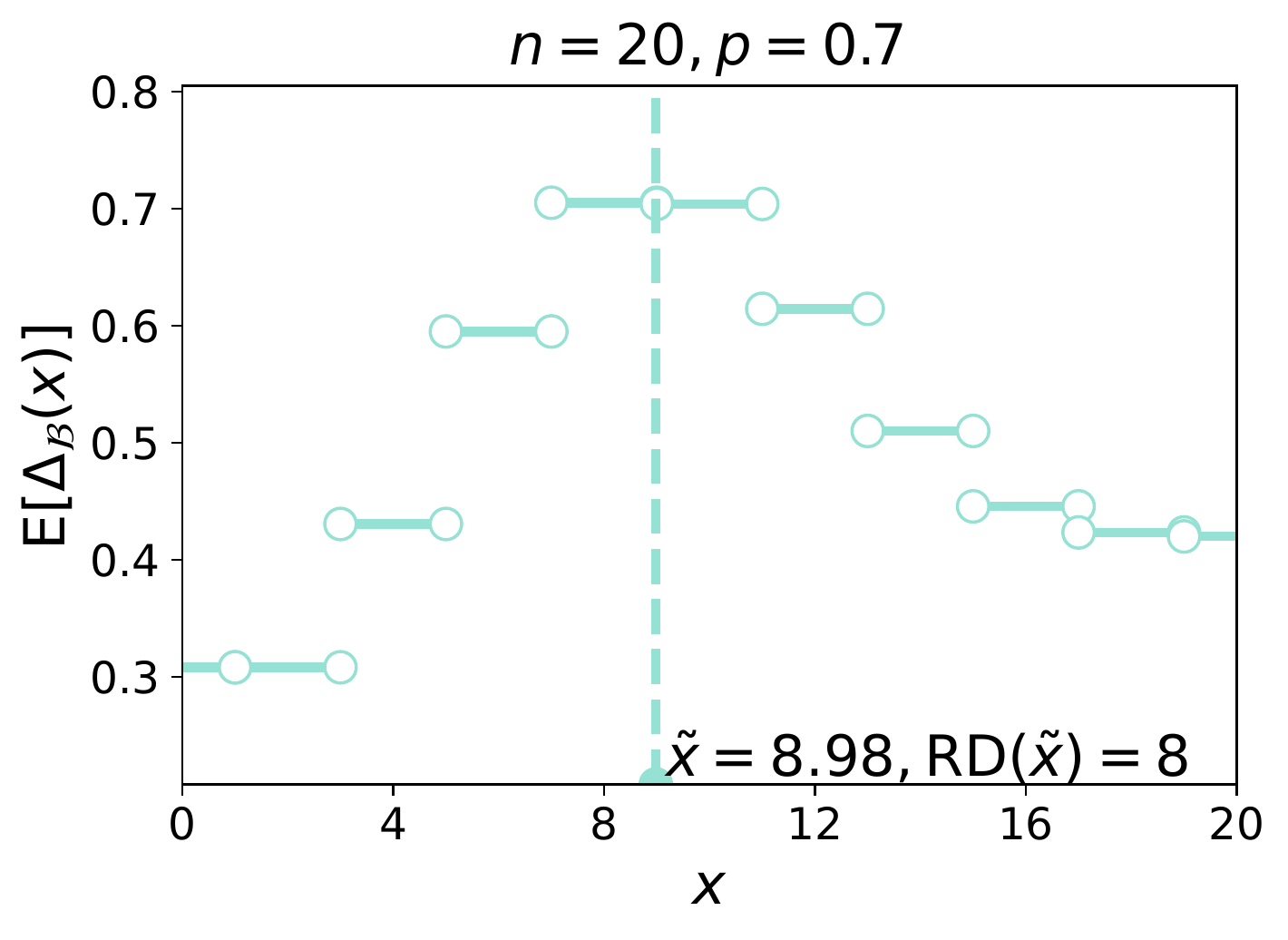}\label{fig:curve_certain_20}}
  \caption{\textbf{Relation between bonus size and overall surprise}}
  \label{fig:numericalv}
\end{figure}

\subsection{General Analysis}\label{sec:analyzegeneral}

In this subsection, we derive a formula which can be applied to all settings in the later sections. 
\begin{align}
\E[\Delta_\belset(x)]   =& \E[\Delta_\belset^{n-1}(x)]*(n+\alpha+\beta-2)*\harmo+\E[\Delta_\belset^n(x)]\tag{Lemma~\ref{lem:ratio}}\\ 
                          =& \bigg(\overbrace{\Pr[S_{n-2}=L_{n-2}]*\Delta^{n-1}_{S_{n-2}=L_{n-2}}}^\mytexta+\overbrace{\Pr[S_{n-2}=U_{n-2}]*\Delta^{n-1}_{S_{n-2}=U_{n-2}}}^\mytextb\label{eq:mainali} \notag\\
                           &+\underbrace{\sum_{j=L_{n-2}+1}^{U_{n-2}-1}\Pr[S_{n-2}=j]*\Delta^{n-1}_{S_{n-2}=j}}_\mytextm\bigg)*(n+\alpha+\beta-2)*\harmo\notag\\
                           &+\underbrace{\sum_{j=L_{n-1}}^{U_{n-1}}\Pr[S_{n-1}=j]*\Delta_{S_{n-1}=j}^n}_\mytextn \tag{recall \eqref{eq:lastround} and \eqref{eq:2tolast} in method overview}
\end{align}

\begin{figure}[!ht]\centering
  \includegraphics[width=.55\linewidth]{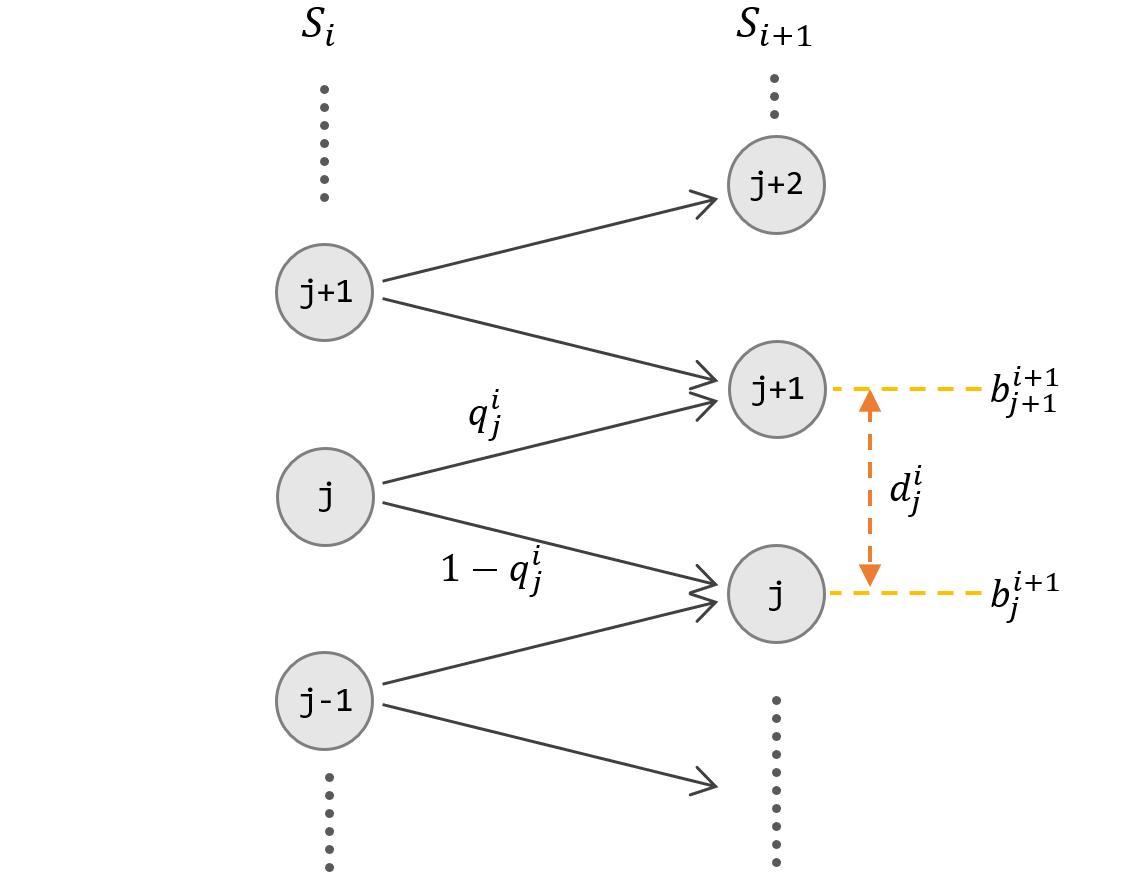}
  \label{fig:shorthand}
  \caption{\textbf{Illustration for the shorthand}}
\end{figure}

Here are shorthand notations:

\[\begin{cases}
q^i_j := \E[p|(\state_i=j)], &0\leq i\leq n-1\\
b^i_j := \Pr[O=1|(\state_i=j)], &0\leq i\leq n-1\\
d^i_j := b^{i+1}_{j+1} - b^{i+1}_{j}, &0\leq i\leq n-2\\
\end{cases}\]

The above definition for $d^i_j$ is only for $0\leq i\leq n-2$ since $d^{n-1}_j$ involves $b^n_j$ but the definition for $b^i_j $ is for $0\leq i\leq n-1$. The final round's belief value is either 0 or 1 and depends on both the number of rounds Alice wins among the first $n-1$ rounds ($S_{n-1}$) and whether Alice wins the final round ($\his_n=+$ or $\his_n=-$). Thus, we define the belief change in the final round directly as follows. \[
d^{n-1}_j := \Pr[O=1|(\his_n=+)\wedge (S_{n-1}=j)]-\Pr[O=1|(\his_n=-)\wedge (S_{n-1}=j)]\\
\]

In fact, for the no-surprise red/blue points in the final round (Figure~\ref{fig:last2example}), the belief change is 0, for other grey points, the belief change is 1. 

By substituting the above shorthand, we have 

\begin{align}
\E[\Delta_\belset(x)]   =& \bigg(\overbrace{\Pr[\state_{n-2}=L_{n-2}]*2 q^{n-2}_{L_{n-2}}*(1-q^{n-2}_{L_{n-2}})*d_{L_{n-2}}^{n-2}}^\mytexta\notag\\

&+\overbrace{\Pr[\state_{n-2}=U_{n-2}]*2q^{n-2}_{U_{n-2}}*(1-q^{n-2}_{U_{n-2}})*d_{U_{n-2}}^{n-2}}^\mytextb\notag\\

&+\overbrace{\sum_{j=L_{n-2}+1}^{U_{n-2}-1}\Pr[\state_{n-2}=j]*2q^{n-2}_{j}*(1-q^{n-2}_{j})*d_{j}^{n-2}}^\mytextm\bigg)*(n+\alpha+\beta-2)*\harmo\notag\\

&+\overbrace{\sum_{j=L_{n-1}}^{U_{n-1}}\Pr[\state_{n-1}=j]*2q^{n-1}_{j}*(1-q^{n-1}_{j})*d_{j}^{n-1}}^\mytextn \label{eq:general}
\end{align}

We further introduce a shorthand $Q^i_j :=\Pr[\state_i=j]*2q^i_j*(1-q^i_j), 0\leq i\leq n-1$ and by substituting this shorthand, we have 
\begin{align}
\E[\Delta_\belset(x)]   =& \bigg(\overbrace{Q_{L_{n-2}}^{n-2}*d_{L_{n-2}}^{n-2}}^\mytexta+\overbrace{Q_{U_{n-2}}^{n-2}*d_{U_{n-2}}^{n-2}}^\mytextb\label{eq:symmainali}\\                         &+\underbrace{\sum_{j=L_{n-2}+1}^{U_{n-2}-1}Q_{j}^{n-2}*d_{j}^{n-2}}_\mytextm\bigg)*(n+\alpha+\beta-2)*\harmo+\underbrace{\sum_{j=L_{n-1}}^{U_{n-1}}Q_{j}^{n-1}*d_{j}^{n-1}}_\mytextn \notag
\end{align}

As we mentioned in the overview, we pick the final round and the penultimate round to represent the overall expected surprise since the belief change, in these two rounds, has a simple representation, as stated in the following lemma.

\begin{lemma}
We have \[
q^{i}_j=\frac{j+\alpha}{i+\alpha+\beta}
\] 
\[
\Pr[\state_i=j]=\frac{(\alpha+\beta-1)\binom{i}{j}\binom{\alpha+\beta-2}{\alpha-1}}{(\alpha+\beta+i-1)\binom{i+\alpha+\beta-2}{j+\alpha-1}}\footnote{When $\alpha+\beta$ is not an integer, we use $\binom{n}{k}:=\frac{\Gamma(n+1)}{\Gamma(k+1)\Gamma(n-k+1)}$ as the continuous generalization.}
\] 
\[
Q^i_j=\frac{2(\alpha+\beta-1)\binom{i}{j}\binom{\alpha+\beta-2}{\alpha-1}}{(\alpha+\beta+i)\binom{i+\alpha+\beta}{j+\alpha}}
\]
and
\[
d^{n-2}_{j}=
\begin{cases}
0, & j<L_{n-2}\\
\frac{L_{n-1}+\alpha}{n+\alpha+\beta-1}, & j=L_{n-2}\\
\frac{1}{n+\alpha+\beta-1}, & L_{n-2}<j<U_{n-2}\\
\frac{n-1-U_{n-1}+\beta}{n+\alpha+\beta-1}, & j=U_{n-2}\\
0, & j>U_{n-2}
\end{cases}
\qquad
d^{n-1}_{j}=
\begin{cases}
0, & j<L_{n-1}\\
1, & L_{n-1}\leq j\leq U_{n-1}\\
0, & j>U_{n-1}
\end{cases}
\]
\label{lem:generald}
\end{lemma}

We prove the lemma by delicate analysis based on the properties of Beta distribution. We defer the proof to appendix. Based on Lemma~\ref{lem:generald}, we substitute the final two rounds' belief change $d^{n-1}_j,d^{n-2}_j$ into formula~\eqref{eq:symmainali}:

\begin{align}
\E[\Delta_\belset(x)]=&\bigg(\overbrace{Q^{n-2}_{L_{n-2}}*\frac{L_{n-1}+\alpha}{n+\alpha+\beta-1}}^\mytexta+\overbrace{Q^{n-2}_{U_{n-2}}*\frac{n-1-U_{n-1}+\beta}{n+\alpha+\beta-1}}^\mytextb\notag\\
&+\underbrace{\sum_{j=L_{n-2}+1}^{U_{n-2}-1}Q^{n-2}_j*\frac{1}{n+\alpha+\beta-1}}_\mytextm\bigg)*(n+\alpha+\beta-2)*\harmo+\underbrace{\sum_{j=L_{n-1}}^{U_{n-1}}Q^{n-1}_j}_\mytextn\notag\\
=&\bigg(Q^{n-2}_{\frac{n-x-2}2}*\frac{\frac{n-x}2+\alpha}{n+\alpha+\beta-1}+Q^{n-2}_{\frac{n+x-2}2}*\frac{\frac{n-x}2+\beta}{n+\alpha+\beta-1}\notag\\
&+\sum_{j=\frac{n-x}2}^{\frac{n+x-4}2}Q^{n-2}_j*\frac{1}{n+\alpha+\beta-1}\bigg)*(n+\alpha+\beta-2)*\harmo+\sum_{j=\frac{n-x}2}^{\frac{n+x-2}2}Q^{n-1}_j \label{eq:gensurp}
\end{align} 
\footnote{For $x<2$, we have $L_{n-2}+1>U_{n-2}-1$, and for $x=0$, we have $L_{n-1}>U_{n-1}$. We define the summation from a larger subscript to a smaller superscript as zero. This definition is valid since in those cases, no surprise is generated.}

In order to find the optimal $x$, we calculate

\begin{align}
&\E[\Delta_\belset(x+1)]-\E[\Delta_\belset(x-1)] \notag\\
=&\bigg(Q^{n-2}_{\frac{n-x-3}2}*(\frac{n-x-1}2+\alpha)+Q^{n-2}_{\frac{n+x-1}2}*(\frac{n-x-1}2+\beta)-Q^{n-2}_{\frac{n-x-1}2}*(\frac{n-x+1}2+\alpha)\notag\\
&-Q^{n-2}_{\frac{n+x-3}2}*(\frac{n-x+1}2+\beta)+Q^{n-2}_{\frac{n-x-1}2}+Q^{n-2}_{\frac{n+x-3}2}\bigg)*\frac{(n+\alpha+\beta-2)\harmo}{n+\alpha+\beta-1}+Q^{n-1}_{\frac{n+x-1}2}+Q^{n-1}_{\frac{n-x-1}2}\notag\\
=&\bigg(\left(Q^{n-2}_{\frac{n-x-3}2}-Q^{n-2}_{\frac{n-x-1}2}\right)*(\frac{n-x-1}2+\alpha)+\left(Q^{n-2}_{\frac{n+x-1}2}-Q^{n-2}_{\frac{n+x-3}2}\right)*(\frac{n-x-1}2+\beta)\bigg)\notag\\
&*\frac{(n+\alpha+\beta-2)\harmo}{n+\alpha+\beta-1}+Q^{n-1}_{\frac{n+x-1}2}+Q^{n-1}_{\frac{n-x-1}2}
\label{eq:gendif}
\end{align}

Then the following claim shows that we can obtain the optimal bonus by finding ``local maximum'' $\Tilde{x}$. We defer the proof to appendix.

\begin{claim}[Local Maximum $\rightarrow$ Optimal Bonus]
\label{cla:roundopt}
If there exists $\Tilde{x}\in(0,n+1)$ such that for all $1\leq x<\Tilde{x}$, $\E[\Delta_\belset(x + 1)]\geq \E[\Delta_\belset(x - 1)]$ and when $\Tilde{x}\leq n-1$, for all $\Tilde{x}\leq x\leq n-1$, $\E[\Delta_\belset(x + 1)]\leq \E[\Delta_\belset(x - 1)]$, then $\round(\Tilde{x})$ is the optimal bonus.
\end{claim}

Later we show that formula~\eqref{eq:gendif} induces a linear algorithm to find the optimal bonus size in general and can be significantly simplified in the symmetric case and certain case.

\subsection{Symmetric Case}
\label{subsec:sym}

We start to analyze the symmetric case. 

\begin{observation}\label{ob:symmetric}
In the symmetric case, i.e. $\alpha=\beta$, $Q^i_j$ is also symmetric, that is $Q^i_j=Q^i_{i-j}$.
\end{observation}

We defer the proof to appendix. Based on the above observation, we can further simplify the formula~\eqref{eq:gendif}

\begin{align}
&\E[\Delta_\belset(x+1)]-\E[\Delta_\belset(x-1)] \notag\\ 
=&\bigg(\left(Q^{n-2}_{\frac{n-x-3}2}-Q^{n-2}_{\frac{n-x-1}2}\right)*(\frac{n-x-1}2+\alpha)+\left(Q^{n-2}_{\frac{n+x-1}2}-Q^{n-2}_{\frac{n+x-3}2}\right)*(\frac{n-x-1}2+\alpha)\bigg)\notag\\ \tag{based on formula~\eqref{eq:gendif}}
&*\frac{(n+2\alpha-2)\harmo}{n+2\alpha-1}+Q^{n-1}_{\frac{n+x-1}2}+Q^{n-1}_{\frac{n-x-1}2}\notag \\ \tag{$Q^{n-2}_{\frac{n+x-1}2}=Q^{n-2}_{\frac{n-x-3}2},Q^{n-2}_{\frac{n+x-3}2}=Q^{n-2}_{\frac{n-x-1}2}$ according to Observation~\ref{ob:symmetric}}
=&\bigg(\left(Q^{n-2}_{\frac{n-x-3}2}-Q^{n-2}_{\frac{n-x-1}2}\right)*(\frac{n-x-1}2+\alpha)+\left(Q^{n-2}_{\frac{n-x-3}2}-Q^{n-2}_{\frac{n-x-1}2}\right)*(\frac{n-x-1}2+\alpha)\bigg)\notag\\ \tag{$Q^{n-1}_{\frac{n+x-1}2}=Q^{n-1}_{\frac{n-x-1}2}$ according to Observation~\ref{ob:symmetric}}
&*\frac{(n+2\alpha-2)\harmo}{n+2\alpha-1}+Q^{n-1}_{\frac{n-x-1}2}+Q^{n-1}_{\frac{n-x-1}2}\notag\\
=&\left(Q^{n-2}_{\frac{n-x-3}2}-Q^{n-2}_{\frac{n-x-1}2}\right)*(n-x-1+2\alpha)*\frac{(n+2\alpha-2)\harmo}{n+2\alpha-1}+2Q^{n-1}_{\frac{n-x-1}2}\label{eq:symdif}
\end{align}

The remaining task is to analyze $Q^i_j$s in the above formula. They share some components which can be used for further simplification. We start from the uniform case and analyze it step by step.

\paragraph{Uniform Case} In this case  $\alpha=\beta=1$ and we can substitute $\alpha, \beta$ in Lemma~\ref{lem:generald} and obtains that the probability that Alice wins any number of rounds in the first $i$ rounds is equal, that is
\[
\Pr[\state_i=j]=\frac{1}{i+1}
\]

and $Q^i_j$ is
\[
Q^i_j=\frac{2(j+1)(i+1-j)}{(i+1)(i+2)^2}
\]

We substitute $Q^i_j$ into formula~\eqref{eq:symdif}

\begin{align*}
&\E[\Delta_\belset(x+1)]-\E[\Delta_\belset(x-1)]\\
=&\left(Q^{n-2}_{\frac{n-x-3}2}-Q^{n-2}_{\frac{n-x-1}2}\right)*(n-x-1+2\alpha)*\frac{(n+2\alpha-2)\harmo}{n+2\alpha-1}+2Q^{n-1}_{\frac{n-x-1}2}\tag{based on formula~\eqref{eq:symdif}}\\
=&\left(\frac{(n-x-1)(n+x+1)}{2(n-1)n^2}-\frac{(n-x+1)(n+x-1)}{2(n-1)n^2}\right)*(n-x+1)*\frac{n\harmo}{n+1}\\
&+\frac{(n-x+1)(n+x-1)}{n(n+1)^2}\\
=&\frac{(-4x)*(n-x+1)}{2(n-1)n^2}*\frac{n\harmo}{n+1}+\frac{(n-x+1)(n+x-1)}{n(n+1)^2}\\
=&\frac{(n-x+1)((n-1-2(1+n)\harmo)x+n^2-1)}{(n-1)n(n+1)^2}
\end{align*}

To find $\Tilde{x}$ that satisfies conditions in Claim~\ref{cla:roundopt}, we solve the equation \[\frac{(n-x+1)((n-1-2(1+n)\harmo)x+n^2-1)}{(n-1)n(n+1)^2}=0,\] and get \[x=\begin{cases}\frac{n^2-1}{2 (1 + n)\harmo-n+1}\\ n+1\end{cases}.\] 
Recall that $x\leq n$, we discard the solution of $x=n+1$, and then pick $\Tilde{x}:=\frac{n^2-1}{2 (1 + n)\harmo-n+1}=\frac{n-1}{2 \harmo-\frac{n-1}{n+1}}$. When $x<\Tilde{x}$, we have $(n-1-2(1+n)\harmo)x+n^2-1>0$, thus $\E[\Delta_\belset(x+1)]-\E[\Delta_\belset(x-1)]>0$. Otherwise, $\E[\Delta_\belset(x+1)]-\E[\Delta_\belset(x-1)]\leq 0$. Based on Claim~\ref{cla:roundopt}, the optimal bonus is \[x^*(1,1,n)=\round(\frac{n-1}{2 \harmo-\frac{n-1}{n+1}}).\] 

\paragraph{Symmetric Case}

Lemma~\ref{lem:generald} shows that if the prior $p$ follows $\betadis(\alpha,\alpha)$, then the probability of Alice wins $j$ rounds in the first $i$ rounds is
\[
\Pr[\state_i=j]=\frac{(2\alpha-1)\binom{i}{j}\binom{2\alpha-2}{\alpha-1}}{(2\alpha+i-1)\binom{i+2\alpha-2}{j+\alpha-1}}
\]
and $Q^i_j$ is
\[
Q^i_j=\frac{2(2\alpha-1)\binom{i}{j}\binom{2\alpha-2}{\alpha-1}}{(2\alpha+i)\binom{i+2\alpha}{j+\alpha}}
\]

Then we substitute $Q^i_j$ into formula~\eqref{eq:symdif}.

\begin{align*}
&\E[\Delta_\belset(x + 1)] - \E[\Delta_\belset(x - 1)]\\
=&\left(Q^{n-2}_{\frac{n-x-3}2}-Q^{n-2}_{\frac{n-x-1}2}\right)*(n-x-1+2\alpha)*\frac{(n+2\alpha-2)\harmo}{n+2\alpha-1}+2Q^{n-1}_{\frac{n-x-1}2}\tag{based on formula~\eqref{eq:symdif}}\\
=&\left(\frac{2(2\alpha-1)\binom{n-2}{\frac{n-x-3}2}\binom{2\alpha-2}{\alpha-1}}{(n+2\alpha-2)\binom{n+2\alpha-2}{\frac{n-x-3}2+\alpha}}-\frac{2(2\alpha-1)\binom{n-2}{\frac{n-x-1}2}\binom{2\alpha-2}{\alpha-1}}{(n+2\alpha-2)\binom{n+2\alpha-2}{\frac{n-x-1}2+\alpha}}\right)*(n-x-1+2\alpha)*\frac{(n+2\alpha-2)\harmo}{n+2\alpha-1}\\
&+\frac{4(2\alpha-1)\binom{n-1}{\frac{n-x-1}2}\binom{2\alpha-2}{\alpha-1}}{(n+2\alpha-1)\binom{n+2\alpha-1}{\frac{n-x-1}2+\alpha}}\tag{substitute $Q^i_j$}\\
=&\frac{2(2\alpha-1)\binom{n-1}{\frac{n-x-1}2}\binom{2\alpha-2}{\alpha-1}}{\binom{n+2\alpha-1}{\frac{n-x-1}2+\alpha}}\bigg(\frac{(n-x-1+2\alpha)\harmo}{n-1}*\left(\frac{n-x-1}{n-x-1+2\alpha}-\frac{n+x-1}{n+x-1+2\alpha}\right)+\frac{2}{n+2\alpha-1}\bigg)\\
\propto&\frac{(n-x-1+2\alpha)\harmo}{n-1}*\left(\frac{n-x-1}{n-x-1+2\alpha}-\frac{n+x-1}{n+x-1+2\alpha}\right)+\frac{2}{n+2\alpha-1}\tag{Since $\alpha\geq 1$, $2\alpha-1>0$ thus the coefficient is positive}\\
\propto& \harmo(n+2\alpha-1)\left((n-x-1)(n+x-1+2\alpha)-(n-x-1+2\alpha)(n+x-1)\right)+2(n-1)\tag{Since the denominators $n-1$, $n-x-1+2\alpha$, $n+x-1+2\alpha$ and $n+2\alpha-1$ are positive}\\
=& \harmo(n+2\alpha-1)(-4\alpha x)+2(n-1)(n+x-1+2\alpha)\\
=& (2(n-1)-4\alpha(n+2\alpha-1)\harmo)x+2(n-1)(n+2\alpha-1)=-bx + c
\end{align*}

We define $\Tilde{x}$ as the solution of $-bx+c=0$:

\[\Tilde{x}:=\frac{c}{b}=\frac{(n+2\alpha-1)(n-1)}{2\alpha(n+2\alpha-1)\harmo-n+1}=\frac{n-1}{2\alpha\harmo-\frac{n-1}{n+2\alpha-1}}.\]

Moreover, for all $\alpha\geq 1$, \begin{align*}
b =& 4\alpha(n+2\alpha-1)\harmo -2(n-1)\\
  =& 4\alpha(n+2\alpha-1)*(\sum_{i=1}^{n-1}\frac{1}{i+2\alpha-1})-2(n-1)\\
  \geq& 4\alpha(n+2\alpha-1)*\frac{n-1}{n+2\alpha-2}-2(n-1)\\
  >&4(n-1)-2(n-1)\\
  =&2(n-1)>0
\end{align*}

Therefore, based on Claim~\ref{cla:roundopt}, the optimal bonus $x^*(\alpha,\alpha,n)$ is 
\[
x^*(\alpha,\alpha,n)=\round(\Tilde{x})=\round(\frac{n-1}{2\alpha\harmo-\frac{n-1}{n+2\alpha-1}})
\]

\subsection{Certain Case}

In the certain case, $\alpha=\lambda p,\beta = \lambda(1-p), \lambda\rightarrow \infty$. Thus, when $n$ is finite, the winning probability of Alice is fixed to $p$ across all rounds. Note that we only consider $\alpha\geq\beta$ without loss of generality. Therefore, $p\geq \frac12$. The number of rounds Alice wins follows a binomial distribution. Formally, the probability of Alice wins $j$ rounds in the first $i$ rounds is
\[
\Pr[\state_i=j]=\binom{i}{j} p^j (1-p)^{i-j}
\]
Then we calculate $Q^i_j$
\begin{align}
Q^i_j=2\binom{i}{j}p^{j+1} (1-p)^{i-j+1} \label{eq:certainq}
\end{align}

Moreover, we can apply the Main Technical Lemma~\ref{lem:ratio} to show that the first $n-1$ rounds have the same expected surprise. 
\begin{corollary}
When $\alpha=\lambda p,\beta = \lambda(1-p), \lambda\rightarrow \infty$, given $n$ that is finite, \[
\E[\Delta_\belset]=\sum_{i=1}^{n}\E[\Delta_\belset^i]=\E[\Delta_\belset^{n-1}]*(n-1)+\E[\Delta_\belset^{n}]
\]
\label{cor:certain}
\end{corollary}

\begin{proof}[Proof of Corollary~\ref{cor:certain}]
\[
\frac{\E[\Delta_\belset^{i}]}{\E[\Delta_\belset^{i+1}]}=\frac{i+\alpha+\beta}{i+\alpha+\beta-1}\rightarrow 1\] as $\lambda\rightarrow \infty$.
\end{proof}

By formula~\eqref{eq:gensurp}, we have 

\begin{align}
\E[\Delta_\belset(x)]
=&\bigg(Q^{n-2}_{\frac{n-x-2}2}*\frac{\frac{n-x}2+\alpha}{n+\alpha+\beta-1}+Q^{n-2}_{\frac{n+x-2}2}*\frac{\frac{n-x}2+\beta}{n+\alpha+\beta-1}\notag\\
&+\sum_{j=\frac{n-x}2}^{\frac{n+x-4}2}Q^{n-2}_j*\frac{1}{n+\alpha+\beta-1}\bigg)*(n-1)+\sum_{j=\frac{n-x}2}^{\frac{n+x-2}2}Q^{n-1}_j \tag{Formula~\eqref{eq:gensurp} and Corollary~\ref{cor:certain}}\\

=&\bigg(Q^{n-2}_{\frac{n-x-2}2}*p+Q^{n-2}_{\frac{n+x-2}2}*(1-p)\bigg)*(n-1)+\sum_{j=\frac{n-x}2}^{\frac{n+x-2}2}Q^{n-1}_j\tag{$\alpha=\lambda p,\beta = \lambda(1-p), \lambda\rightarrow \infty$}\\
= & \left(2\binom{n-2}{\frac{n-\last-2}2}p^{\frac{n-\last+2}2}(1-p)^{\frac{n+\last}2}+2\binom{n-2}{\frac{n+\last-2}2}p^{\frac{n+\last}2}(1-p)^{\frac{n-\last+2}2}\right)*(n-1)\notag\\
&+\sum_{j=\frac{n-\last}2}^{\frac{n+\last-2}2}2\binom{n-1}{j}p^{j+1}(1-p)^{n-j}\tag{Apply formula~\eqref{eq:certainq} to substitute $Q^i_j$}
\end{align}

Then we calculate

\begin{align*}
& \E[\Delta_\belset(x + 1)] - \E[\Delta_\belset(x - 1)] \\
= & \underbrace{2(n-1)\bigg(\binom{n-2}{\frac{n-\last-3}2}p^{\frac{n-\last+1}2}(1-p)^{\frac{n+\last+1}2}+\binom{n-2}{\frac{n+\last-1}2}p^{\frac{n+\last+1}2}(1-p)^{\frac{n-\last+1}2}} \\
&\overbrace{-\binom{n-2}{\frac{n-\last-1}2}p^{\frac{n-\last+3}2}(1-p)^{\frac{n+\last-1}2}\bigg)-\binom{n-2}{\frac{n+\last-3}2}p^{\frac{n+\last-1}2}(1-p)^{\frac{n-\last+3}2}}^\text{difference of first n-1 round’s expected surprise} \\
&+\overbrace{\sum_{i=\frac{n-\last-1}2}^{\frac{n+\last-1}2}2\binom{n-1}{i}p^{i+1}(1-p)^{n-i}-\sum_{i=\frac{n-\last+1}2}^{\frac{n+\last-3}2}2\binom{n-1}{i}p^{i+1}(1-p)^{n-i}}^\text{difference of final round's expected surprise}\notag\\
= & (n-x-1)\left(\binom{n-1}{\frac{n-\last-1}2}p^{\frac{n-\last+1}2}(1-p)^{\frac{n+\last+1}2}\right) +
\binom{n-1}{\frac{n+\last-1}2}p^{\frac{n+\last+1}2}(1-p)^{\frac{n-\last+1}2}\\
& - (n+x-1)\left(\binom{n-1}{\frac{n-\last-1}2}p^{\frac{n-\last+3}2}(1-p)^{\frac{n+\last-1}2}+
\binom{n-1}{\frac{n+\last-1}2}p^{\frac{n+\last-1}2}(1-p)^{\frac{n-\last+3}2}\right)\\
& + 2\binom{n-1}{\frac{n+\last-1}2}p^{\frac{n+\last+1}2}(1-p)^{\frac{n-\last+1}2}+2\binom{n-1}{\frac{n-\last-1}2}p^{\frac{n-\last+1}2}(1-p)^{\frac{n+\last+1}2}\\
\tag{Apply $n\binom{n-1}{k}=(n-k)\binom{n}{k}$}\\
= & \binom{n-1}{\frac{n+\last-1}2}p^{\frac{n-\last+1}2}(1-p)^{\frac{n-\last+1}2}\left((n-x+1)((1-p)^x+p^x)-(n+x-1)(p(1-p)^{x-1}+p^{x-1}(1-p))\right)\\
= & \binom{n-1}{\frac{n+\last-1}2}p^{\frac{n-\last+1}2}(1-p)^{\frac{n-\last+1}2}\left((2np-n-(x-1))p^{x-1}+(n-2np-(x-1))(1-p)^{x-1}\right) \\
\propto & (2np-n-(x-1))p^{x-1}+(n-2np-(x-1))(1-p)^{x-1} \tag{Since $\binom{n-1}{\frac{n+\last-1}2}p^{\frac{n-\last+1}2}(1-p)^{\frac{n-\last+1}2}>0$}
\end{align*}

We define $F(x):=(2np-n-(x-1))p^{x-1}+(n-2np-(x-1))(1-p)^{x-1},x\in[1,n+1)$ and analyze the solution of $F(x)=0$. We show several examples of $F(x)$ with different $n$ in Figure~\ref{fig:F_example}. We can see that, intuitively, when $n$ is large enough, $F(x)=0$ has 2 solutions. One is $x=1$, and the other solution is close to $2np-n$ when $n$ is large.

\begin{figure}[!ht]\centering
  \subfigure[$n=5$]{\includegraphics[width=.24\linewidth]{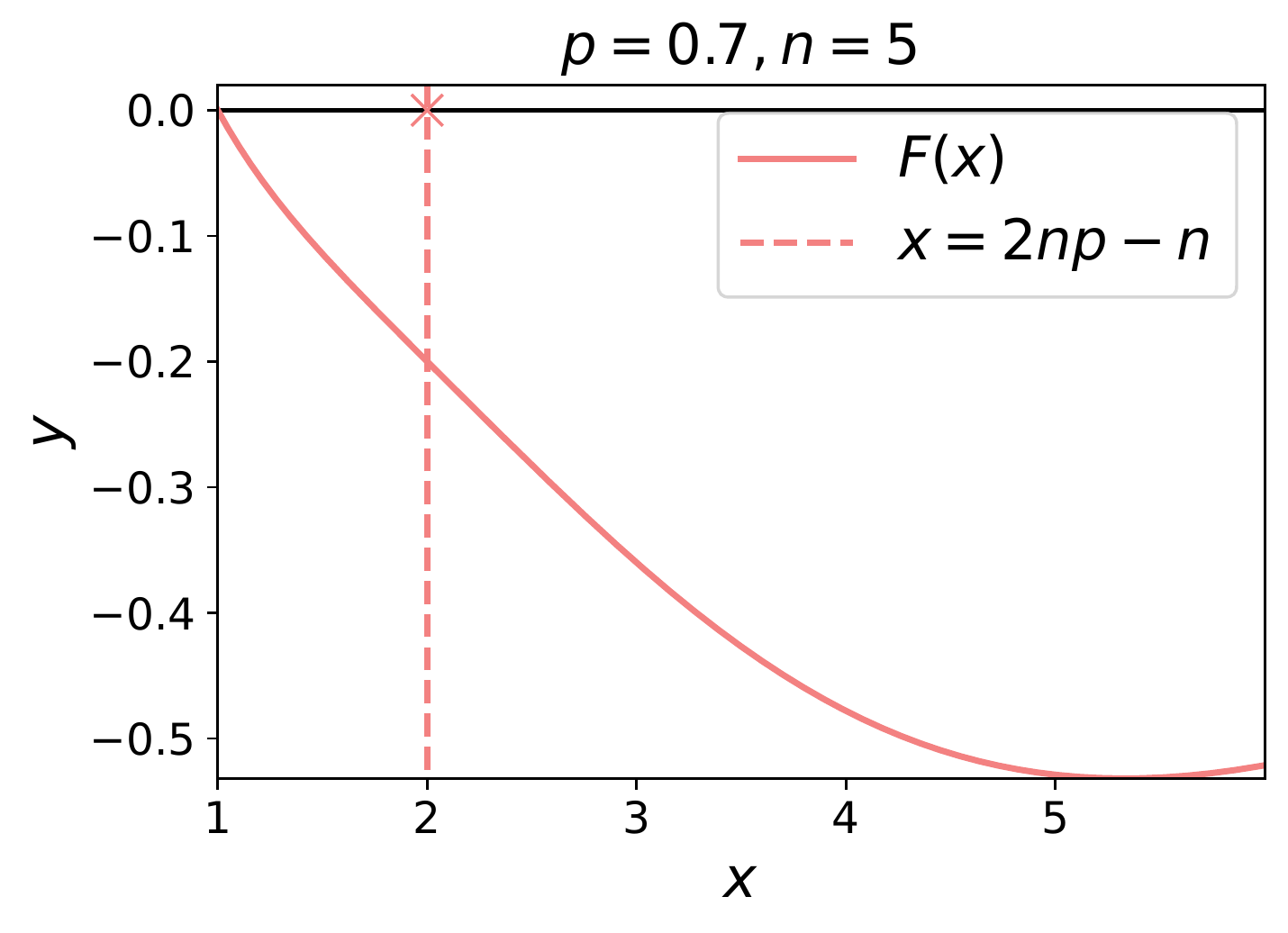}}
  \subfigure[$n=6$]{\includegraphics[width=.24\linewidth]{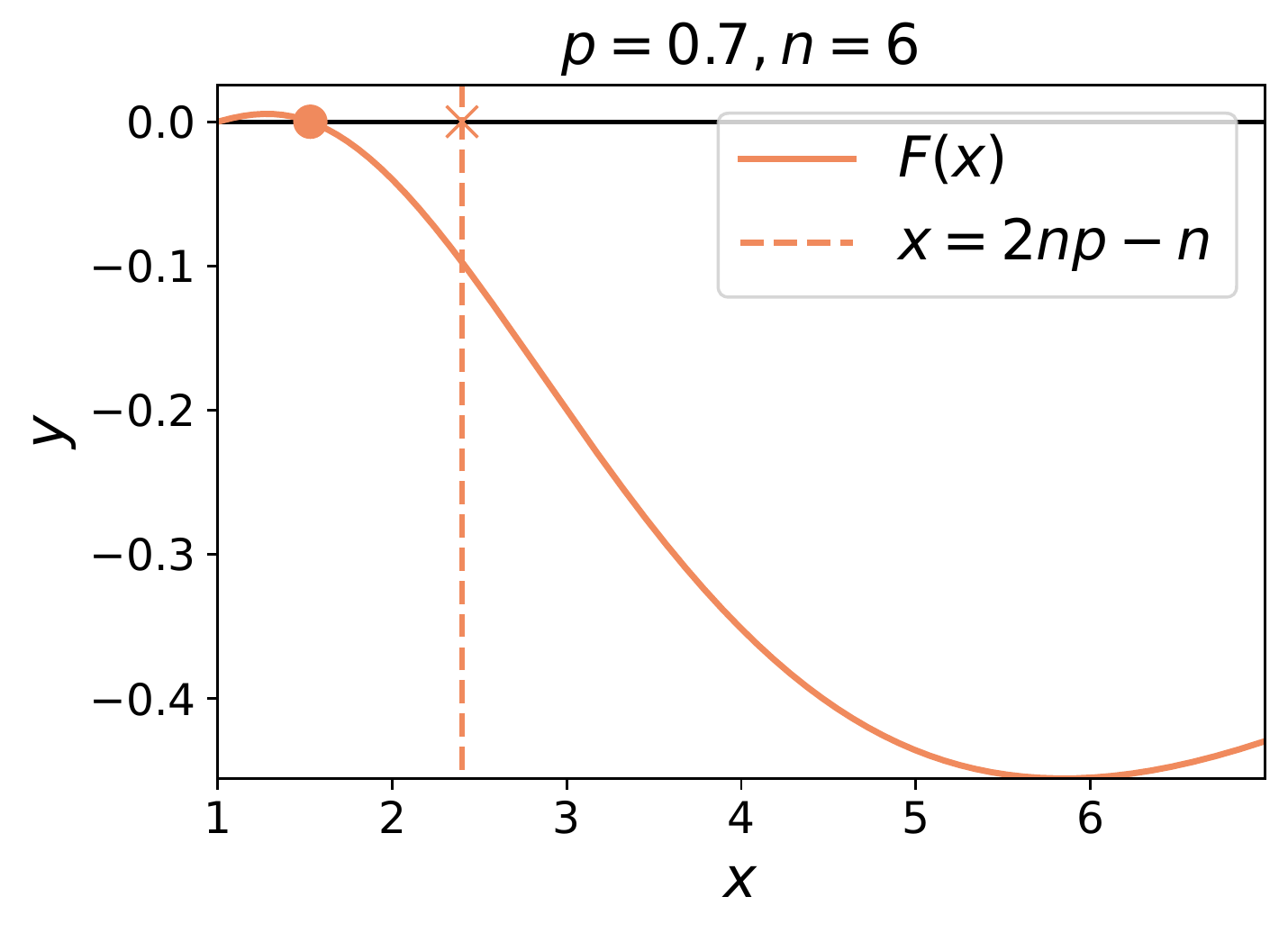}}
  \subfigure[$n=7$]{\includegraphics[width=.24\linewidth]{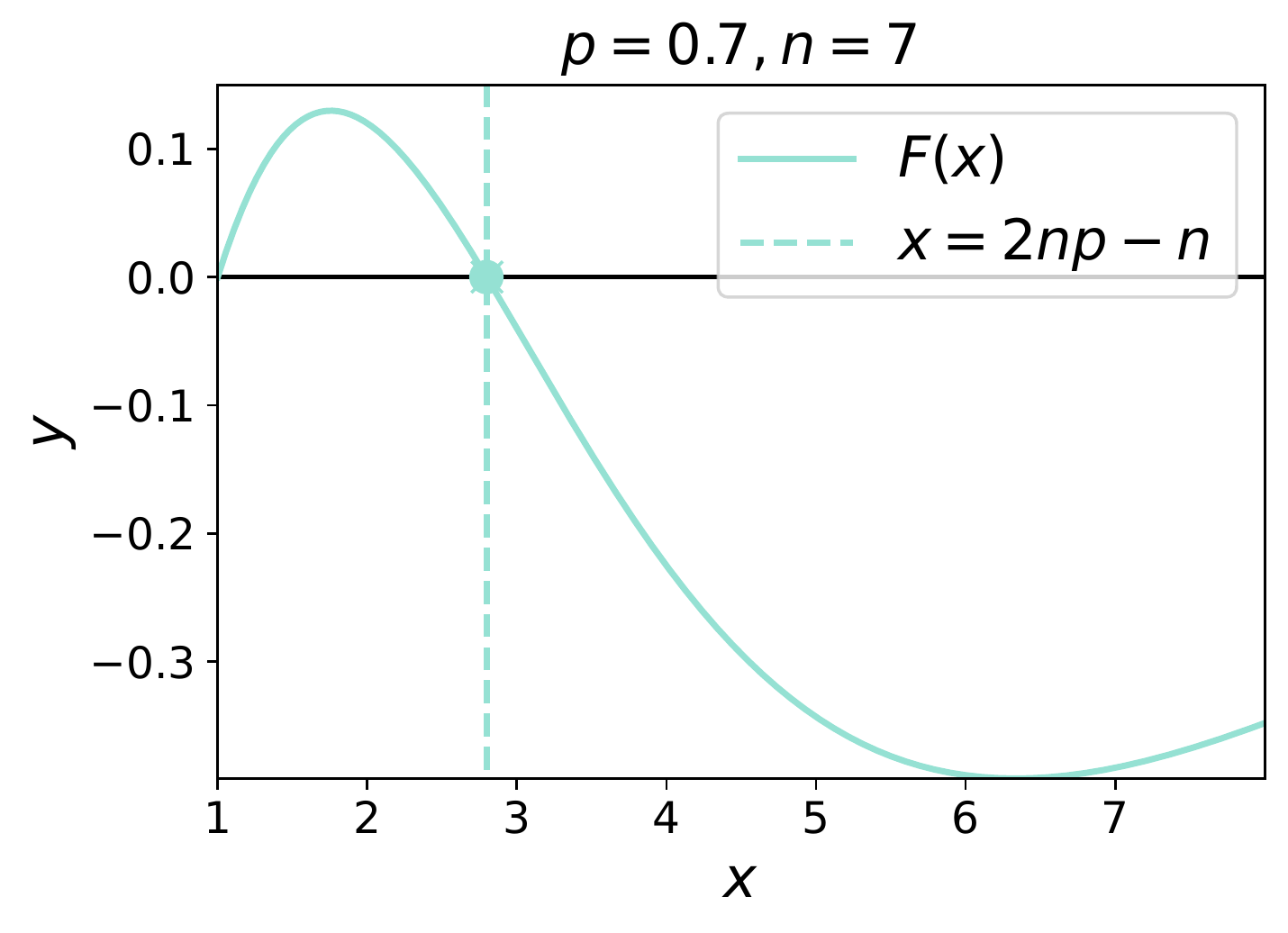}}
  \subfigure[$n=8$]{\includegraphics[width=.24\linewidth]{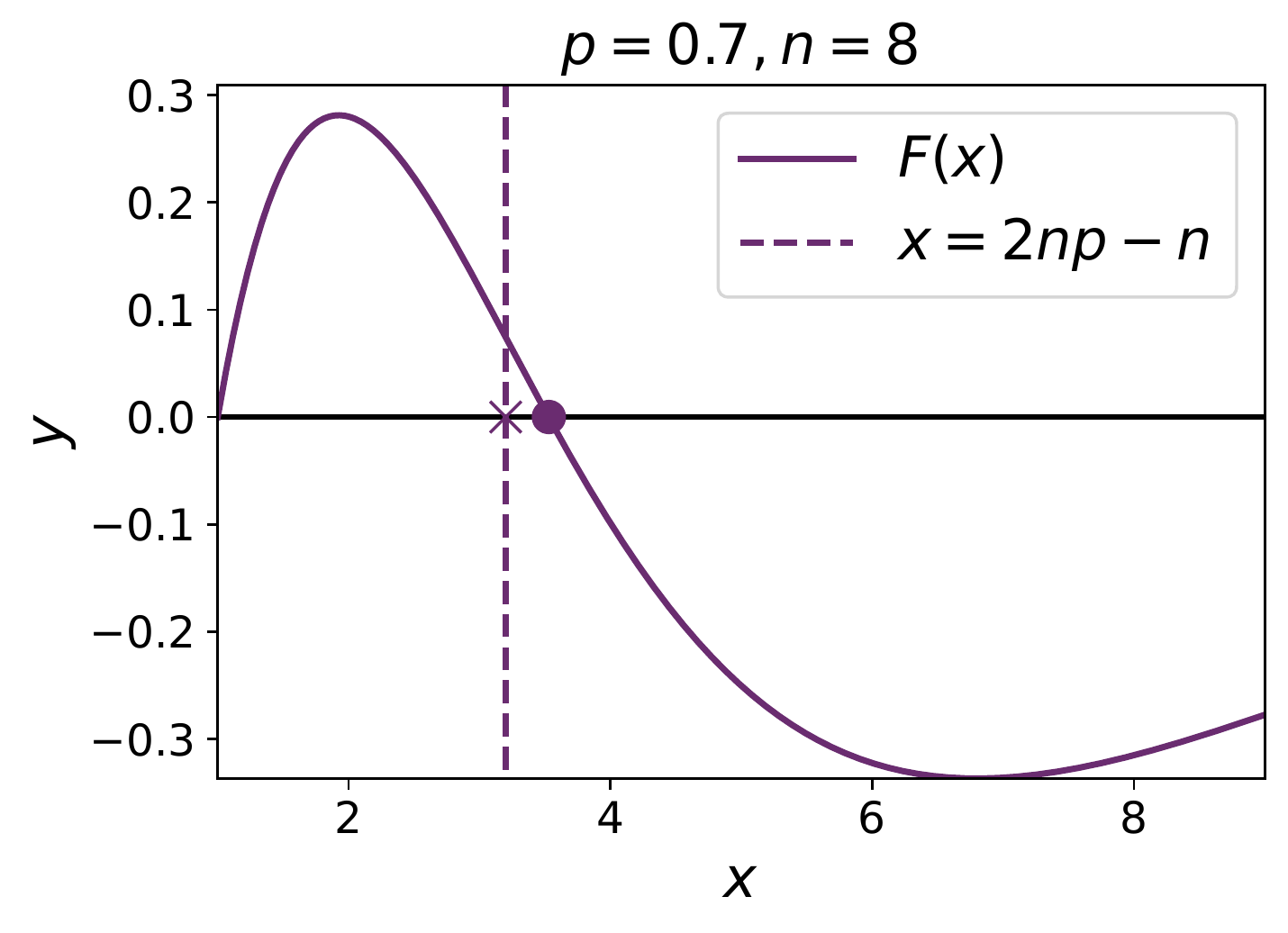}}
  \caption{\textbf{Examples for function F}}
  \label{fig:F_example}
\end{figure}

\begin{lemma}\label{lem:fx}
When $p\geq \frac12$, $F(x)=0$ has a trivial solution at $x=1$, and has a non-trivial solution $\Tilde{x}\in (1,2np-n+1)$ if and only if $p>\frac12$ and $n>\frac{1}{(\frac{1}{2}-p)\ln (\frac{1-p}{p})}$. There is no other solution. Moreover, when $x\in (1,\Tilde{x})$, $F(x)>0$, when $\Tilde{x}<n-1$ and $x\in (\Tilde{x},n-1]$, $F(x)<0$. 

Besides, let $a=2np-n-2$, when $p>\frac{1}{1+(a+1)^{-\frac{1}{a}}}$, the non-trivial solution of $F(x)=0$ is in $(2np-n-1,2np-n+1)$.
\end{lemma}

We defer the proof to appendix. $F(x)$ has a solution $\Tilde{x}$ in $(1,2np-n+1)$ when $p>\frac12$ and $n>\frac{1}{(\frac{1}{2}-p)\ln (\frac{1-p}{p})}$. When $x\in [1,\Tilde{x})$, $F(x)\geq 0$, i.e. $\E[\Delta_\belset(x + 1)] - \E[\Delta_\belset(x - 1)] \geq 0$. 
When $\Tilde{x}\leq n-1$ and $x\in[\Tilde{x},n-1]$, $\E[\Delta_\belset(x + 1)] - \E[\Delta_\belset(x - 1)]\leq 0$. Therefore, the conditions of Claim \ref{cla:roundopt} are satisfied. We apply Claim \ref{cla:roundopt} to obtain:
        \[
        x^*(\alpha,\beta,n) = \begin{cases}
        \round(\Tilde{x}) & \text{if $p>\frac12$ and $n>\frac{1}{(\frac{1}{2}-p)\ln(\frac{1-p}{p})}$}\\
        \round(1) & \text{otherwise}\\
        \end{cases}
        \]
        where $\Tilde{x}$ is the non-trivial solution of the equation $F(x)=0$.

Finally, we study the approximation of the $x^*(\alpha,\beta,n)$ in the certain case and show that under certain conditions, it is around the "expected lead" $\round(2np-n)$, the number of points the weaker player needs to comeback in expectation.

Lemma~\ref{lem:fx} shows that when $p>\frac{1}{1+(a+1)^{-\frac{1}{a}}}$, the non-trivial solution of $F(x)=0$ is in $(2np-n-1,2np-n+1)$. Recall that $\round(x)\in [x-1,x+1)$. Then we have $|\round(\Tilde{x})-\round(2np-n)|<3$. Moreover, $\round(x)$ is an integer that has the same parity as $n$. Therefore, $|\round(\Tilde{x})-\round(2np-n)|\leq 2$, i.e. the difference between the approximation $\round(2np-n)$ and the optimal bonus $\round(\Tilde{x})$ is $\leq 2$. 

\subsection{General Beta Prior Setting}

We provide an $O(n)$ algorithm for general Beta prior setting. A natural idea is to enumerate all possible bonus $x$ and calculate the corresponding expected total surprise value $\E[\Delta_\belset(x)]$. However, $\E[\Delta_\belset(x)]$ contains an $O(n)$ summation and lead to an $O(n^3)$ running time algorithm. 

Recall formula~\eqref{eq:gendif}, 
\begin{equation}
\label{eq:difference}
\begin{aligned}
    &\E[\Delta_\belset(x+1)]-\E[\Delta_\belset(x-1)]\\
    =&\bigg(\left(Q^{n-2}_{\frac{n-x-3}2}-Q^{n-2}_{\frac{n-x-1}2}\right)*(\frac{n-x-1}2+\alpha)+\left(Q^{n-2}_{\frac{n+x-1}2}-Q^{n-2}_{\frac{n+x-3}2}\right)*(\frac{n-x-1}2+\beta)\bigg)\notag\\
    &*\frac{(n+2\alpha-2)\harmo}{n+2\alpha-1}+Q^{n-1}_{\frac{n+x-1}2}+Q^{n-1}_{\frac{n-x-1}2}
\end{aligned}
\end{equation}
Note that in this formula, only $\harmo$ and $Q^i_j$ cannot be calculated in O(1) time. We can preprocess $\harmo$ in the time of $O(n)$. Then notice that if we can calculate all the $Q^i_j$ in the formula within $O(1)$ time, we can compute the difference between $\E[\Delta_\belset(x+1)]$ and $\E[\Delta_\belset(x-1)]$ in $O(1)$ time. 

Recall Lemma \ref{lem:generald}

\begin{align*}
Q^i_j&=\frac{2(\alpha+\beta-1)\binom{i}{j}\binom{\alpha+\beta-2}{\alpha-1}}{(\alpha+\beta+i)\binom{i+\alpha+\beta}{j+\alpha}}\\
&=\frac{\Gamma(i+1)\Gamma(\alpha+\beta)\Gamma(j+\alpha+1)\Gamma(i-j+\beta+1)}{(\alpha+\beta+i) \Gamma(i+\alpha+\beta+1)\Gamma(j+1)\Gamma(i-j+1)\Gamma(\alpha)\Gamma(\beta)}\\
&=\frac{1}{\alpha+\beta+i}*\frac{\Gamma(i+1)}{\Gamma(j+1)\Gamma(i-j+1)}*\frac{\Gamma(j+\alpha+1)}{\Gamma(\alpha)}*\frac{\Gamma(i-j+\beta+1)}{\Gamma(\beta)}*\frac{\Gamma(\alpha+\beta)}{\Gamma(i+\alpha+\beta+1)}
\end{align*}

We consider the following four parts separately 
\[
\begin{cases}
\Gamma(1+y)\\
\frac{\Gamma(\alpha+y)}{\Gamma(\alpha)}\\
\frac{\Gamma(\beta+y)}{\Gamma(\beta)}\\
\frac{\Gamma(\alpha+\beta)}{\Gamma(\alpha+\beta+y)}\\
\end{cases}
y\in \{0, 1, 2, \ldots, n\}
\]

Due to the property of the Gamma function $\forall z> 0, \Gamma(z+1)=z\Gamma(z)$, we can use the recursive method to preprocess the above four parts in $O(n)$ time. Then for any $i,j\in\{0, 1, 2, \ldots, n\},j\leq i$, we can calculate the value of $Q_j^i$ in $O(1)$ time.

Based on this, when we enumerate all possible bonus $x$ in ascending order, we can calculate the value of $\E[\Delta_\belset(x+1)]$ based on $\E[\Delta_\belset(x-1)]$ in $O(1)$ time. We then find the optimal bonus $x$. We present the pseudo code in Algorithm~\ref{algo:optimal_x}. The total time and space complexity of the algorithm is $O(n)$.

\begin{algorithm}[!ht]
\SetAlgoNoLine
\caption{Calculate optimal bonus $x^*$}
\label{algo:optimal_x}
    \KwIn{Number of rounds $n$, the parameters of the prior Beta distribution  $\alpha,\beta$}
    \KwOut{Optimal bonus $x^*$}
    \For{$i$ in $\{0, 1,\ldots,n\}$}{
    Initialize $\Gamma(i+1),\frac{\Gamma(\alpha+i)}{\Gamma(\alpha)},\frac{\Gamma(\beta+i)}{\Gamma(\beta)},\frac{\Gamma(\alpha+\beta)}{\Gamma(\alpha+\beta+i)}$}
    $surp\_max:=0$\\
    $surp\_sum:=0$\\
    $x:= n\%2$\\
    \For{$i$ in $\{n\%2+1,\ldots,n-1\}$} {
        $surp\_sum += \E[\Delta_\belset(i+1)]-\E[\Delta_\belset(i-1)]$\tcc*{Formula~\eqref{eq:gendif}}
        \If{$surp\_sum>surp\_max$}{
            $surp\_max:=surp\_sum$\\
            $x:= i + 1$
        }
    }
    \Return $x^*=x$
\end{algorithm}

We implement the algorithm to conduct numerical experiments for $n=5, n=10, n=20, n=40$ (Figure~\ref{fig:contourf_finite}) and the figure becomes more and more similar to the asymptotic case (Figure~\ref{fig:infinite}) as $n$ increases. 

\section{Asymptotic Case}

In the asymptotic case, we can use a continuous integral to approximate the discrete summation. Here we define bonus ratio $\mu:=\frac{x}{n}$ and use the integration of $\mu$ to approximate the overall surprise. Formally,

\begin{theorem}
For all $\alpha\geq\beta\geq 1$, there exists a function $\asyf_{\alpha,\beta,n}(\mu)$ such that $\forall \mu\in (0,1)$, $\E[\Delta_{\mathcal{B}}(\mu*n)]=\asyf_{\alpha,\beta,n}(\mu)*(1+O(\frac1n))$. When we define $\mu^*:=\arg\max_{\mu} \asyf_{\alpha,\beta,n}(\mu)$, 
\begin{itemize}
    \item \textbf{Symmetric $\alpha=\beta$} 
    \[\mu^*= \frac{1}{2\alpha \harmo_{2\alpha}(n-1)-1}\]
    \item \textbf{Near-certain $\alpha=\lambda p, \beta=\lambda (1-p)$} fixing $p$, for all sufficiently small $\epsilon>0$, when $\lambda>O(\log\frac{1}{\epsilon})$\footnote{See detailed conditions in Lemma~\ref{lem:gu}}, the optimal $\mu^*$ is around the ``expected lead'', 
    \[\mu^*\in (\frac{(\alpha-\beta)\harmo+1}{(\alpha+\beta)\harmo-1}-\epsilon,\frac{(\alpha-\beta)\harmo+1}{(\alpha+\beta)\harmo-1})\approx (\frac{\alpha-\beta}{\alpha+\beta}-\epsilon,\frac{\alpha-\beta}{\alpha+\beta})=(2p-1-\epsilon, 2p-1)\]
    \item \textbf{General} $\mu^*$ is the unique solution of $G(\mu)=0$ and $\mu^*<\frac{(\alpha-\beta)\harmo+1}{(\alpha+\beta)\harmo-1}$.
\end{itemize}
\[G(\mu):=(1+\mu)^{\alpha-\beta}\left(\frac{(\alpha-\beta)\harmo+1}{(\alpha+\beta)\harmo-1}-\mu\right)+(1-\mu)^{\alpha-\beta}\left(\frac{(-\alpha+\beta)\harmo+1}{(\alpha+\beta)\harmo-1}-\mu\right)\]
\label{the:asymptotic}
\end{theorem}

\begin{figure}[!ht]\centering
  \subfigure[$\alpha-\beta=0$]{\includegraphics[width=.24\linewidth]{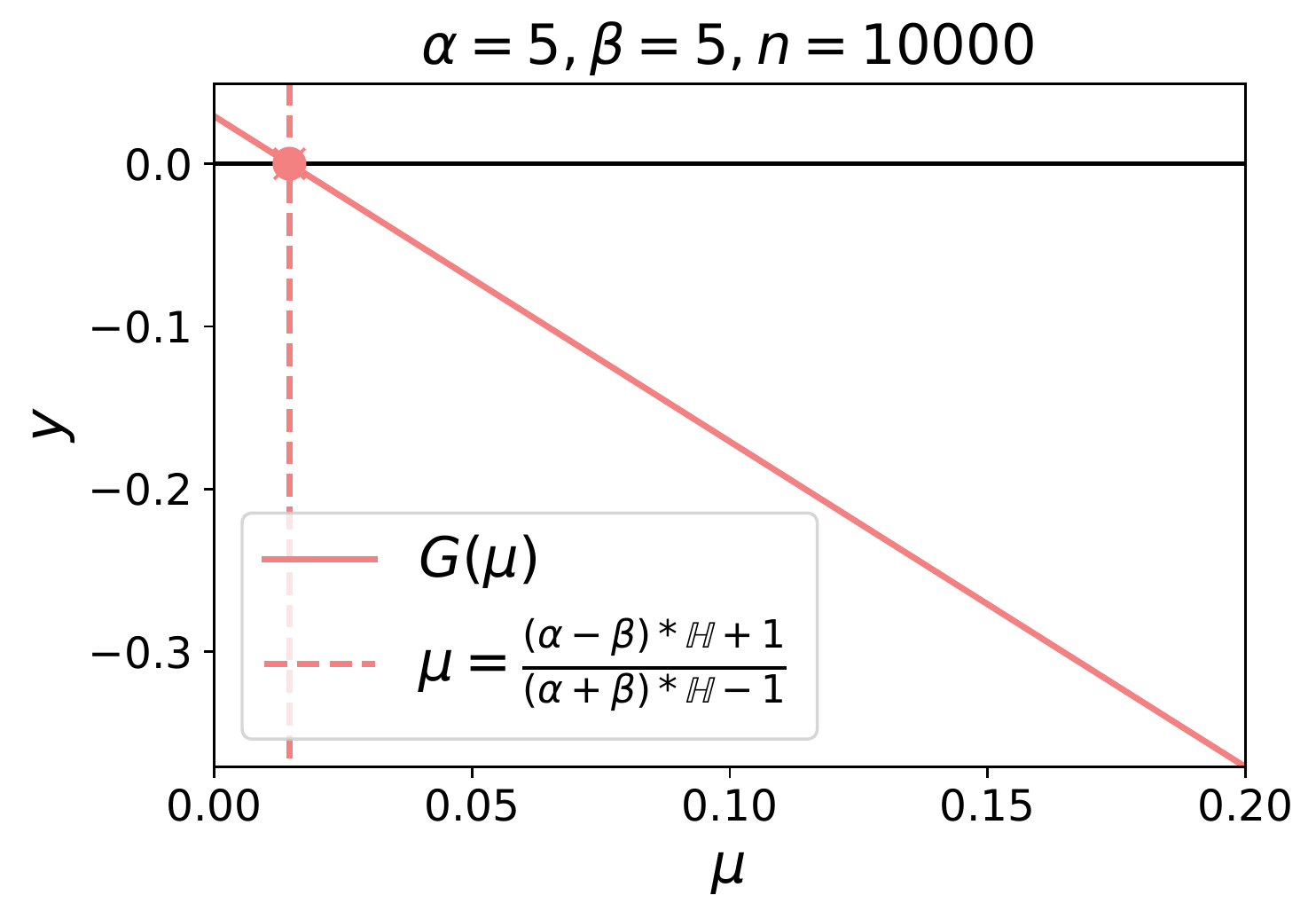}}
  \subfigure[$\alpha-\beta=2$]{\includegraphics[width=.24\linewidth]{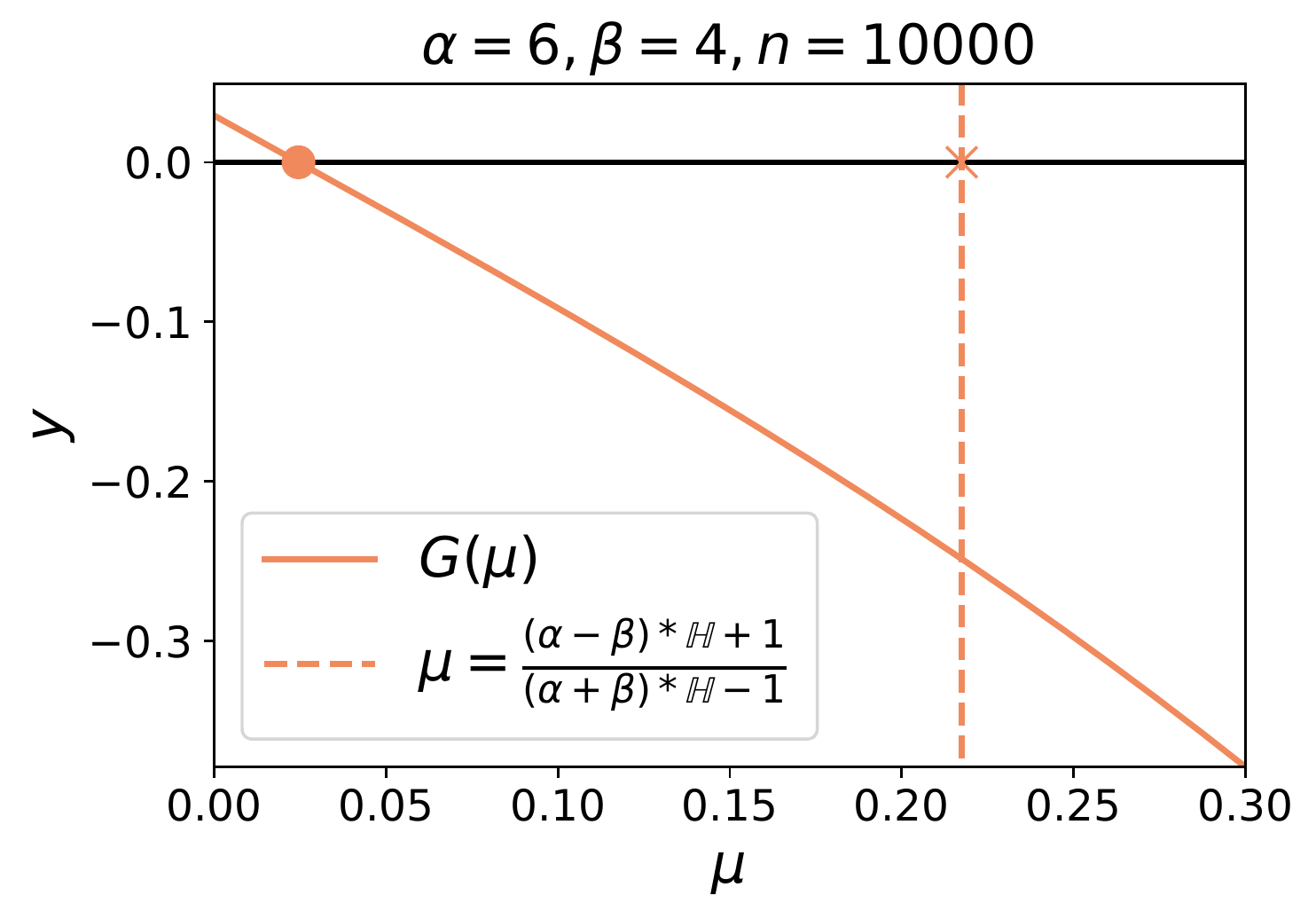}}
  \subfigure[$\alpha-\beta=4$]{\includegraphics[width=.24\linewidth]{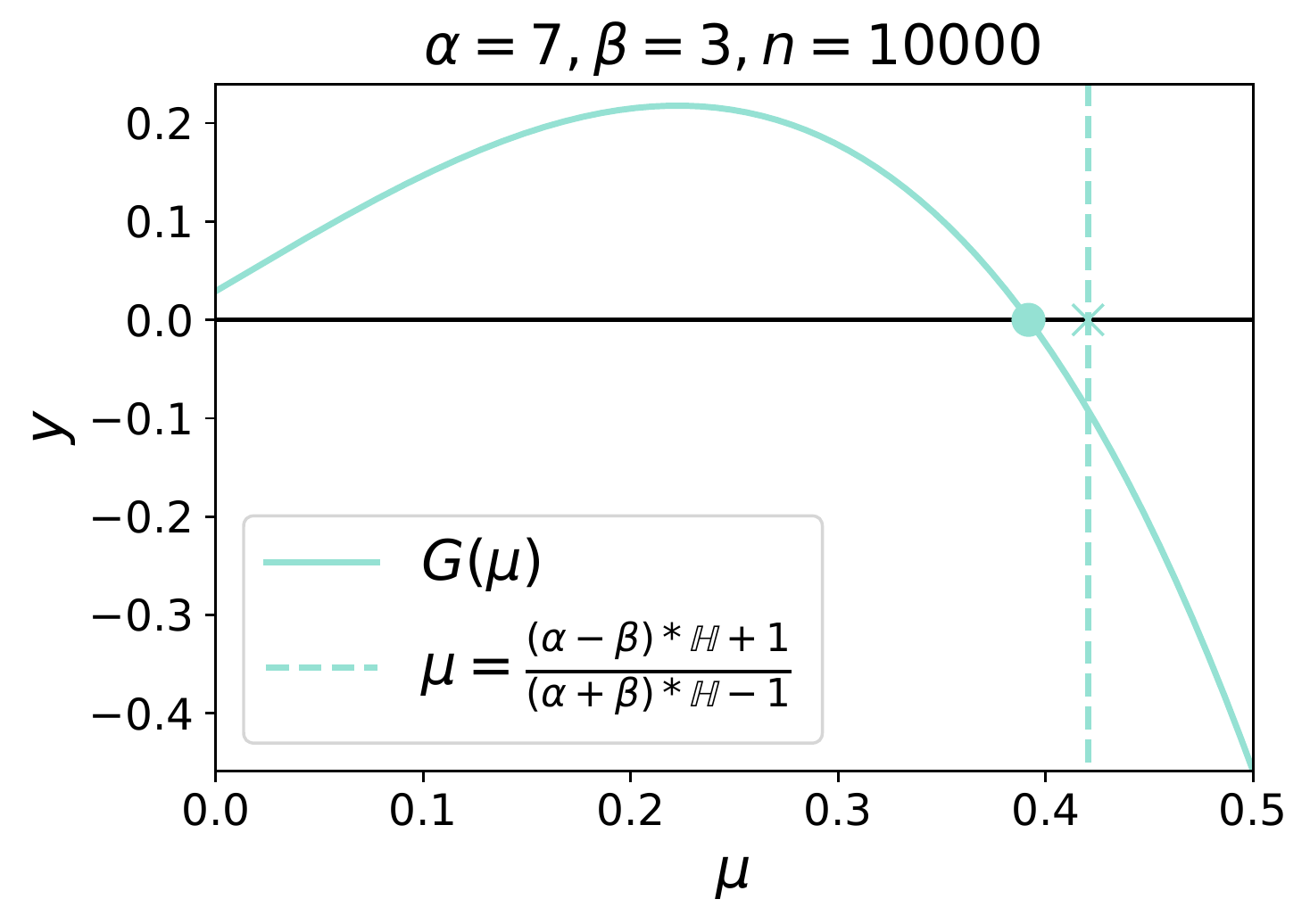}}
  \subfigure[$\alpha-\beta=6$]{\includegraphics[width=.24\linewidth]{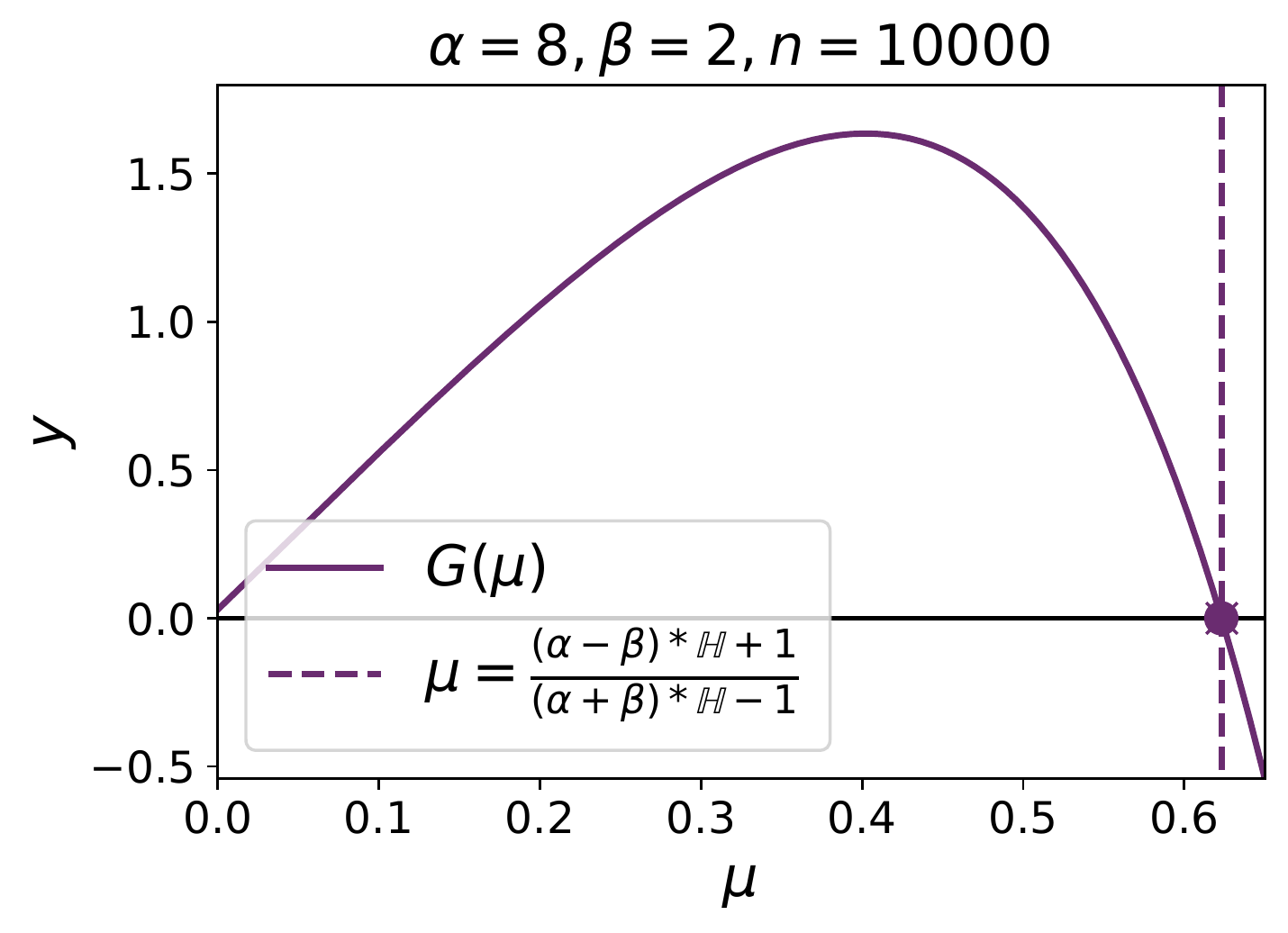}}
  \caption{\textbf{Examples for function G}}
  \label{fig:G_example}
\end{figure}

\begin{figure}[!ht]\centering
  \includegraphics[width=1\linewidth]{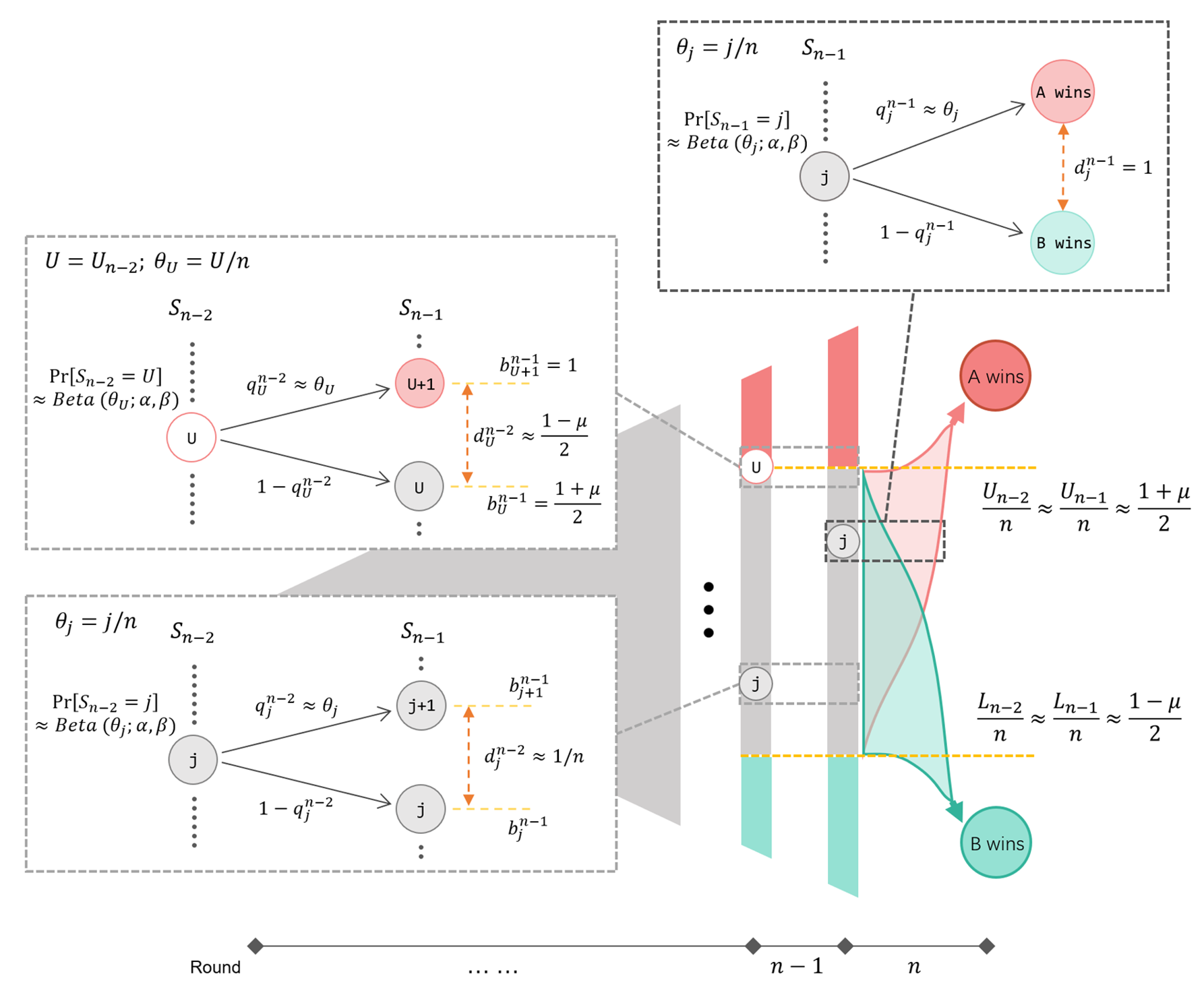}
  \caption{\textbf{Asymptotic case}}
  \label{fig:asymptotic}
\end{figure}

We define $\betadens{\theta}$ as the density function of Beta distribution, i.e. $\betadens{\theta}=\frac1{B(\alpha,\beta)}\theta^{\alpha-1}(1-\theta)^{\beta-1}$.\footnote{$B(\alpha,\beta)$ is beta function, i.e.  $B(\alpha,\beta)=\frac{\Gamma(\alpha)\Gamma(\beta)}{\Gamma(\alpha+\beta)}$} In the asymptotic case when $n$ is sufficiently large, we can simplify Lemma~\ref{lem:generald} to obtain Lemma~\ref{lem:asymtotic}, illustrated in Figure~\ref{fig:asymptotic}. The simplification for $d,L,U$ is straightforward. For $\Pr[\state_{n-1}=j]$, note that informally, due to law of large number, when $n$ is sufficiently large, the state in the $n-1$ or $n-2$ round concentrate to $p n$. Thus, given that $p$ follows distribution $\betadens{\theta}$, we have $\Pr[\state_{n-1}=j]\approx \Pr[\state_{n-2}=j]$ which are approximately proportional to $\betadens{\theta_j}$, where $\theta_j=\frac{j}n$

\begin{lemma}[Informal]
When $n$ is sufficiently large, fixed $\alpha,\beta,\mu\in(0,1)$, we have
\[
\begin{cases}
\frac{L_{n-2}}{n}\approx \frac{L_{n-1}}{n}\approx \frac{1-\mu}2\\
\frac{U_{n-2}}{n}\approx \frac{U_{n-1}}{n}\approx \frac{1+\mu}2\\
\end{cases}
\]
For any $L_{n-1}\leq j\leq U_{n-1}$, let $\theta_j=\frac{j}{n}$, we have
\[
\Pr[\state_{n-1}=j]\approx \Pr[\state_{n-2}=j]\approx \frac{\betadens{\theta_j}}n
\]
\[
q^{n-1}_j\approx q^{n-2}_j\approx \theta_j
\]
and
\[
d^{n-2}_j\approx
\begin{cases}
0,&\theta_j<\frac{L_{n-2}}{n}\\
\frac{1-\mu}2,&\theta_j=\frac{L_{n-2}}{n}\\
\frac{1}{n},&\frac{L_{n-2}}{n}<\theta_j<\frac{U_{n-2}}{n}\\
\frac{1-\mu}2,&\theta_j=\frac{U_{n-2}}{n}\\
0,&\theta_j>\frac{U_{n-2}}{n}
\end{cases}
\qquad
d^{n-1}_j=
\begin{cases}
0,&\theta_j<\frac{L_{n-1}}n\\
1,&\frac{L_{n-1}}n\leq \theta_j\leq\frac{U_{n-1}}n\\
0,&\theta_j>\frac{U_{n-1}}n
\end{cases}
\]
\label{lem:asymtotic}
\end{lemma}

Recall the general formula~\eqref{eq:general} here:

\begin{align*}
\E[\Delta_\belset(x)]   =& \bigg(\overbrace{\Pr[\state_{n-2}=L_{n-2}]*2 q^{n-2}_{L_{n-2}}*(1-q^{n-2}_{L_{n-2}})*d_{L_{n-2}}^{n-2}}^\mytexta\notag\\

&+\overbrace{\Pr[\state_{n-2}=U_{n-2}]*2q^{n-2}_{U_{n-2}}*(1-q^{n-2}_{U_{n-2}})*d_{U_{n-2}}^{n-2}}^\mytextb\notag\\

&+\overbrace{\sum_{j=L_{n-2}+1}^{U_{n-2}-1}\Pr[\state_{n-2}=j]*2q^{n-2}_{j}*(1-q^{n-2}_{j})*d_{j}^{n-2}}^\mytextm\bigg)*(n+\alpha+\beta-2)*\harmo\notag\\

&+\overbrace{\sum_{j=L_{n-1}}^{U_{n-1}}\Pr[\state_{n-1}=j]*2q^{n-1}_{j}*(1-q^{n-1}_{j})*d_{j}^{n-1}}^\mytextn 
\end{align*}

By substituting Lemma~\ref{lem:asymtotic} into formula~\eqref{eq:general}, we can obtain an approximation for the overall expected surprise $\E[\Delta_{\mathcal{B}}(\mu*n)]$:

\begin{align}
&\asyf_{\alpha,\beta,n}(\mu)\notag\\
:=&\bigg(\overbrace{\frac{\betadens{\frac{1-\mu}{2}}}n*2*\frac{1+\mu}{2}*\frac{1-\mu}{2}*\frac{1-\mu}{2}}^\mytexta+\overbrace{\int_{\frac{1-\mu}2}^{\frac{1+\mu}2}\betadens{\theta}*2*\theta*(1-\theta)*\frac{1}{n}d\theta}^\mytextm\notag\\
                    &+\underbrace{\frac{\betadens{\frac{1+\mu}{2}})}n*2*\frac{1+\mu}{2}*\frac{1-\mu}{2}*\frac{1-\mu}{2}}_\mytextb\bigg)*n\harmo+\underbrace{\int_{\frac{1-\mu}2}^{\frac{1+\mu}2}\betadens{\theta}*2*\theta*(1-\theta)d\theta}_\mytextn\notag\\
                    =&\left(\int_{\frac{1-\mu}2}^{\frac{1+\mu}2}2\betadens{\theta}\theta(1-\theta)d\theta\right)*(\harmo+1)\\
                    &+\left(\betadens{\frac{1+\mu}{2}})+\betadens{\frac{1-\mu}{2}})\right)\frac{(1-\mu)^2(1+\mu)\harmo}{4}\notag
\end{align}

Formally, we need a formal version of Lemma~\ref{lem:asymtotic} to delicately analyze the relationship between $\asyf_{\alpha,\beta,n}(\mu)$ and $\E[\Delta_{\mathcal{B}}(\mu*n)]$ and prove the main result of Theorem~\ref{the:asymptotic}. We defer the formal version of Lemma~\ref{lem:asymtotic} and the proof to appendix.

We then analyze the property of $\asyf_{\alpha,\beta,n}(\mu)$.

In order to find the optimal $\mu$, we calculate the derivation of $\asyf_{\alpha,\beta,n}(\mu)$
\begin{align*}
&\frac{d\asyf_{\alpha,\beta,n}(\mu)}{d\mu}\\ =&\left(\betadens{\frac{1+\mu}{2}}+\betadens{\frac{1-\mu}{2}}\right)\frac{1-\mu^2}{4}*(\harmo+1)\\ 
                    &+\frac{d\left(\betadens{\frac{1+\mu}{2}}+\betadens{\frac{1-\mu}{2}}\right)}{d\mu}*\frac{(1-\mu)^2(1+\mu)\harmo}{4}\\
                    &+\left(\betadens{\frac{1+\mu}{2}}+\betadens{\frac{1-\mu}{2}}\right)*\frac{d((1-\mu)^2(1+\mu))}{d\mu}*\frac{\harmo}{4}\\
                    =&\left(\betadens{\frac{1+\mu}{2}}+\betadens{\frac{1-\mu}{2}}\right)\frac{1-\mu^2}{4}*(\harmo+1)\\ 
                    &+\frac{\left((\alpha-\beta-(\alpha+\beta-2)\mu)\betadens{\frac{1+\mu}{2}}+(-\alpha+\beta-(\alpha+\beta-2) \mu )\betadens{\frac{1-\mu}{2}}\right)(1-\mu)\harmo}{4}\\ 
                    &+\left(\betadens{\frac{1+\mu}{2}}+\betadens{\frac{1-\mu}{2}}\right)\frac{(1-\mu)(-3\mu-1)\harmo}{4}\\
                    =&\frac{(1-\mu)\betadens{\frac{1+\mu}{2}}}{4}\left(-((\alpha+\beta)\harmo-1)\mu+(\alpha-\beta)\harmo+1\right)\tag{Combining like terms}\\
                    &+\frac{(1-\mu)\betadens{\frac{1-\mu}{2}}}{4}\left(-((\alpha+\beta)\harmo-1)\mu+(-\alpha+\beta)\harmo+1\right)\\
                    \propto & (1+\mu)^{\alpha-1}*(1-\mu)^{\beta-1}\left(-((\alpha+\beta)\harmo-1)\mu+(\alpha-\beta)\harmo+1\right)\tag{Substituting density of Beta distribution}\\
                    &+(1-\mu)^{\alpha-1}*(1+\mu)^{\beta-1}\left(-((\alpha+\beta)\harmo-1)\mu+(-\alpha+\beta)\harmo+1\right)\\ \tag{$\mu<1$}
                    \propto & (1+\mu)^{\alpha-\beta}\left(-((\alpha+\beta)\harmo-1)\mu+(\alpha-\beta)\harmo+1\right)\\
                    &+(1-\mu)^{\alpha-\beta}\left(-((\alpha+\beta)\harmo-1)\mu+(-\alpha+\beta)\harmo+1\right)\\
                    \propto & (1+\mu)^{\alpha-\beta}\left(\frac{(\alpha-\beta)\harmo+1}{(\alpha+\beta)\harmo-1}-\mu\right)+(1-\mu)^{\alpha-\beta}\left(\frac{(-\alpha+\beta)\harmo+1}{(\alpha+\beta)\harmo-1}-\mu\right)\\
\end{align*}

Let \[G(\mu):=(1+\mu)^{\alpha-\beta}\left(\frac{(\alpha-\beta)\harmo+1}{(\alpha+\beta)\harmo-1}-\mu\right)+(1-\mu)^{\alpha-\beta}\left(\frac{(-\alpha+\beta)\harmo+1}{(\alpha+\beta)\harmo-1}-\mu\right)\] 

The examples of $G(\mu)$ are illustrated in Figure~\ref{fig:G_example}. 

\begin{lemma}[Property of $G(\mu)$]\label{lem:gu}
For all $\alpha\geq \beta$, when $n$ is sufficiently large, $G(0)>0$ and $G(1)<0$. $G(\mu)=0,\mu\in[0,1]$ has a unique solution and the solution is in $(0,\frac{(\alpha-\beta)\harmo+1}{(\alpha+\beta)\harmo-1})$. 

Moreover, for all $0<\epsilon< \frac{(\alpha-\beta)\harmo+1}{(\alpha+\beta)\harmo-1}$, when $\alpha-\beta>\frac{\log (\frac{2(\alpha-\beta)\harmo}{(\alpha+\beta)\harmo-1}-\epsilon)-\log \epsilon}{\log (1+\frac{(\alpha-\beta)\harmo+1}{(\alpha+\beta)\harmo-1}-\epsilon)-\log (1-(\frac{(\alpha-\beta)\harmo+1}{(\alpha+\beta)\harmo-1}-\epsilon))) } \footnote {This is approximately $\frac{\log\frac{\alpha-\beta}{\alpha+\beta}+\log\frac{1}{\epsilon}}{\log \alpha-\log \beta }<\frac{1}{\log \alpha-\log \beta } \log\frac{1}{\epsilon} $}$, the solution is within $(\frac{(\alpha-\beta)\harmo+1}{(\alpha+\beta)\harmo-1}-\epsilon,\frac{(\alpha-\beta)\harmo+1}{(\alpha+\beta)\harmo-1})\approx (\frac{\alpha-\beta}{\alpha+\beta}-\epsilon, \frac{\alpha-\beta}{\alpha+\beta})$. 
\end{lemma}

We defer the proof to appendix. The above lemma implies that $\asyf_{\alpha,\beta,n}(\mu)$ first increases and then decreases. Thus, the global optimal $\mu^*$ is also the local maximum, $G(\mu)=0$'s unique solution. Thus, the results for general case in Theorem~\ref{the:asymptotic} follow from the lemma. In the symmetric case, $\alpha=\beta$, \begin{align*}
G(\mu)&=\frac{2}{(\alpha+\beta)\harmo-1}-2\mu=0\\
\Rightarrow \mu^*&=\frac{1}{2\alpha \harmo-1}.
\end{align*} The results for near certain case also directly follow from the above lemma. Therefore, we finish the proof. 

\section{Conclusion and Discussion}

In a multi-round competition, we show that we can increase the audience's overall surprise by setting a proper bonus in the final round. We further show that the optimal bonus size depends on the audience's prior and in the following settings, we obtain solutions of various forms for both the case of a finite number of rounds and the asymptotic case: 
\begin{description}
\item [Symmetric] the audience's prior belief does not lean towards any player, here we obtain a clean closed-form solution in both the finite the and asymptotic case; 
\item [Certain] the audience is a priori certain about the two players' relative abilities, here the optimal bonus is a special function's solution, and approximately and asymptotically equal to the ``expected lead'', the number of points the weaker player will need to come back in expectation;
\item [General] the optimal bonus can be obtained by a linear algorithm and, in the asymptotic case, is a special function's solution. 
\end{description}

One natural extension is to validate our theoretical predictions using field experiments. We can potentially conduct online field experiments, e.g., AB test, to examine the effectiveness of the scoring rules. Moreover, the results from field experiments could potentially capture features that might be neglected in the existing models and consequently inform the development of new theories.

Regarding the theoretical work, one future direction is to incorporate the time factor into the model since a line of psychology literature \cite{kahneman1993more,baddeley1993recency} shows that the audience may judge their experience largely on their feeling in the later part of the game (end-effect). We can generalize our analysis to factor in that surprise may be more valued in some time periods (e.g. the last round) than others. Another future direction is extending our results to the setting where we allow the scores to increase from the first round to the final round (e.g., 1,2,3,4,.....), or the bonus size depends on the results of previous rounds, or the setting where the game may not end in a pre-determined fixed time (e.g. Quidditch), or even a more complicated game setting (e.g. board games). 

\newpage
\bibliographystyle{plainnat}
\bibliography{main.bbl}
\appendix\section{Proof of Lemma \ref{lem:generald}}
{
\renewcommand{\thetheorem}{\ref{lem:generald}}
\begin{lemma}
We have \[
q^{i}_j=\frac{j+\alpha}{i+\alpha+\beta}
\] 
\[
\Pr[\state_i=j]=\frac{(\alpha+\beta-1)\binom{i}{j}\binom{\alpha+\beta-2}{\alpha-1}}{(\alpha+\beta+i-1)\binom{i+\alpha+\beta-2}{j+\alpha-1}}
\] 
\[
Q^i_j=\frac{2(\alpha+\beta-1)\binom{i}{j}\binom{\alpha+\beta-2}{\alpha-1}}{(\alpha+\beta+i)\binom{i+\alpha+\beta}{j+\alpha}}
\]
and
\[
d^{n-2}_{j}=
\begin{cases}
0, & j<L_{n-2}\\
\frac{L_{n-1}+\alpha}{n+\alpha+\beta-1}, & j=L_{n-2}\\
\frac{1}{n+\alpha+\beta-1}, & L_{n-2}<j<U_{n-2}\\
\frac{n-1-U_{n-1}+\beta}{n+\alpha+\beta-1}, & j=U_{n-2}\\
0, & j>U_{n-2}
\end{cases}
\qquad
d^{n-1}_{j}=
\begin{cases}
0, & j<L_{n-1}\\
1, & L_{n-1}\leq j\leq U_{n-1}\\
0, & j>U_{n-1}
\end{cases}
\]
\end{lemma}
\addtocounter{theorem}{-1}
}
\begin{proof}[Proof of Lemma \ref{lem:generald}]

Recall that $\state_i=j$ means Alice wins $j$ rounds in the first $i$ rounds. 

When the prior $p$ follows $\betadis(\alpha,\beta)$, the posterior $p|(\state_i=j)$ follows $\betadis(\alpha+j,\beta+i-j)$ (Claim \ref{cla:beta}). 

We can calculate the expectation of $p|(\state_{n-1}=j)$,
\[
q^{n-1}_j=\E[p|\state_{n-1}=j]=\frac{j+\alpha}{n+\alpha+\beta-1}
\]

We use induction to show that \[
\forall i,j\Pr[\state_i=j]=\frac{(\alpha+\beta-1)\binom{i}{j}\binom{\alpha+\beta-2}{\alpha-1}}{(\alpha+\beta+i-1)\binom{i+\alpha+\beta-2}{j+\alpha-1}}. 
\]

For the base case $i=0$, note that $j\leq i$ is also zero, the above equation holds naturally since \[1=\Pr[\state_0=0]=\frac{(\alpha+\beta-1)\binom{i}{j}\binom{\alpha+\beta-2}{\alpha-1}}{(\alpha+\beta+i-1)\binom{i+\alpha+\beta-2}{j+\alpha-1}}.\]

For the induction, we first show the relationship between the distribution of $\state_i$ and $\state_{i-1}$. 

\[
\Pr[\state_i=j]=
\begin{cases}
\Pr[\state_{i-1}=0]*\frac{i+\beta-1}{i+\alpha+\beta-1},&j=0\\
\Pr[\state_{i-1}=j-1]*\frac{j+\alpha-1}{i+\alpha+\beta-1}+\Pr[\state_{i-1}=j]*\frac{i-j+\beta-1}{i+\alpha+\beta-1},&0<j<i\\
\Pr[\state_{i-1}=i-1]*\frac{i+\alpha-1}{i+\alpha+\beta-1},&j=i\\
\end{cases}
\]

We assume that the equation holds for $i-1$, i.e.,  \[\forall 0\leq j\leq i-1, \Pr[\state_{i-1}=j]=\frac{(\alpha+\beta-1)\binom{(i-1)}{j}\binom{\alpha+\beta-2}{\alpha-1}}{(\alpha+\beta+(i-1)-1)\binom{(i-1)+\alpha+\beta-2}{j+\alpha-1}}\]

Then we have
\begin{align*}
\Pr[\state_i=0] =&\Pr[\state_{i-1}=0]*\frac{i+\beta-1}{i+\alpha+\beta-1}\\
                =&\frac{(\alpha+\beta-1)\binom{\alpha+\beta-2}{\alpha-1}}{(\alpha+\beta+(i-1)-1)\binom{(i-1)+\alpha+\beta-2}{\alpha-1}}*\frac{i+\beta-1}{i+\alpha+\beta-1}\\
                =&\frac{(\alpha+\beta-1)\binom{\alpha+\beta-2}{\alpha-1}}{(\alpha+\beta+i-1)\binom{i+\alpha+\beta-2}{\alpha-1}}\\
                =&\frac{(\alpha+\beta-1)\binom{i}{j}\binom{\alpha+\beta-2}{\alpha-1}}{(\alpha+\beta+i-1)\binom{i+\alpha+\beta-2}{j+\alpha-1}}\tag{Note that $j=0$}\\
\Pr[\state_i=j] =&\Pr[\state_{i-1}=j-1]*\frac{j+\alpha-1}{i+\alpha+\beta-1}+\Pr[\state_{i-1}=j]*\frac{i-j+\beta-1}{i+\alpha+\beta-1}\\
                =&\frac{(\alpha+\beta-1)\binom{(i-1)}{(j-1)}\binom{\alpha+\beta-2}{\alpha-1}}{(\alpha+\beta+(i-1)-1)\binom{(i-1)+\alpha+\beta-2}{(j-1)+\alpha-1}}*\frac{j+\alpha-1}{i+\alpha+\beta-1}+\\
                &\frac{(\alpha+\beta-1)\binom{(i-1)}{j}\binom{\alpha+\beta-2}{\alpha-1}}{(\alpha+\beta+(i-1)-1)\binom{(i-1)+\alpha+\beta-2}{j+\alpha-1}}*\frac{i-j+\beta-1}{i+\alpha+\beta-1}\\
                =&\frac{(\alpha+\beta-1)\binom{(i-1)}{(j-1)}\binom{\alpha+\beta-2}{\alpha-1}}{(\alpha+\beta+i-1)\binom{i+\alpha+\beta-2}{i+\alpha-1}}+\frac{(\alpha+\beta-1)\binom{(i-1)}{j}\binom{\alpha+\beta-2}{\alpha-1}}{(\alpha+\beta+i-1)\binom{i+\alpha+\beta-2}{\alpha-1}}\\
                =&\frac{(\alpha+\beta-1)\binom{i}{j}\binom{\alpha+\beta-2}{\alpha-1}}{(\alpha+\beta+i-1)\binom{i+\alpha+\beta-2}{j+\alpha-1}}\\
\Pr[\state_i=i] =&\Pr[\state_{i-1}=i-1]*\frac{i+\alpha-1}{i+\alpha+\beta-1}\\
                =&\frac{(\alpha+\beta-1)\binom{\alpha+\beta-2}{\alpha-1}}{(\alpha+\beta+(i-1)-1)\binom{(i-1)+\alpha+\beta-2}{i+\alpha-1}}*\frac{i+\alpha-1}{(i-1)+\alpha+\beta-1}\\
                =&\frac{(\alpha+\beta-1)\binom{\alpha+\beta-2}{\alpha-1}}{(\alpha+\beta+i-1)\binom{i+\alpha+\beta-2}{i+\alpha-1}}\\
                =&\frac{(\alpha+\beta-1)\binom{i}{j}\binom{\alpha+\beta-2}{\alpha-1}}{(\alpha+\beta+i-1)\binom{i+\alpha+\beta-2}{j+\alpha-1}}\tag{Note that $j=i$}\\
\end{align*}

Thus, the equation also holds for $i$. Therefore, we have
\[
\forall i,j\Pr[\state_i=j]=\frac{(\alpha+\beta-1)\binom{i}{j}\binom{\alpha+\beta-2}{\alpha-1}}{(\alpha+\beta+i-1)\binom{i+\alpha+\beta-2}{j+\alpha-1}}
\]

Then we can calculate $Q^i_j$
\begin{align*}
Q^i_j   &=\Pr[\state_i=j]*2*\E[p|\state_i=j]*(1-\E[p|\state_i=j])\\
        &=\Pr[\state_i=j]*2*q^i_j*(1-q^i_j)\\
        &=\frac{(\alpha+\beta-1)\binom{i}{j}\binom{\alpha+\beta-2}{\alpha-1}}{(\alpha+\beta+i-1)\binom{i+\alpha+\beta-2}{j+\alpha-1}}*2*\frac{j+\alpha}{i+\alpha+\beta}*\frac{i-j+\beta}{i+\alpha+\beta}\\
        &=\frac{2(\alpha+\beta-1)\binom{i}{j}\binom{\alpha+\beta-2}{\alpha-1}}{(\alpha+\beta+i)\binom{i+\alpha+\beta}{j+\alpha}}
\end{align*}

Recall in Section~\ref{sec:analyzegeneral}, we have defined $b^{n-1}_j=\Pr[O=1|(\state_{n-1}=j)]$  and \[d^{n-1}_j = \Pr[O=1|(\his_n=+)\wedge (S_{n-1}=j)]-\Pr[O=1|(\his_n=-)\wedge (S_{n-1}=j)].\] When $j<L_{n-1}$ ($j>U_{n-1}$), Alice must lose (win) regardless of the final round's outcome. When $L_{n-1}\leq j\leq U_{n-1}$, Alice wins if and only if she wins the final round. Therefore, 
\[
b^{n-1}_j=
\begin{cases}
0,&j<L_{n-1}\\
\frac{j+\alpha}{n+\alpha+\beta-1},& L_{n-1}\leq j\leq U_{n-1}\\
1,&j>U_{n-1}
\end{cases}
\qquad
d^{n-1}_{j}=
\begin{cases}
0, & j<L_{n-1}\\
1, & L_{n-1}\leq j\leq U_{n-1}\\
0, & j>U_{n-1}
\end{cases}
\]
Moreover, recall that $d^{n-2}_j=b^{n-1}_{j+1}-b^{n-1}_j$, thus 
\[
d^{n-2}_{j}=
\begin{cases}
0, & j<L_{n-2}\\
\frac{L_{n-1}+\alpha}{n+\alpha+\beta-1}, & j=L_{n-2}\\
\frac{1}{n+\alpha+\beta-1}, & L_{n-2}<j<U_{n-2}\\
\frac{n-1-U_{n-1}+\beta}{n+\alpha+\beta-1}, & j=U_{n-2}\\
0, & j>U_{n-2}
\end{cases}
\]

\end{proof}

\section{Formal Proof for Asymptotic Case}
{
\renewcommand{\thetheorem}{\ref{lem:asymtotic}}
\begin{lemma}[Formal]

When $n$ is sufficiently large, fixed $\alpha,\beta,\mu\in(0,1)$, we have
\[
\begin{cases}
\frac{L_{n-2}}{n}=\frac{1-\mu}2*(1+O(\frac1n))\\
\frac{L_{n-1}}{n}=\frac{1-\mu}2\\
\frac{U_{n-2}}{n}=\frac{1+\mu}2*(1+O(\frac1n))\\
\frac{U_{n-1}}{n}=\frac{1+\mu}2*(1+O(\frac1n))
\end{cases}
\]
For all $L_{n-1}\leq j\leq U_{n-1}$, let $\theta_j:=\frac{j}{n}$, we have
\[
\begin{cases}
\Pr[\state_{n-1}=j]=\frac{\betadens{\theta_j}}n*(1+O(\frac1n))\\
\Pr[\state_{n-2}=j]=\frac{\betadens{\theta_j}}n*(1+O(\frac1n))
\end{cases}
\qquad
\begin{cases}
q^{n-1}_j=\theta_j*(1+O(\frac1n))\\
q^{n-2}_j=\theta_j*(1+O(\frac1n))
\end{cases}
\]
and
\[
d^{n-2}_j=
\begin{cases}
0,&j<L_{n-2}\\
\frac{1-\mu}2*(1+O(\frac1n)),&j=L_{n-2}\\
\frac{1}{n}*(1+O(\frac1n)),&L_{n-2}<j<U_{n-2}\\
\frac{1-\mu}2*(1+O(\frac1n)),&j=U_{n-2}\\
0,&j>U_{n-2}
\end{cases}
\qquad
d^{n-1}_j=
\begin{cases}
0,&j<L_{n-1}\\
1,&L_{n-1}\leq j\leq U_{n-1}\\
0,&j>U_{n-1}
\end{cases}
\]
\end{lemma}
\addtocounter{theorem}{-1}
}

\begin{proof}[Proof of Lemma \ref{lem:asymtotic}]
First recall the definition of $L_{n-2},L_{n-1},U_{n-2},U_{n-1}$,

\[
\begin{cases}
L_{n-2}=\frac{n-x-2}2\\
L_{n-1}=\frac{n-x}2\\
U_{n-2}=\frac{n+x-2}2\\
U_{n-1}=\frac{n+x-2}2\\
\end{cases}
\]

Then by substituting $\mu=\frac{x}{n}$, 
\[
\begin{cases}
\frac{L_{n-2}}{n}=\frac{1-\mu}2-\frac1n\\
\frac{L_{n-1}}{n}=\frac{1-\mu}2\\
\frac{U_{n-2}}{n}=\frac{1+\mu}2-\frac1n\\
\frac{U_{n-1}}{n}=\frac{1+\mu}2-\frac1n
\end{cases}
\]

Given fixed $\mu\in(0,1)$,

\[
\begin{cases}
\frac{L_{n-2}}{n}=\frac{1-\mu}2*(1-\frac1n\big/\frac{1-\mu}2)=\frac{1-\mu}2*(1+O(\frac1n))\\
\frac{L_{n-1}}{n}=\frac{1-\mu}2\\
\frac{U_{n-2}}{n}=\frac{1+\mu}2*(1-\frac1n\big/\frac{1+\mu}2)=\frac{1-\mu}2*(1+O(\frac1n))\\
\frac{U_{n-1}}{n}=\frac{1+\mu}2*(1-\frac1n\big/\frac{1+\mu}2)=\frac{1-\mu}2*(1+O(\frac1n))\\
\end{cases}
\]

Second, recall the formula of $\Pr[S_{n-1}=j]$ and $\Pr[S_{n-2}=j]$. 
\[
\begin{cases}
\Pr[S_{n-1}=j]=\frac{(\alpha+\beta-1)\binom{n-1}j\binom{\alpha+\beta-2}{\alpha-1}}{(\alpha+\beta+n-2)\binom{n+\alpha+\beta-3}{j+\alpha-1}}\\
\Pr[S_{n-2}=j]=\frac{(\alpha+\beta-1)\binom{n-2}j\binom{\alpha+\beta-2}{\alpha-1}}{(\alpha+\beta+n-3)\binom{n+\alpha+\beta-4}{j+\alpha-1}}\\
\end{cases}
\]

To show that $\Pr[S_{n-1}=j]=\frac{\betadens{\theta_j}}n*(1+O(\frac1n))$, we calculate the ratio between $\Pr[S_{n-1}=j]$ and $\frac{\betadens{\theta_j}}n$

\begin{align}
    &\Pr[S_{n-1}=j]\bigg/\frac{\betadens{\theta_j}}n \notag\\
    =&\frac{(\alpha+\beta-1)\binom{n-1}j\binom{\alpha+\beta-2}{\alpha-1}}{(\alpha+\beta+n-2)\binom{n+\alpha+\beta-3}{j+\alpha-1}}\bigg/\frac{\Gamma(\alpha+\beta)(\frac{j}{n})^{\alpha-1}*(\frac{n-j}{n})^{\beta-1}}{n\Gamma(\alpha)\Gamma(\beta)} \notag\\
    =&\frac{n\Gamma(n)\Gamma(j+\alpha)\Gamma(n-j+\beta-1)}{(\alpha+\beta+n-2)\Gamma(n+\alpha+\beta-2)\Gamma(j+1)\Gamma(n-j)*(\frac{j}{n})^{\alpha-1}*(\frac{n-j}{n})^{\beta-1}} \tag{Note that $(\alpha+\beta-1)\binom{\alpha+\beta-2}{\alpha-1}=\frac{\Gamma(\alpha+\beta)}{\Gamma(\alpha)\Gamma(\beta)}$, $\binom{n-1}j=\frac{\Gamma(n)}{\Gamma(j+1)\Gamma(n-j)}$,  $\binom{n+\alpha+\beta-3}{j+\alpha-1}=\frac{\Gamma(n+\alpha+\beta-2)}{\Gamma(j+\alpha)\Gamma(n-j+\beta-1)}$}\\
    =&\frac{n}{n+\alpha+\beta-2}*\frac{(j+1)^{\overline{\alpha-1}\footnotemark[18] }}{j^{\alpha-1}}*\frac{(n-j)^{\overline{\beta-1}}}{(n-j)^{\beta-1}}*\frac{n^{\alpha+\beta-2}}{n^{\overline{\alpha+\beta-2}}} \tag{Note that $\frac{\Gamma(j+\alpha)}{\Gamma(j+1)}=(j+1)^{\overline{\alpha-1}}$, $\frac{\Gamma(n-j+\beta-1)}{\Gamma(n-j)}=(n-j)^{\overline{\beta-1}}$, $\frac{\Gamma(n)}{\Gamma(n+\alpha+\beta-2)}=\frac{1}{n^{\overline{\alpha+\beta-2}}}$} \\
    =&\frac{(j+1)^{\overline{\alpha-1}}}{j^{\alpha-1}}*\frac{(n-j)^{\overline{\beta-1}}}{(n-j)^{\beta-1}}*\left(\frac{n}{n+\alpha+\beta-2}*\frac{n^{\alpha+\beta-2}}{n^{\overline{\alpha+\beta-2}}}\right)\notag\\
    =&\frac{j+\alpha-1}{j}*\frac{j^{\overline{\alpha-1}}}{j^{\alpha-1}}*\frac{(n-j)^{\overline{\beta-1}}}{(n-j)^{\beta-1}}*\frac{n^{\alpha+\beta-1}}{n^{\overline{\alpha+\beta-1}}} \notag
\end{align}

\footnotetext[21]{We define $x^{\overline{t}}$ as the continuous generalization of rising factorials, i.e. $x^{\overline{t}}=\frac{\Gamma(x+t)}{\Gamma(x)}$.}

Then we need to prove that each item of the above formula is $1+O(\frac1n)$. 

Since $L_{n-1}\leq j\leq U_{n-1}$ and $L_{n-1}=\frac{1-\mu}2*(n+O(1)),U_{n-1}=\frac{1+\mu}2*(n+O(1))$, we have $j=\Theta(n)$ and $n-j=\Theta(n)$. Therefore, $\frac{j+\alpha-1}{j}=1+O(\frac1n)$.

It's left to prove that $\frac{j^{\overline{\alpha-1}}}{j^{\alpha-1}}$, $\frac{(n-j)^{\overline{\beta-1}}}{(n-j)^{\beta-1}}$ and $\frac{n^{\alpha+\beta-1}}{n^{\overline{\alpha+\beta-1}}}$ are all $1+O(\frac1n)$. We give the following claim: 
\begin{claim}
For any $x=\Theta(n),t=\Theta(1)$, $\frac{x^{\overline{t}}}{x^t}=1+O(\frac1n)$
\label{cla:rising_power}
\end{claim}
\begin{proof}[Proof of Claim \ref{cla:rising_power}]
When $t$ is an integer, we have
\[
x^t\leq x^{\overline{t}}=\frac{\Gamma(x+t)}{\Gamma(x)}=\prod_{i=1}^t(x+i-1)\leq(x+t-1)^t,
\]
then
\[
1\leq \frac{x^{\overline{t}}}{x^t}\leq (1+\frac{t-1}{x})^t
\]
Since $1+\frac{t-1}x=1+O(\frac1n)$, then $(1+\frac{t}x)^t=(1+O(\frac1n))^t=1+O(\frac1n)$, and we have $\frac{x^{\overline{t}}}{x^t}=1+O(\frac1n)$.

When $t$ is not an integer, we need to use Gautschi's inequality (see 5.6.4 in~\citet{lozier2003nist}) \[\forall s\in(0,1), x^{1-s}<\frac{\Gamma(x+1)}{\Gamma(x+s)}<(x+1)^{1-s}.\]

Let $c$ be the smallest integer larger than $t$, then we have
\begin{align*}
x^c<x^{\overline{c}}=&\frac{\Gamma(x+c)}{\Gamma(x)}<(x+c-1)^c\tag{Due to previous analysis}\\
(x+c-1)^{c-t}<&\frac{\Gamma(x+c)}{\Gamma(x+t)}<(x+c)^{c-t}\tag{Gautschi's inequality}\\
\Rightarrow \frac{x^c}{(x+c)^{c-t}}<x^{\overline{t}}=&\frac{\Gamma(x+t)}{\Gamma(x)}<(x+c-1)^t
\end{align*}
then
\[
\frac{x}{x+c}< (\frac{x}{x+c})^{c-t}< \frac{x^{\overline{t}}}{x^t}<(1+\frac{c-1}{x})^t<(1+\frac{c-1}{x})^c
\]
Since $\frac{x}{x+c}=1+O(\frac1n)$ and $(1+\frac{c-1}x)=1+O(\frac1n)$, then $(1+\frac{c-1}{x})^c=1+O(\frac1n)$.

Thus we proved $\frac{x^{\overline{t}}}{x^t}=1+O(\frac1n)$.
\end{proof}
Based on Claim \ref{cla:rising_power}, we have $\frac{(j)^{\overline{\alpha-1}}}{j^{\alpha-1}}=1+O(\frac{1}{n})$, $\frac{(n-j)^{\overline{\beta-1}}}{(n-j)^{\beta-1}}=1+O(\frac{1}{n})$ and $\frac{n^{\alpha+\beta-1}}{n^{\overline{\alpha+\beta-1}}}=\frac1{1+O(\frac{1}{n})}=1+O(\frac1n)$

Then 
\begin{align*}
&\Pr[S_{n-1}=j]\big/\frac{\betadens{\theta_j}}n\\
=&(1+O(\frac1n))*(1+O(\frac1n))*(1+O(\frac1n))*(1+O(\frac1n))*(1+O(\frac1n))*(1+O(\frac1n))\\
=&1+O(\frac1n)
\end{align*}

In order to prove $\Pr[S_{n-2}=j]=\frac{\betadens{\theta_j}}n*(1+O(\frac1n))$, we first prove $\Pr[S_{n-2}=j]=\Pr[S_{n-1}=j]*(1+O(\frac1n))$
\begin{align*}
\Pr[S_{n-2}=j]\bigg/\Pr[S_{n-1}=j]=&\frac{(\alpha+\beta+n-3)\binom{n+\alpha+\beta-3}{j+\alpha-1}\binom{n-2}j}{(\alpha+\beta+n-2)\binom{n+\alpha+\beta-4}{j+\alpha-1}\binom{n-1}j}\\
=&\frac{\alpha+\beta+n-3}{\alpha+\beta+n-2}*\frac{n+\alpha+\beta-3}{n-j+\beta-2}*\frac{n-j-1}{n-1}\\
=&\frac{\alpha+\beta+n-3}{\alpha+\beta+n-2}*\frac{n+\alpha+\beta-3}{n-1}*\frac{n-j-1}{n-j+\beta-2}\\
=&(1+O(\frac1n))*(1+O(\frac1n))*(1+O(\frac1n))\tag{$\alpha,\beta$ are constants}\\
=&1+O(\frac1n)
\end{align*}

Since $\Pr[S_{n-1}=j]=\frac{\betadens{\theta_j}}n*(1+O(\frac1n))$, we proved $\Pr[S_{n-2}=j]=\frac{\betadens{\theta_j}}n*(1+O(\frac1n))$

Third, recall the formula of $q_j^{n-1}$ and $q_j^{n-2}$
\[
\begin{cases}
q_j^{n-1}=\frac{j+\alpha}{n+\alpha+\beta-1}\\
q_j^{n-2}=\frac{j+\alpha}{n+\alpha+\beta-2}\\
\end{cases}
\]
We have
\[
\begin{cases}
q_j^{n-1}\big/\theta_j=\frac{j+\alpha}{n+\alpha+\beta-1}\big/\frac{j}{n}=\frac{j+\alpha}j*\frac{n}{n+\alpha+\beta-1}\\
q_j^{n-2}\big/\theta_j=\frac{j+\alpha}{n+\alpha+\beta-2}\big/\frac{j}{n}=\frac{j+\alpha}j*\frac{n}{n+\alpha+\beta-2}\\
\end{cases}
\]
Thus
\[
\begin{cases}
q_j^{n-1}\big/\theta_j=(1+O(\frac1n))*(1+O(\frac1n))=1+O(\frac1n)\\
q_j^{n-2}\big/\theta_j=(1+O(\frac1n))*(1+O(\frac1n))=1+O(\frac1n)\\
\end{cases}
\Rightarrow
\begin{cases}
q^{n-1}_j=\theta_j*(1+O(\frac1n))\\
q^{n-2}_j=\theta_j*(1+O(\frac1n))
\end{cases}
\]

Finally, recall the formula of $d^{n-2}_j$
\[
d^{n-2}_j=
\begin{cases}
0, & j<L_{n-2}\\
\frac{L_{n-1}+\alpha}{n+\alpha+\beta-1}, & j=L_{n-2}\\
\frac{1}{n+\alpha+\beta-1}, & L_{n-2}<j<U_{n-2}\\
\frac{n-1-U_{n-1}+\beta}{n+\alpha+\beta-1}, & j=U_{n-2}\\
0, & j>U_{n-2}
\end{cases}
\]
We have
\[
d^{n-2}_j=
\begin{cases}
\frac{L_{n-1}+\alpha}{n+\alpha+\beta-1}=\frac{L_{n-1}}n*(1+O(\frac1n)), & j=L_{n-2}\\
\frac{1}{n+\alpha+\beta-1}=\frac1n*(1+O(\frac1n)), & L_{n-2}<j<U_{n-2}\\
\frac{n-1-U_{n-1}+\beta}{n+\alpha+\beta-1}=\frac{n-U_{n-1}}n*(1+O(\frac1n)), & j=U_{n-2}
\end{cases}
\]
By substituting the value of $L_{n-2}$ and $U_{n-2}$, we get

\[
d^{n-2}_j=
\begin{cases}
0,&j<L_{n-2}\\
\frac{1-\mu}2*(1+O(\frac1n)),&j=L_{n-2}\\
\frac{1}{n}*(1+O(\frac1n)),&L_{n-2}<j<U_{n-2}\\
\frac{1-\mu}2*(1+O(\frac1n)),&j=U_{n-2}\\
0,&j>U_{n-2}
\end{cases}
\]
\end{proof}

We use Lemma~\ref{lem:asymtotic} to prove \[\E[\Delta_\belset(n*\mu)]=\asyf_{\alpha,\beta,n}(\mu)*(1+O(\frac1n)).\]

We prove the four parts separately.

\paragraph{\mytextabig} We need to prove 
\[
\Pr[\state_{n-2}=L_{n-2}]*2 q^{n-2}_{L_{n-2}}*(1-q^{n-2}_{L_{n-2}})*d_{L_{n-2}}^{n-2}=\frac{\betadens{\frac{1-\mu}{2}}}n*2*\frac{1+\mu}{2}*\frac{1-\mu}{2}*\frac{1-\mu}{2}*(1+O(\frac1n))
\]
Recall Lemma \ref{lem:asymtotic}, we have 
\[
\begin{cases}
\Pr[\state_{n-2}=L_{n-2}]&=\frac{\betadens{\frac{1-\mu}2}}{n}*(1+O(\frac1n))\\
q_{L_{n-2}}^{n-2}&=\frac{1-\mu}2*(1+O(\frac1n))\\
1-q_{L_{n-2}}^{n-2}&=\frac{1+\mu}2*(1+O(\frac1n))\\
d_{L_{n-2}}^{n-2}&=\frac{1-\mu}2*(1+O(\frac1n))
\end{cases}
\]
By substituting the above items into the \textbf{\mytexta} part, we get
\begin{align*}
&\Pr[\state_{n-2}=L_{n-2}]*2 q^{n-2}_{L_{n-2}}*(1-q^{n-2}_{L_{n-2}})*d_{L_{n-2}}^{n-2}\\
=&\frac{\betadens{\frac{1-\mu}{2}}}n*(1+O(\frac1n))*2*\frac{1-\mu}{2}*(1+O(\frac1n))*\frac{1+\mu}{2}*(1+O(\frac1n))*\frac{1-\mu}{2}*(1+O(\frac1n))\\
=&\frac{\betadens{\frac{1-\mu}{2}}}n*2*\frac{1+\mu}{2}*\frac{1-\mu}{2}*\frac{1-\mu}{2}*(1+O(\frac1n))
\end{align*}
\paragraph{\mytextbbig}
We need to prove 
\[
\Pr[\state_{n-2}=U_{n-2}]*2 q^{n-2}_{U_{n-2}}*(1-q^{n-2}_{U_{n-2}})*d_{U_{n-2}}^{n-2}=\frac{\betadens{\frac{1+\mu}{2}}}n*2*\frac{1+\mu}{2}*\frac{1-\mu}{2}*\frac{1-\mu}{2}*(1+O(\frac1n))
\]
Recall Lemma \ref{lem:asymtotic}, we have 
\[
\begin{cases}
\Pr[\state_{n-2}=U_{n-2}]&=\frac{\betadens{\frac{1+\mu}2}}{n}*(1+O(\frac1n))\\
q_{U_{n-2}}^{n-2}&=\frac{1+\mu}2*(1+O(\frac1n))\\
1-q_{U_{n-2}}^{n-2}&=\frac{1-\mu}2*(1+O(\frac1n))\\
d_{U_{n-2}}^{n-2}&=\frac{1-\mu}2*(1+O(\frac1n))
\end{cases}
\]
By substituting the above items into the \textbf{\mytextb} part, we get
\begin{align*}
&\Pr[\state_{n-2}=U_{n-2}]*2 q^{n-2}_{U_{n-2}}*(1-q^{n-2}_{U_{n-2}})*d_{U_{n-2}}^{n-2}\\
=&\frac{\betadens{\frac{1-\mu}{2}}}n*(1+O(\frac1n))*2*\frac{1+\mu}{2}*(1+O(\frac1n))*\frac{1-\mu}{2}*(1+O(\frac1n))*\frac{1-\mu}{2}*(1+O(\frac1n))\\
=&\frac{\betadens{\frac{1-\mu}{2}}}n*2*\frac{1+\mu}{2}*\frac{1-\mu}{2}*\frac{1-\mu}{2}*(1+O(\frac1n))
\end{align*}

To prove the remaining parts, we need to extend our results for $\theta_j$ to $\theta\in[\theta_j,\theta_{j+1}]$. Note that $\forall \theta\in[\theta_j,\theta_{j+1}], \frac{L_{n-1}}n<\theta<\frac{U_{n-1}}n$, we have

\[
\begin{cases}
\betadens{\theta_j}=\betadens{\theta}*O(1+\frac1n)\footnotemark[19]\\
q^{n-2}_{\theta_j*n}=q^{n-2}_{\theta*n}*O(1+\frac1n)\\
q^{n-1}_{\theta_j*n}=q^{n-1}_{\theta*n}*O(1+\frac1n)\\
d^{n-2}_{\theta_j*n}=d^{n-2}_{\theta*n}*O(1+\frac1n)\\
d^{n-1}_{\theta_j*n}=d^{n-1}_{\theta*n}*O(1+\frac1n)\\
\end{cases}
\]
\footnotetext[22]{Note that the derivative of $\betadens{\theta}$ is bounded.}
We are ready to prove the last two parts. 

\paragraph{\mytextmbig}
Here we need to prove 
\[
\sum_{j=L_{n-2}+1}^{U_{n-2}-1}\Pr[\state_{n-2}=j]*2q^{n-2}_{j}*(1-q^{n-2}_{j})*d_{j}^{n-2}=\int_{\frac{1-\mu}2}^{\frac{1+\mu}2}\betadens{\theta}*2*\theta*(1-\theta)*\frac{1}{n}d\theta * (1+O(\frac1n))
\]

Recall Lemma \ref{lem:asymtotic}, we have 
\[
\begin{cases}
\Pr[\state_{n-2}=j]&=\frac{\betadens{\theta_j}}{n}*(1+O(\frac1n))\\
q_{j}^{n-2}&=\theta_j*(1+O(\frac1n))\\
1-q_{j}^{n-2}&=(1-\theta_j)*(1+O(\frac1n))\\
d_{j}^{n-2}&=\frac1n*(1+O(\frac1n))
\end{cases}
\]

By substituting the above items into \textbf{\mytextm} part, we get
\begin{align}
&\sum_{j=L_{n-2}+1}^{U_{n-2}-1}\Pr[\state_{n-2}=j]*2q^{n-2}_{j}*(1-q^{n-2}_{j})*d_{j}^{n-2}\notag\\
=&\sum_{j=L_{n-2}+1}^{U_{n-2}-1}\frac{\betadens{\theta_j}}{n}*(1+O(\frac1n))*2*\theta_j*(1+O(\frac1n))*(1-\theta_j)*(1+O(\frac1n))*\frac1n*(1+O(\frac1n))\notag\\
=&\sum_{j=L_{n-2}+1}^{U_{n-2}-1}\frac{\betadens{\theta_j}}{n}*2*\theta_j*(1-\theta_j)*\frac1n*(1+O(\frac1n))\label{eq:textm}
\end{align}

Note that
\begin{align}
&\frac{\betadens{\theta_j}}{n}*2*\theta_j*(1-\theta_j)*\frac1n\notag\\
=&(\theta_{j+1}-\theta_j)*(\betadens{\theta_j}*2*\theta_j*(1-\theta_j))*\frac1n\notag\\
=&\int_{\theta_j}^{\theta_{j+1}}(\betadens{\theta_j}*2*\theta_j*(1-\theta_j))*\frac1n d\theta\notag\\
=&\int_{\theta_j}^{\theta_{j+1}}(\betadens{\theta}*(1+O(\frac1n)))*2*\theta*(1+O(\frac1n)))*(1-\theta)*(1+O(\frac1n)))*\frac1n)d\theta\notag\\
=&\int_{\theta_j}^{\theta_{j+1}}(\betadens{\theta}*2*\theta*(1-\theta)*\frac1n)d\theta*(1+O(\frac1n))\label{eq:sum_to_int}
\end{align}

Then we have 
\begin{align*}
\eqref{eq:textm}=&\sum_{j=L_{n-2}+1}^{U_{n-2}-1}\int_{\theta_j}^{\theta_{j+1}}(\betadens{\theta}*2*\theta*(1-\theta)*\frac1n)d\theta*(1+O(\frac1n))\\
=&\int_{\frac{L_{n-2}+1}n}^{\frac{U_{n-2}}n}(\betadens{\theta}*2*\theta*(1-\theta)*\frac1n)d\theta*(1+O(\frac1n))\\
=&\frac1n(\int_{\frac{1-\mu}2}^{\frac{1+\mu}2}\betadens{\theta}*2*\theta*(1-\theta)d\theta+O(\frac1n))*(1+O(\frac1n))\tag{$\betadens{\theta}*2*\theta*(1-\theta)$ is bounded by two positive constants}\\
=&\int_{\frac{1-\mu}2}^{\frac{1+\mu}2}(\betadens{\theta}*2*\theta*(1-\theta)*\frac1n)d\theta*(1+O(\frac1n))
\end{align*}

\paragraph{\mytextnbig}
Finally, we need to prove
\[
\sum_{j=L_{n-1}}^{U_{n-1}}\Pr[\state_{n-1}=j]*2q^{n-1}_{j}*(1-q^{n-1}_{j})*d_{j}^{n-1}=\int_{\frac{1-\mu}2}^{\frac{1+\mu}2}\betadens{\theta}*2*\theta*(1-\theta)d\theta * (1+O(\frac1n))
\]

The proof is similar to the \textbf{\mytextm} part. 

Recall Lemma \ref{lem:asymtotic}, we have 
\[
\begin{cases}
\Pr[\state_{n-1}=j]&=\frac{\betadens{\theta_j}}{n}*(1+O(\frac1n))\\
q_{j}^{n-1}&=\theta_j*(1+O(\frac1n))\\
1-q_{j}^{n-1}&=(1-\theta_j)*(1+O(\frac1n))\\
d_{j}^{n-1}&=1
\end{cases}
\]

By substituting the above items into \textbf{\mytextn} part, we get
\begin{align*}
&\sum_{j=L_{n-1}}^{U_{n-1}}\Pr[\state_{n-1}=j]*2q^{n-1}_{j}*(1-q^{n-1}_{j})*d_{j}^{n-1}\notag\\
=&\sum_{j=L_{n-1}}^{U_{n-1}}\frac{\betadens{\theta_j}}{n}*(1+O(\frac1n))*2*\theta_j*(1+O(\frac1n))*(1-\theta_j)*(1+O(\frac1n))*\frac1n*(1+O(\frac1n))\notag\\
=&\sum_{j=L_{n-1}}^{U_{n-1}}\frac{\betadens{\theta_j}}{n}*2*\theta_j*(1-\theta_j)*\frac1n*(1+O(\frac1n))\\
=&\sum_{j=L_{n-1}}^{U_{n-1}}\int_{\theta_j}^{\theta_{j+1}}(\betadens{\theta}*2*\theta*(1-\theta)*\frac1n)d\theta*(1+O(\frac1n))\\
=&\int_{\frac{L_{n-1}}n}^{\frac{U_{n-1}+1}n}(\betadens{\theta}*2*\theta*(1-\theta)*\frac1n)d\theta*(1+O(\frac1n))\tag{Recall formula~\eqref{eq:sum_to_int}}\\
=&\frac1n(\int_{\frac{1-\mu}2}^{\frac{1+\mu}2}\betadens{\theta}*2*\theta*(1-\theta)d\theta+O(\frac1n))*(1+O(\frac1n))\tag{$\betadens{\theta}*2*\theta*(1-\theta)$ is bounded by two positive constants}\\
=&\int_{\frac{1-\mu}2}^{\frac{1+\mu}2}(\betadens{\theta}*2*\theta*(1-\theta)*\frac1n)d\theta*(1+O(\frac1n))
\end{align*}

\section{Property of $F(x)$ and $G(\mu)$}

{
\renewcommand{\thetheorem}{\ref{lem:fx}}
\begin{lemma}
When $p\geq \frac12$, $F(x)=0$ has a trivial solution at $x=1$, and has a non-trivial solution $\Tilde{x}\in (1,2np-n+1)$ if and only if $p>\frac12$ and $n>\frac{1}{(\frac{1}{2}-p)\ln (\frac{1-p}{p})}$. There is no other solution. Moreover, when $x\in (1,\Tilde{x})$, $F(x)>0$, when $\Tilde{x}<n-1$ and $x\in (\Tilde{x},n-1]$, $F(x)<0$. 

Besides, let $a=2np-n-2$, when $p>\frac{1}{1+(a+1)^{-\frac{1}{a}}}$, the non-trivial solution of $F(x)=0$ is in $(2np-n-1,2np-n+1)$.
\end{lemma}
\addtocounter{theorem}{-1}
}

\begin{proof}[Proof of Lemma \ref{lem:fx}]

We first define two shorthand notations: 
\[\begin{cases}
b := 2np-n\\
t := x-1
\end{cases}\]

We have $F(t+1)=(b-t)p^t+(-b-t)(1-p)^t$, where $t\in[0,n)$. When $p=\frac12$, $F(t+1)=-2t(\frac12)^t$ thus only has a trivial solution at $t=0$. We only need to consider the case of $p>\frac12$. Here $b>0$. $F(t+1)|_{t=0}=0$, and $F(t+1)|_{t\geq b}<0$ since $(b-t)p^t|_{t\geq b}\leq 0$ and $(-b-t)(1-p)^t|_{t\geq b}< 0$. Then to analyze $F(t+1)=0$'s other solutions, we only need to discuss $t\in [0,b)$, i.e. $x\in[1,2np-n+1)$. 

\begin{align*}
F(x)&\geq0 \\
\Leftrightarrow \frac{b-t}{b+t}&\geq\left(\frac{1-p}{p}\right)^t \\
\Leftrightarrow \ln \frac{b-t}{b+t}&\geq t \ln \frac{1-p}{p}
\end{align*}

Let $f(t)= \ln \frac{b-t}{b+t}- t \ln \frac{1-p}{p} $, then we have $F(x)>0 \Leftrightarrow f(t)>0$, and $F(x)=0 \Leftrightarrow f(t)=0$. 

In addition to the trivial solution $f(0)=0$, to analyze other solutions in $t\in(0,b)$, we calculate the first and second derivative of $f(t)$, 

\begin{align*}
    \frac{d}{dt}f(t)= -\frac{2b}{b^2-t^2}-\ln \frac{1-p}{p}
\end{align*}

\begin{align*}
    \frac{d^2}{dt^2} f(t)=-\frac{4bt}{(b-t)^2(b+t)^2}
\end{align*}

Recall that $b=2np-n=(2p-1)n>0$ since $n>1,p>\frac12$.
Thus, $\frac{d^2}{dt^2}f(t)<0$ when $t\in (0,b)$, i.e. $f(t)$ is concave. Moreover, $f(0)=0$ and $\lim_{t\rightarrow b^{-}}f(t)=-\infty$. Therefore, $f(t)=0$ has a non-trivial solution $\Tilde{t} \in (0,b)$ if and only if $\frac{d}{dt}f(0)=-\frac{2}{2np-n}-\ln \frac{1-p}{p}>0$, i.e. $n>\frac{1}{(\frac{1}{2}-p)\ln (\frac{1-p}{p})}$. Moreover, $f(t)>0$ when $t<\Tilde{t}$ and $f(t)<0$ when $t>\Tilde{t}$. 

This implies that when $n>\frac{1}{(\frac{1}{2}-p)\ln (\frac{1-p}{p})}$, $F(x)$ has a non-trivial solution, denoted by $\Tilde{x}$ and $F(x)>0$ when $1<x<\Tilde{x}$ and $F(x)<0$ when $\Tilde{x}<x<2np-n+1$. 

It's left to show that when $p>\frac{1}{1+(a+1)^{-\frac{1}{a}}}$, where $a=2np-n-2$, the non-trivial solution of $F(x)=0$ is in $(2np-n-1,2np-n+1)$. Note that $F(2np-n+1)=2(n-2np)(1-p)^{2np-n}<0$. Thus, when $F(2np-n-1)>0$, i.e.,\begin{align*}2p^{2np-n-2}+2(n-2np-1)(1-p)^{2np-n-2}>0,\end{align*} there is a solution $\Tilde{x}$ of $F(x)=0$ in $(2np-n-1,2np-n+1)$. The above inequality is equivalent to $p>\frac{1}{1+(a+1)^{-\frac{1}{a}}}$.

Thus, when $p>\frac{1}{1+(a+1)^{-\frac{1}{a}}}$, $F(x)=0$ has a solution $\Tilde{x}\in(2np-n-1,2np-n+1)$, which means the difference between $2np-n$ ("expected lead") and $\Tilde{x}$ (the solution of $F(x)=0$) is less than $1$, i.e. $|\Tilde{x}-(2np-n)|<1$. 

\end{proof}

{
\renewcommand{\thetheorem}{\ref{lem:gu}}
\begin{lemma}[Property of $G(\mu)$]
For all $\alpha\geq \beta$, when $n$ is sufficiently large, $G(0)>0$ and $G(1)<0$. $G(\mu)=0,\mu\in[0,1]$ has a unique solution and the solution is in $(0,\frac{(\alpha-\beta)\harmo+1}{(\alpha+\beta)\harmo-1})$. 

Moreover, for all $0<\epsilon< \frac{(\alpha-\beta)\harmo+1}{(\alpha+\beta)\harmo-1}$, when $\alpha-\beta>\frac{\log (\frac{2(\alpha-\beta)\harmo}{(\alpha+\beta)\harmo-1}-\epsilon)-\log \epsilon}{\log (1+\frac{(\alpha-\beta)\harmo+1}{(\alpha+\beta)\harmo-1}-\epsilon)-\log (1-(\frac{(\alpha-\beta)\harmo+1}{(\alpha+\beta)\harmo-1}-\epsilon))) }$, the solution is within $(\frac{(\alpha-\beta)\harmo+1}{(\alpha+\beta)\harmo-1}-\epsilon,\frac{(\alpha-\beta)\harmo+1}{(\alpha+\beta)\harmo-1})\approx (\frac{\alpha-\beta}{\alpha+\beta}-\epsilon, \frac{\alpha-\beta}{\alpha+\beta})$. 
\end{lemma}
\addtocounter{theorem}{-1}
}

\begin{proof}[Proof of Lemma~\ref{lem:gu}]
We first define three shorthand notations: 
\[\begin{cases}
w := \alpha-\beta\geq 0\\
\ell := \frac{(\alpha-\beta)\harmo+1}{(\alpha+\beta)\harmo-1}\\
f := -\frac{(-\alpha+\beta)\harmo+1}{(\alpha+\beta)\harmo-1}\\
\end{cases}\]

Notice that \[\ell-f=\frac{2}{(\alpha+\beta)\harmo-1}.\] In the asymptotic case, $\harmo$ can be sufficiently large, thus $\ell\approx f, \ell-f> 0$ and both $\ell,f\in[0,1]$ when $n$ is sufficiently large.  

We have 
\[G(\mu)=(1+\mu)^w\left(\ell-\mu\right)+(1-\mu)^w\left(-f-\mu\right)\] 

First we notice that $G(0)=\ell-f>0$ and $G(1)=2^w(\ell-1)<0$.

\[G(\mu)=0 \Leftrightarrow \frac{(1+\mu)^w}{(1-\mu)^w}=\frac{f+\mu}{\ell-\mu} \]

Thus $G(\mu)=0$'s solution must be less than $\ell$. We only need to consider $\mu\in(0,\ell)$. Then we can use logarithm:

\[ G(\mu)=0 \Leftrightarrow w\ln\frac{1+\mu}{1-\mu}=\ln\frac{f+\mu}{\ell-\mu} \]

Let $g(\mu):=w\ln\frac{1+\mu}{1-\mu}-\ln\frac{f+\mu}{\ell-\mu}$. Note that $g(0)=\ln \ell -\ln f>0$ and $\lim_{\mu\rightarrow \ell^{\minus}} g(\mu)=-\infty$. Thus, both $g(\mu)=0$ and $G(\mu)=0$ must have a solution in $(0,\ell)$. It's left to show that this solution is unique. 

\[\frac{d}{d\mu}g(\mu)=w \frac{2}{1-\mu^2}-\frac{\ell+f}{(\ell-\mu)(f+\mu)}\]

\[\frac{d}{d\mu}g(\mu)=0 \Leftrightarrow 2w(\ell-\mu)(f+\mu)=(\ell+f)(1-\mu^2) \]

Let \[\phi(\mu):=2w(\ell-\mu)(f+\mu)-(\ell+f)(1-\mu^2)=(\ell+f-2w)\mu^2+2w(\ell-f)\mu+2\ell w f-\ell-f.\] 

We will show that at most one of $\phi(\mu)=0$'s solutions is $\geq$ 0. Intuitively, $\phi(\mu)$ is symmetric regarding $\mu=0$ approximately when $n$ is sufficiently large. Thus, one of the solutions must be less than 0. Formally, note that $\phi(\mu)=0$'s solutions are

\[\frac{-b\pm\sqrt{b^2-4ac}}{2 a}\]

where $a=\ell+f-2w=\frac{2(\alpha-\beta)\harmo}{(\alpha+\beta)\harmo-1}-2(\alpha-\beta)\leq 0$ since $\alpha,\beta\geq 1$. 

If $a=0$, i.e., $\alpha=\beta$, $\phi(\mu)=0$ only has one solution. 

If $a<0$, i.e., $\alpha>\beta$, note that 

\begin{align*}
c &= 2 w \ell  f-\ell-f\\
& = 2 (\alpha-\beta) \frac{(\alpha-\beta)\harmo+1}{(\alpha+\beta)\harmo-1} \frac{(\alpha-\beta)\harmo-1}{(\alpha+\beta)\harmo-1} -(\frac{(\alpha-\beta)\harmo+1}{(\alpha+\beta)\harmo-1}+\frac{(\alpha-\beta)\harmo-1}{(\alpha+\beta)\harmo-1})\\
& = \frac{2(\alpha-\beta)}{((\alpha+\beta)\harmo-1)^2}\left(((\alpha-\beta)\harmo+1)((\alpha-\beta)\harmo-1)-\harmo((\alpha+\beta)\harmo-1)\right)\\
& = \frac{2(\alpha-\beta)}{((\alpha+\beta)\harmo-1)^2}\left(\harmo^2((\alpha-\beta)^2-(\alpha+\beta))+\harmo-1\right)
\end{align*}

If $(\alpha-\beta)^2\geq (\alpha+\beta)$, when $n$ is sufficiently large, $c>0$, recall that $a<0$, then $-b+\sqrt{b^2-4ac}>0$, one solution $\frac{-b+\sqrt{b^2-4ac}}{2 a}<0$. 

If $(\alpha-\beta)^2<(\alpha+\beta)$, $c$ goes to a negative value when $n$ goes to infinity while $b$ goes to 0. In this case, $b^2-4ac$ goes to a negative value when $n$ goes to infinity. Thus, when $n$ is sufficiently large, $\phi(\mu)=0$ does not have any solution. 

Therefore, at most one of $\phi(\mu)=0$'s solutions is in $[0,1]$. This shows that at most one of $\frac{d}{d\mu}g(\mu)=0$'s solutions is in $[0,1]$ since $\frac{d}{d\mu}g(\mu)=0$ is equivalent to $\phi(\mu)=0$. 

\[\frac{d}{d\mu}g(\mu)|_{\mu=0}=2w-\frac{\ell+f}{\ell f}\]

\[\lim_{\mu\rightarrow \ell^{\minus}}\frac{d}{d\mu}g(\mu)=-\infty\]

Thus, $\frac{d}{d\mu}g(\mu)$ is either always negative or first positive and then negative. This shows that $g(\mu)$ is either always decreasing or first increasing and then decreasing. Recall that $g(0)=\ln \ell -\ln f>0$ and $\lim_{\mu\rightarrow \ell^{\minus}} g(\mu)=-\infty$. Therefore, $g(\mu)=0$ must have a unique solution in $(0,\ell)$. 

Moreover, for all $0<\epsilon<\ell$ if $w$ is sufficiently large such that $g(\ell-\epsilon)>0$, the unique solution is in $(\ell-\epsilon,\ell)$. By algebraic computations, we have $g(\ell-\epsilon)>0$ is equivalent to $\alpha-\beta>\frac{\ln (\frac{2(\alpha-\beta)\harmo}{(\alpha+\beta)\harmo-1}-\epsilon)-\ln \epsilon}{\ln (1+\frac{(\alpha-\beta)\harmo+1}{(\alpha+\beta)\harmo-1}-\epsilon)-\ln (1-(\frac{(\alpha-\beta)\harmo+1}{(\alpha+\beta)\harmo-1}-\epsilon))) } $. Thus, the results follow. 

\end{proof}
\section{Additional Proofs}
{
\renewcommand{\thetheorem}{\ref{cla:roundopt}}
\begin{claim}[Local Maximum $\rightarrow$ Optimal Bonus]
If there exists $\Tilde{x}\in(0,n+1)$ such that for all $1\leq x<\Tilde{x}$, $\E[\Delta_\belset(x + 1)]\geq \E[\Delta_\belset(x - 1)]$ and when $\Tilde{x}\leq n-1$, for all $\Tilde{x}\leq x\leq n-1$, $\E[\Delta_\belset(x + 1)]\leq \E[\Delta_\belset(x - 1)]$, then $\round(\Tilde{x})$ is the optimal bonus.
\end{claim}
\addtocounter{theorem}{-1}
}

\begin{proof}[Proof of Claim~\ref{cla:roundopt}]
We will prove that $\E[\Delta_\belset(\round(\Tilde{x}))]$ is the maximum, i.e, for all $x\in\mathcal{X}(n)$, $\E[\Delta_\belset(\round(x))]\leq \E[\Delta_\belset(\round(\Tilde{x}))]$. 

First, we prove that 1) when $\round(\Tilde{x})-2\geq 0$, $\E[\Delta_\belset(\round(\Tilde{x}))]$ is larger than or equal to $\E[\Delta_\belset(\round(\Tilde{x})-2)]$; 2) when $\round(\Tilde{x})+2\leq n$, $\E[\Delta_\belset(\round(\Tilde{x}))]$ is larger than or equal to $\E[\Delta_\belset(\round(\Tilde{x})+2)]$.

Note that according to the definition of $\round$, we have $\round(\Tilde{x})\in[\Tilde{x}-1,\Tilde{x}+1)$. Therefore, when $\round(\Tilde{x})-2\geq 0$, $\round(\Tilde{x})-1\in[1,\Tilde{x})$. Then we have $\E[\Delta_\belset(\round(\Tilde{x})-1-1)]\leq \E[\Delta_\belset(\round(\Tilde{x})-1 + 1)]$ which implies that $\E[\Delta_\belset(\round(\Tilde{x})-2)]\leq\E[\Delta_\belset(\round(\Tilde{x}))]$. Similarly, when $\round(\Tilde{x})+2\leq n$,  $\round(\Tilde{x})+1\in[\Tilde{x},n-1]$, and we have $\E[\Delta_\belset(\round(\Tilde{x}))]\geq \E[\Delta_\belset(\round(\Tilde{x})+2)]$.

Second, we find that the properties of $\Tilde{x}$ also imply that for other $x$, $\E[\Delta_\belset(x)]$ is less than or equal to either $\E[\Delta_\belset(\round(\Tilde{x})-2)]$ or $\E[\Delta_\belset(\round(\Tilde{x})+2)]$.

Finally, when $\Tilde{x}>n-1$, for all $x<\Tilde{x}$, $\E[\Delta_\belset(x + 1)]\geq \E[\Delta_\belset(x - 1)]$, the surprise is increasing with $x$ and $\round(\Tilde{x})=n$ is the optimal bonus.

Therefore, $\E[\Delta_\belset(\round(\Tilde{x}))]$ is the maximum and $\round(\Tilde{x})$ is the optimal bonus.
\end{proof}

{
\renewcommand{\thetheorem}{\ref{ob:symmetric}}
\begin{observation}
In the symmetric case, i.e. $\alpha=\beta$, $Q^i_j$ is also symmetric, that is $Q^i_j=Q^i_{i-j}$.
\end{observation}
\addtocounter{theorem}{-1}
}
\begin{proof}[Proof of Observation~\ref{ob:symmetric}]
Recall Lemma~\ref{lem:generald}, we calculate the value of $Q^i_j$
\[
Q^i_j=\frac{2(\alpha+\beta-1)\binom{i}{j}\binom{\alpha+\beta-2}{\alpha-1}}{(\alpha+\beta+i)\binom{i+\alpha+\beta}{j+\alpha}}=\frac{2(2\alpha-1)\binom{i}{j}\binom{2\alpha-2}{\alpha-1}}{(2\alpha+i)\binom{i+2\alpha}{j+\alpha}}
\]
Then we calculate the value of $Q^i_{i-j}$
\[
Q^i_{i-j}=\frac{2(\alpha+\beta-1)\binom{i}{i-j}\binom{\alpha+\beta-2}{\alpha-1}}{(\alpha+\beta+i)\binom{i+\alpha+\beta}{i-j+\alpha}}=\frac{2(2\alpha-1)\binom{i}{i-j}\binom{2\alpha-2}{\alpha-1}}{(2\alpha+i)\binom{i+2\alpha}{i-j+\alpha}}=\frac{2(2\alpha-1)\binom{i}{j}\binom{2\alpha-2}{\alpha-1}}{(2\alpha+i)\binom{i+2\alpha}{j+\alpha}}
\]
Thus $Q^i_j=Q^i_{i-j}$
\end{proof}
\end{document}